\documentclass[a4paper,11pt]{article}
\usepackage{amsfonts,amssymb,amsmath,mathrsfs,graphicx}
\usepackage{booktabs}
\usepackage{tabularx}
\usepackage{array}
\usepackage{titlesec}
\usepackage{amsthm}
\usepackage{amsmath}
\usepackage{amssymb}
\usepackage{color}
\usepackage{mathtools}
\usepackage{geometry}
\usepackage[normalem]{ulem}
\allowdisplaybreaks

\usepackage{algorithm2e}

\usepackage{verbatim}
\numberwithin{equation}{section}

\newtheorem{theorem}{Theorem}[section]

\newtheorem{lemma}[theorem]{Lemma}

\theoremstyle{definition}

\newtheorem{remark}[theorem]{Remark}
\newcommand{\bfs}[1]{{\boldsymbol #1}}
\newcommand{\bfb}{\boldsymbol{b}}

\newcommand{\bfss}{\boldsymbol{s}}
\newcommand{\bfx}{\boldsymbol{x}}
\newcommand{\bfz}{\boldsymbol{z}}
\newcommand{\bfy}{\boldsymbol{y}}
\newcommand{\bfzr}{\boldsymbol{0}}
\newcommand{\bfv}{\boldsymbol{v}}
\newcommand{\bfn}{\boldsymbol{n}}
\newcommand{\bft}{\boldsymbol{t}}
\newcommand{\bfu}{\boldsymbol{u}_o}
\newcommand{\bfuu}{\boldsymbol{u}}
\newcommand{\bff}{\boldsymbol{f}}
\newcommand{\bfF}{\boldsymbol{F}}
\newcommand{\bfM}{\boldsymbol{M}}
\newcommand{\bfw}{\boldsymbol{w}}

\newcommand{\bfc}{\boldsymbol{c}}

\newcommand{\bbr}{\mathbb R}

\usepackage{array}
\newcolumntype{C}[1]{>{\centering\arraybackslash}m{#1}}
\usepackage{graphics}
\usepackage{amssymb, latexsym, amsmath, amsfonts, epsfig}
\usepackage{bm}
\usepackage{graphicx}
\usepackage{color}
\usepackage{epstopdf}
\usepackage{comment}
\usepackage{soul}
\usepackage[normalem]{ulem}
\usepackage{float}
\usepackage[sort&compress,numbers]{natbib}

\def\d{{  \rm d}}

\newcommand{\T}{\mathcal{T}}

\newcommand{\E}{\mathcal{E}}

\usepackage{caption}
\date{\vspace{-6ex}}

\begin{document}

\newcommand{\Question}[1]{{\marginpar{\color{blue}\footnotesize #1}}}
\newcommand{\blue}[1]{{\color{blue}#1}}
\newcommand{\red}[1]{{\color{red} #1}}

\newif \ifNUM \NUMtrue

\title{Particle, kinetic and hydrodynamic models for sea ice floes. Part II: Rotating floes with environmental forces} 

\author{Quanling Deng\thanks{Yau Mathematical Sciences Center, Tsinghua University, Beijing, 100084, China (qldeng@tsinghua.edu.cn); School of Computing, Australian National University, Canberra, ACT 2601, Australia (quanling.deng@anu.edu.au).}
\and Seung-Yeal Ha\thanks{Department of Mathematical Sciences and Research Institute of Mathematics, Seoul National University, Seoul 08826, Republic of Korea (syha@snu.ac.kr).}
\and Jaemoon Lee\thanks{Department of Mathematical Sciences, Seoul National University, Seoul 08826, Republic of Korea (dlwoans0001@snu.ac.kr).}
\and Alberto Alberello\thanks{School of Engineering, Mathematics and Physics, University of East Anglia, NR4 7TJ, Norwich, UK (alberto.alberello@outlook.com).} 
\and Luke G. Bennetts\thanks{School of Mathematics and Statistics, The University of Melbourne, Parkville, VIC 3010, Australia (luke.bennetts@unimelb.edu.au).}
}

\maketitle

\begin{abstract}
This paper builds on the multiscale modeling framework introduced in Part I (Deng and Ha, \textit{Physica D: Nonlinear Phenomena 483 (2025) 134951}) for sea-ice floe dynamics with non-rotating floes to the case with rotational floes, nonlinear contact interactions, Coriolis force and ocean tilt. Building on the particle–kinetic–hydrodynamic hierarchy developed for non-rotating floes, we generalize the particle model to describe ice floes as rigid bodies characterized by position, linear velocity, angular velocity, size, and moment of inertia. The interaction rules now include nonlinear contact forces and torques arising from short-range compression, restitution, and tangential friction laws, together with both oceanic and atmospheric drags that couple translational and rotational motions. These particle descriptions lead to an enriched Vlasov-type kinetic equation posed on an extended phase space, whose moments yield a hydrodynamic system for mass, momentum, and angular-momentum balances. Compared with Part I, the resulting macroscopic equations feature additional hydrodynamic and stress contributions, rotational transport, and dissipative mechanisms stemming from nonlinear collisions. The proposed framework provides a more accurate description of sea-ice floe dynamics and offers a systematic pathway toward multiscale modeling of sea-ice rheology under complex environmental forcing.
%
\end{abstract}
%
%


\section{Introduction} 
Sea ice plays a fundamental role in Earth's polar climate system by regulating heat exchange, momentum transfer, and biogeochemical processes at the ocean--atmosphere interface \cite{lund2020arctic,bennetts2024closing,massom2025influence}. Its dynamics emerge from the complex interplay between environmental forcing, including wind, ocean currents, Coriolis, and waves, and mechanical interactions among heterogeneous floes (cf., \cite{parkinson1979large, feltham2008sea,blockley2020future, hunke2011multiphase, roach2018emergent, meylan2021floe, alberello2020drift, alberello2022three}). In the Marginal Ice Zone (MIZ), where the ice cover is highly fragmented, individual floe-floe interactions play a vital role in the mechanical behavior of the ice pack, giving rise to phenomena that are often described by particle floes (cf., \cite{bennetts2022marginal,bennetts2022theory}).

Continuum-based models, such as elastic-plastic  \cite{coon1974modeling}, viscous--plastic \cite{hibler1979dynamic,shen1987role}, elastic--viscous--plastic \cite{hunke1997elastic}, and Maxwell-type rheological models \cite{dansereau2016maxwell}, have been developed for simulating the large-scale behavior of compact ice packs (for ice sheets and glaciers, we refer to, e.g., \cite{nye1959motion}). However, the continuum assumption can break down in both large and small scales in the MIZ or other areas \cite{weiss2017linking}, where collisions, heterogeneity, and anisotropic deformation fields produce granular-like behavior that cannot be adequately represented by homogenized constitutive laws \cite{marquart2025wice}. These limitations have motivated the development of particle-based (discrete element) approaches \cite{hopkins2004discrete, lindsay2004new, wilchinsky2006modelling, de2024modelling, damsgaard2018application, manucharyan2022subzero, bateson2025simulating}, which model sea ice as a collection of rigid bodies and capture processes such as collisions, ridging, rafting, and fracture patterns.
In earlier model developments, Gudkovich \textit{et al.} \cite{gudkovich1963} studied the dynamics of an individual floe in 1963, based on the observations of Nansen \cite{nansen1902norwegian} in 1902, and developed a model to describe the speed and direction of drift that depended on floe size and shape. 
We refer to \cite{hunke2010sea,blockley2020future,golden2020modeling} for further discussions on sea ice modelling.

One of the challenges in particle-based modeling lies in accurately representing contact interactions \cite{herman2016discrete,damsgaard2018application}. Nonlinear contact mechanics, including Hertzian-type normal forces, velocity-dependent restitution, tangential friction laws, and torque transfer at contact points, play a critical role in determining floe trajectories and rotational dynamics. Such nonlinear collision models, widely used in granular and rigid-body simulations~\cite{johnson1987contact, cundall1979discrete}, have recently been applied to sea ice floes \cite{herman2016discrete}. Rotational degrees of freedom further enrich the dynamics: collisions generate torques; frictional impulses modify angular momentum; and ocean drag acts simultaneously on translational and rotational motions. These effects contribute to complex behaviors such as spin alignment, rotational clustering, collision-induced jamming, and enhanced dissipation, all observed in high-resolution DEM studies~\cite{herman2016discrete}.

In Part~I of this series~\cite{deng2025particle}, the first two authors developed a particle--kinetic--hydrodynamic hierarchy for non-rotating cylindrical floes governed by linear contact forces. The formulation established a systematic pathway from Newtonian-type particle models to Vlasov-type kinetic descriptions and further to hydrodynamic equations based on velocity moments. While the non-rotating framework provides conceptual clarity, it neglects nonlinear contact mechanics, rotational effects, and complex environmental forcings such as oceanic and atmospheric drags that are essential for sea-ice dynamics in the MIZ.

The present paper extends the multiscale framework of Part~I by incorporating floe rotation, nonlinear contact force, atmospheric drag, Coriolis force, and ocean tilt into the model in a realistic setting enlightened by \cite{coon1974modeling, feltham2008sea, alberello2020drift}. This provides a pathway to couple particle and continuum models within a multiscale framework \cite{deng2024particle}, together with their integration into data assimilation schemes \cite{deng2025lemda,chen2022superfloe,chen2022efficient}, enabling more accurate predictions. At the particle level, each floe is modeled as a rigid body characterized by its position, orientation, linear velocity, angular velocity, size, and moment of inertia. Collisions generate both forces and torques, described through physically realistic nonlinear force laws and frictional interactions. This leads to a richer kinetic formulation defined on an extended phase space that includes angular velocities and orientations. Suitable moment closures then yield hydrodynamic equations that track mass, linear momentum, and angular momentum, together with stress and couple-stress contributions arising from torque-generating interactions.
The objective of Part~II is to establish a rigorous theoretical foundation for such a rotating particle--kinetic--hydrodynamic hierarchy. Our main contributions include:
\begin{itemize}
    \item a detailed study on the behavior of total momentum and energy of a particle model for rotating floes with nonlinear contact force, atmospheric drag, Coriolis force, and ocean tilt in a realistic setting;
    \item the derivation of a kinetic model in an extended phase space with rotational dynamics, followed by a study on the momentum and energy;
    \item the development and study of a hydrodynamic system that captures macroscopic evolution of both linear and angular momentum and energies;
    \item numerical simulations for the comparison of multiscale descriptions, including a realistic scenario with observed buoy trajectories and environmental (atmospheric, oceanic, Coriolis, and ocean tilt) forcing.

\end{itemize}

\begin{table}[ht]
\centering
\small
\setlength{\tabcolsep}{4pt}
\renewcommand{\arraystretch}{1.2}

\begin{tabularx}{\textwidth}{
|>{\raggedright\arraybackslash}p{0.18\textwidth}
|>{\raggedright\arraybackslash}p{0.30\textwidth}
|>{\raggedright\arraybackslash}X|
}
\hline
Feature & Part I & Present Part II \\
\hline

Floe motion 
& Translational motion only 
& Translational and rotational motion \\
\hline

Phase variables 
& $x,v,r,h$ 
& $x,v,\theta,\omega,r,h$ \\
\hline

Contact model 
& Linear  
& Nonlinear with tangential friction \\
\hline

Torque 
& Not included 
& Contact/Drag-induced torque and rotation \\
\hline

Wind drag
& Not included 
& Quadratic wind drag \\
\hline

Coriolis force
& Not included 
& Included \\
\hline

Ocean tilt
& Not included 
& Included \\
\hline

Balance laws 
& Mass, momentum, and energy 
& Includes rotational parts \\
\hline

Kinetic limit 
& Translational Vlasov--McKean equation 
& Vlasov--McKean equation with rotational dynamics \\
\hline

Hydrodynamic closure 
& Non-rotating closure 
& Closure with angular-momentum transport \\
\hline
\end{tabularx}

\caption{Comparison between Part I and the present work.}
\label{tab:1}
\end{table}

The main differences between Parts I and II are summarized in Table~\ref{tab:1}. 
The main distinction lies in the incorporation of rotational dynamics and environmental forcing into the particle model, which leads to additional rotational contributions in the balance laws, as well as in the kinetic and hydrodynamic limits. 
For the theory development, the primary challenge arises from the rotational mechanisms, which introduce additional nonlinear coupling, functional structures, and energy–momentum terms that are significantly more difficult to analyze, particularly in establishing rigorous energy dissipation properties and related estimates.

\color{black}

The rest of the paper is organized as follows.  
In Section \ref{sec:pmodel}, we introduce the rotating sea-ice floe particle model with nonlinear contact laws and environmental forces in Subsection \ref{sec:dem}, followed by a study on the asymptotic behavior of total momentum and energy in Subsection \ref{sec:pmodelab}. 
In Section~\ref{sec:kmodel} and Section~\ref{sec:hmodel}, we present the corresponding kinetic description using the mean-field limit, followed by a hydrodynamic model with monokinetic closure. 
We also study the asymptotic behavior on momentum and energies of these models. 
In Section~\ref{sec:num}, we provide several numerical examples to validate our theoretical findings.
Finally, Section~\ref{sec:conclusion} is devoted to the summary of the paper with a discussion of the main findings and future directions.

\section{Particle description of ice floes} \label{sec:pmodel}

\subsection{The particle model} \label{sec:dem}

We introduce a particle-based model (often referred to as a discrete element model, DEM) for the dynamics of sea ice floes in the MIZ. 
Each floe is modeled as a rigid cylindrical particle that can translate and rotate in the horizontal plane with the following features

\begin{itemize}
    \item floe--floe interactions with nonlinear Hertz contact forces and Coulomb-limited tangential friction, allowing both translation and rotational slip during collisions. 
    
    \item the motion of each floe is coupled to the atmospheric winds via the quadratic drag law acting on the translational velocity and rotational vorticity. 
    
    \item oceanic forcing is included similarly through ocean (quadratic) drag. 
    
    \item Coriolis force due to Earth's rotation is incorporated to capture the large-scale deflection of floe motion in polar regions. 
    
    \item the model accounts for the effect of sea-surface ocean tilt through a gravity-slope force acting along the gradient of the free-surface elevation. 

\end{itemize}

We refer to Figure \ref{fig:floe} for a schematic for the floe particle model. 
\begin{figure}[ht]
    \centering \includegraphics[width=13cm]{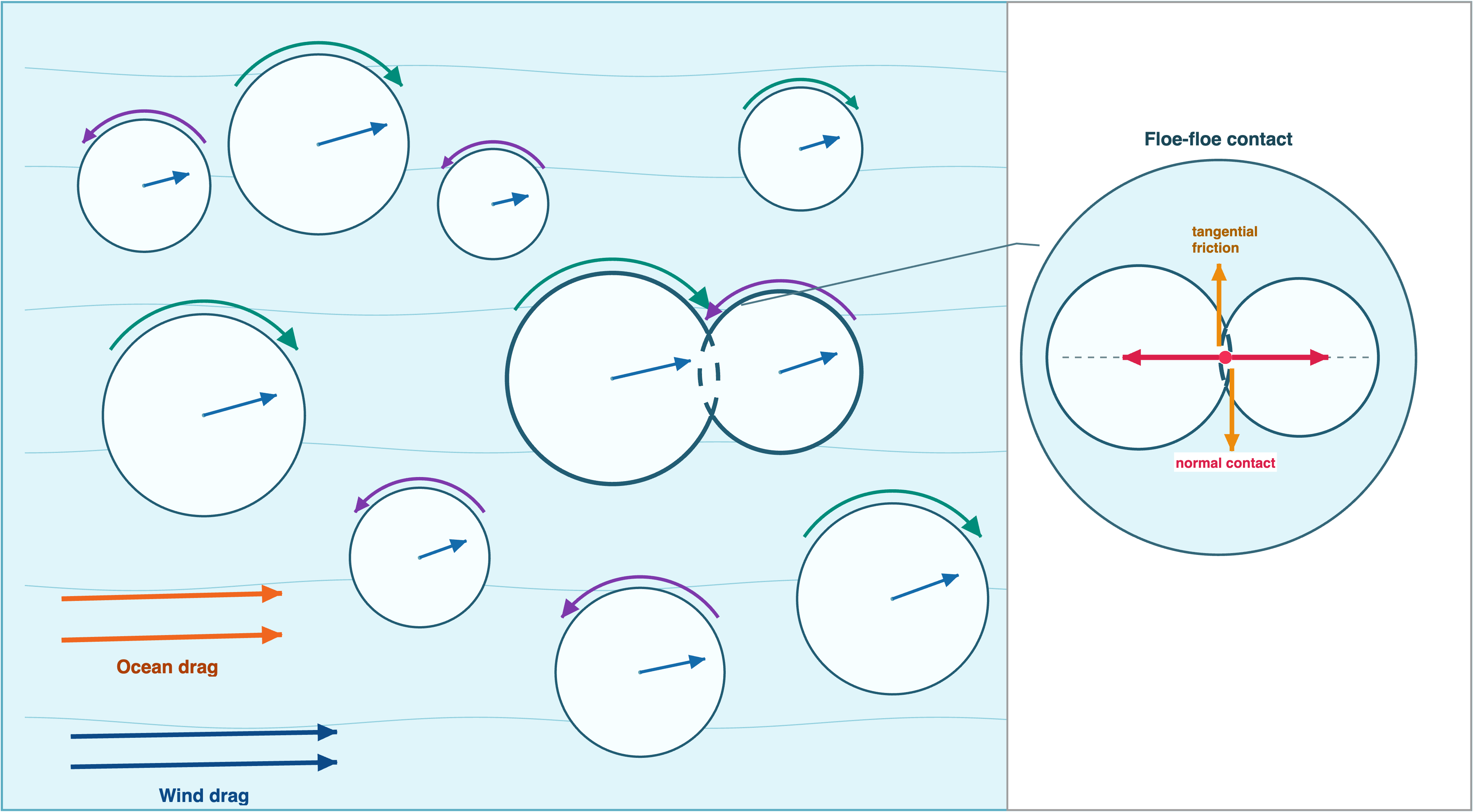} 
    \caption{A schematic for rotating and interacting floes with environmental forces.}
    \label{fig:floe}
\end{figure}

Specifically, consider colliding ice floes modeled as cylinders with motion in the $(x,y)$ plane. 
Given a system of $n$ floes, 
we denote by $r^i$ the radius and $h^i$ the thickness (or height) of the $i$-th floe with $i \in [n] := \{1,2, \cdots, n \}$. 
The radius and thickness characterize the floe size. 
In a realistic sea ice floe setting, the floe size has been observed to follow a power law distribution \cite{stern2018seasonal}, while the floe thickness distribution follows a Gamma distribution (see, for example, \cite{thorndike1975thickness, bourke1987sea, toppaladoddi2015theory} for the Arctic region and \cite{toyota2011size} for the Antarctic region). 
The mass of the $i$-th floe is $m^i=\rho_{ice}\pi (r^i)^2 h^i$, where the constant $\rho_{ice}$ is the density of sea ice floes. 
We assume that the mass of each floe does not change over time (thus, no melting, freezing, or fracturing).
We denote by $I^i$ the moment of inertia.
The floe position is denoted by $\bfx^i=(x^i,y^i)^T$ and the floe velocity is $\bfv^i=(u^i,v^i)^T$. 
We denote the floe angular displacement as $\theta^i$ and angular velocity as $ \omega^i\hat{\bfs{z}}$. Herein, $\hat{\bfs{z}}$ is the unit vector along the $z$-axis (perpendicular to the $(x,y)$ plane). 
We consider the angular velocity as a scalar in this dimension and omit the multiplication of $\hat{\bfs{z}}$ for simplicity.
Thus, if the result of the cross product is a vector along the $z$-axis, we consider it as a $z$-value scalar as in \eqref{dem_t} and \eqref{dem_w} below. 

Let $\bfu = \bfu(x,y)$ be the given ocean surface velocity.
The governing equations of rotating-colliding sea ice floe dynamics are given by Newton's equation such that \cite{damsgaard2018application,deng2024particle}
\begin{subequations}\label{eq:dem}
\begin{align}
  \frac{\d\bfx^i}{\d t} &= \bfv^i, \quad i \in [n], \label{dem_x}\\
  \frac{\d\theta^i}{\d t}  &= \omega^i,\label{dem_t}\\
  m^i\frac{\d\bfv^i}{\d t} &=   \frac{1}{n} \sum_{j=1}^{n}(\bff_{\bfn}^{ij} + \bff_{\bft}^{ij}) + \alpha^i_o\left(\bfu^i -\bfv^i\right)\left| \bfu^i -\bfv^i\right| + \alpha^i_w\left(\bfw^i -\bfv^i\right)\left| \bfw^i -\bfv^i\right| \nonumber \\
  & \quad - m^i f_E \hat{\bfs{z}}\times \bfv^i  + m^i \bfuu_g 
  =: \bfF^i_c + \bfF^i_o + \bfF^i_w + \bfF^i_E + \bfF^i_g := \bfF^i,\label{dem_v}\\
  I^i\frac{\d\omega^i}{\d t} &= \frac{1}{n} \sum_{j=1}^{n} ( r^i\bfn^{ij}\times \bff_{\bft}^{ij})  + \beta^i_o \left(\nabla\times \bfu^i/2-\omega^i \right)\left|\nabla\times \bfu^i/2-\omega^i \right| \nonumber \\
  & \quad +  \beta^i_w \left(\nabla\times \bfw^i/2-\omega^i \right)\left|\nabla\times \bfw^i/2-\omega^i \right|
  =: T^i_c + T^i_o +T^i_w =: T^i,\label{dem_w}
\end{align}  
\end{subequations}
where $|\cdot|$ is the $\ell^2$-norm and
\begin{subequations}\label{eq:demparm}
\begin{align} 
    \bff_{\bfn}^{ij} & = \big( \kappa^{ij}_1 \delta^{ij} + \kappa^{ij}_2 (\bfv^i  - \bfv^j )\cdot\bfn^{ij} \big) \bfn^{ij}, \quad
     \bff_{\bft}^{ij}  = \zeta^{ij} \tilde{\bff}_{\bft}^{ij}, \quad  \tilde{\bff}_{\bft}^{ij} = \kappa^{ij}_3 \sigma^{ij} \bft^{ij}, \label{eq:prmfnft} \\
     \kappa^{ij}_1 & = \pi E_e \chi^{ij} h^{ij}_e g\left(\frac{ \delta^{ij} r_e^{ij}}{2(h^{ij}_e)^2} \right), \quad g(\xi) = \frac{0.9117\xi^2 - 0.2722\xi + 0.003324}{\xi^2 - 1.524\xi + 0.03159} \label{eq:prmk1g} \\
     \kappa^{ij}_2 & = \eta_0 \sqrt{ 5 \kappa^{ij}_1 m^{ij}_e}, \quad \kappa^{ij}_3 = 6 \frac{G_e}{E_e}\kappa^{ij}_1,  \label{eq:prmk2k3} \\
     \delta^{ij} & = d^{ij} - (r^i+r^j), 
    \quad\mbox{with}\quad d^{ij} = |\bfx^i-\bfx^j|, \label{eq:prmdelt} \\
     \sigma^{ij}_t & = \frac{\d \sigma^{ij}}{\d t} = (\bfv^j - \bfv^i)\cdot\bft^{ij} - r^j\omega^j - r^i\omega^i, \label{eq:prmsigt} \\ 
     \sigma^{ij} & = \int_{t_\text{start}}^t \sigma^{ij}_t(s) \d s  = t_c^{ij} \sigma^{ij}_t, \label{eq:prmsig} \\
     t_c^{ij} & \approx 2.94 \Big(\frac{m_e^{ij}}{\kappa_1^{ij}} \Big)^{2/5} |\bfv_i - \bfv_j|^{-1/5}.  \label{eq:prmtc}
\end{align}
\end{subequations}
For simplicity, we denote the environmental force and torque for floe $i$ at time $t$ as
\begin{equation}
    \label{eq:envft}
    \bfF^i_e = \bfF^i_o + \bfF^i_w + \bfF^i_E + \bfF^i_g, \qquad T^i_e = T^i_o +T^i_w.
\end{equation}
Herein, for a rigid circular floe with torque acting on the center of mass, the torques induced by the Coriolis force and ocean tilt are zero. 
The parameter 
$f_E$ denotes the Coriolis force associated with Earth's rotation, defined by
\[
f_E = 2\Omega \sin\phi,
\]
where $\Omega$ is the angular rotational velocity of the Earth (\( \approx 7.2921 \times 10^{-5}\) radians per second) and $\phi$ is the latitude depending on position $\bfx$. 
The Coriolis force $m^i f_E \hat{\bfs{z}}\times \bfv^i$ acts perpendicular to the floe velocity and accounts for the deflection of sea ice motion in the rotating reference frame of the Earth, which becomes important for large-scale sea ice drift in polar regions \cite{cushman2011introduction,alberello2020drift}.
The term $m^i \bfuu_g$ in \eqref{dem_v} describes the horizontal gravitational pull induced by a tilted ocean surface, where 
\[
\bfuu_g = U_g f_E (-\sin(f_E t), \cos(f_E t))^T
\]
with $U_g = 0.125$ the amplitude of the near-inertial oscillations estimated from the mean amplitude of the oscillations in 12–14 hr identified utilizing a band-pass filter \cite{alberello2020drift, cushman2011introduction, lepparanta2005drift}.
The material properties of sea ice are
\begin{equation} \label{eq:EG}
E_e = \frac{E}{2(1 - \nu^2)}, \quad \text{and} \quad G_e = \frac{E}{4(2+\nu)(1-\nu)},
\end{equation}
where $E$ is Young's modulus, and $\nu$ is Poisson's ratio (identical for all floes).

In the integration for $\sigma^{ij}$ defined in \eqref{eq:prmsig}, $t_\text{start}$ specifies the starting time of the collision, and $t_c^{ij}$ in \eqref{eq:prmtc} is an approximate contact duration time that compensates the integration \cite{gugan2000inelastic}.
$\chi$ in \eqref{eq:prmk1g} is a characteristic function
\begin{equation} \label{eq:chi}
\chi^{ij} = 
\begin{cases}
1, \quad \text{when} \ i\ne j\text{ and }\delta^{ij}<0, \\
0, \quad \text{otherwise},
\end{cases}
\end{equation}
which characterizes whether floes $i$ and $j$ are in contact or not. 
We multiply this by $\kappa_1^{ij}$ for notational simplicity. 
$\bff_{\bfn}^{ij}$ and $\bff_{\bft}^{ij}$ in \eqref{eq:prmfnft} are normal and tangential components of the contact force, respectively.
For notational convenience, we denote
$
\bff_c^{ij} = \bff_{\bfn}^{ij} + \bff_{\bft}^{ij}
$
with the auxiliary parameter $\zeta$ defined as
\begin{equation} \label{eq:zeta}
\zeta^{ij} = 
\begin{cases}
1, \quad & \text{when} \ |\tilde{\bff}_{\bft}^{ij}| \leq \mu |\bff_{\bfn}^{ij}|, \\
\frac{\mu |\bff_{\bfn}^{ij}|}{|\tilde{\bff}_{\bft}^{ij}| }, \quad & \text{otherwise}.
\end{cases}
\end{equation}
This imposes the Coulomb friction law that plays an important role in limiting the tangential contact force relative to the magnitude of the normal contact force \cite{hopkins2004discrete} such that
\begin{equation*} 
|\bff_{\bft}^{ij}| \leq \mu|\bff_{\bfn}^{ij}|,
\end{equation*}
 where $\mu$ is the coefficient of friction that characterizes the condition of the surfaces of two floes in contact.

We assume zero tangential damping as in \cite{cundall1979discrete,damsgaard2018application}. The floe contact force $\bfF^i$ and torque $T^i$ are non-zero only when two floes are in contact, i.e., $\delta^{ij}<0$ (see \eqref{eq:prmdelt}).
We also remark that for $\delta^{ij}<0$, we have $\xi<0$ and  $g(\xi)>0$. This guarantees the correct sign of $\kappa^{ij}_1 \delta^{ij}$ in the term $\bff_{\bfn}^{ij}$. 
We follow the Hertz contact theory \cite{hertz1882ueber,puttock1969elastic} and adopt the model for the normal contact force (see the supplementary material of \cite{herman2016discrete}).
 We assume $\delta^{ij}=0$ when $i=j$ (hence, $\bff_{\bfn}^{ij} = \bff_{\bft}^{ij}= \bf0$; by default, $i\ne j, i,j \in [n]$) for notational simplicity. 
$\sigma_t^{ij}$ in \eqref{eq:prmsigt} is the slip-rate along the tangential direction caused by both translation and rotation (relative contact velocity projected in the tangential direction). 
$\sigma$ is the rotational deformation, i.e., the tangential shear deformation. 
$\rho_o$ is ocean density. 
$\alpha^i_s, \beta^i_s,$ where $ s = o, w$ are drag coefficients. 

The vector 
\begin{align} \label{eq:n}
    \bfn^{ij} =  
    \begin{cases}
        \frac{\bfx^j-\bfx^i}{|\bfx^j-\bfx^i|}, & i\ne j,\\
        \bfzr, &  i=j
    \end{cases}
\end{align}
 is the unit normal vector pointing from the center of floe $i$ to floe $j$, 
$\bft^{ij}$ is the unit vector along the tangential direction (rotating $\bfn^{ij}$ counterclockwise by $90^\circ$),  $r^j$ is the radius of the $j$-th floe,
$E_e$ is the effective contact modulus, $h^{ij}_e:= \min\{h^i, h^j\}$ is the effective contact thickness (part of the thickness in contact), $m^{ij}_e := \frac{m^i m^j}{m^i + m^j}$ is the effective mass, $r^{ij}_e := \frac{r^i r^j}{r^i + r^j}$ is the effective radius, and in \eqref{eq:prmk2k3}
\begin{equation} \label{eq:beta}
\eta_0 := \frac{\ln e_r}{\sqrt{\ln^2 e_r + \pi^2} } < 0.
\end{equation}
Here $e_r$ is the restitution coefficient in $(0, 1)$. The restitution coefficient for ice floes is a measure of kinetic energy loss during floe collisions. The coefficient is usually a positive real number between 0 and 1. A value of 0 indicates a perfectly inelastic collision, while a value of 1 indicates a perfectly elastic collision. The typical values of $e_r$  for sea ice are between 0.1 and 0.3 \cite{li2020laboratory}. In sea ice studies, values between 0 and 1 are often employed; see, for example \cite{herman2019wave}. 
 Details are referred to equations (1), (8), (20), and (26) in the supplementary material of \cite{herman2016discrete}.
We scale the summation of the contact force by the factor of $1/n$ to apply the mean-field theory for deriving the kinetic and hydrodynamic models.
For a system with a fixed number of floes, this scaling factor may be understood as a factor of the mass and drag coefficient.

\color{black}

\subsection{Asymptotic behavior} \label{sec:pmodelab}

To study the asymptotic behavior of the particle model  \eqref{eq:dem}, we first define the strain and kinetic energies. 
The normal strain energy for two colliding floes is defined as 
\begin{equation}\label{eq:nse}
    M_{2,\bfx}^{ij} := \int_0^{\delta^{ij}} \bff_{\bfn,s}^{ij} \d s = \int_0^{\delta^{ij}}  \kappa^{ij}_1(s) s \d s.
\end{equation} 
 %
 This normal strain energy evolves with respect to time. Using the Leibniz integral rule, we have
  \begin{equation} \label{eq:nsedt}
 \frac{\d M_{2,\bfx}^{ij}}{\d t} 
 = \frac{\d}{\d t} \int_0^\delta   \kappa^{ij}_1(s) s \d s  =   \kappa^{ij}_1(\delta^{ij}) \delta^{ij} \frac{\d \delta^{ij}}{\d t}.
 \end{equation}


\noindent
We define the floe moments as
\begin{equation} \label{eq:Ms}
\begin{aligned}
    M_0 & = \sum_{i=1}^n m^i, \qquad \bfM_{1,\bfv} = \sum_{i=1}^n m^i \bfv^i, 
    \qquad M_{1,\omega} = \sum_{i=1}^n (m^i \bfx^i \times \bfv^i + I^i \omega^i), \\
    M_2 & = M_{2,\bfv} + M_{2,\bfx} + M_{2,\omega}, \\
    M_{2,\bfv} &= \frac{1}{2} \sum_{i=1}^n m^i | \bfv^i|^2, \qquad M_{2,\bfx} = \frac{1}{2n}\sum_{i,j=1}^n M_{2,\bfx}^{ij}, \qquad    M_{2,\omega}  = \frac{1}{2} \sum_{i=1}^n I^i (\omega^i)^2,
    \end{aligned}
\end{equation}
where $M_0$ is the zero-order moment, i.e., total mass. 
$\bfM_{1,\bfv} $ is the total momentum.
$M_{1,\omega}$ is the total angular momentum, which includes the orbital angular momentum (due to translational motion) and the spin angular momentum (due to rotation about the center of mass). 
$M_{2,\bfx}$,  $M_{2,\bfv},\ M_{2,\omega}$, and  $M_2$ represent the total normal strain energy, the total translational kinetic energy, the total rotational kinetic energy, and the total energy, respectively. 
The scaling $\frac{1}{n}$ in $M_{2,\bfx}$ ensures that total strain energy remain $\mathcal{O}(1)$ in the mean-field limit as $n\to \infty.$
Note that the energies are positive by definition.

\begin{lemma}[Total momentum balances] \label{lem:sumparticleMV}
Let $(\bfx^i, \bfv^i, \theta^i, \omega^i), i\in [n],$ be a global solution to the system \eqref{eq:dem}. The following assertions hold.
\begin{enumerate}
\item The total linear momentum satisfies
\begin{align}
\begin{aligned} \label{eq:dtM1v}
\frac{\d\bfM_{1,\bfv}}{\d t}  = \sum_{i=1}^{n} \bfF^i_e.
\end{aligned}
\end{align}

\item The total angular momentum satisfies
\begin{align}
\begin{aligned} \label{eq:dtM1w}
\frac{\d M_{1,\omega}}{\d t}  = \sum_{i=1}^{n} \big[ \bfx^i \times \bfF^i_e
      + T^i_e \big].
\end{aligned}
\end{align}
\end{enumerate}
\end{lemma}

\begin{proof}
For the first assertion, we note the relations 
\[ \delta^{ij} = \delta^{ji}, \quad  \bfn^{ij} = -\bfn^{ji}, \quad  \bft^{ij} = -\bft^{ji}, \quad i, j \in [n],
\]
to see the anti-symmetry of the floe-floe contact forces
\begin{equation} \label{eq:cfnt}
\bff_{\bfn}^{ij} = -\bff_{\bfn}^{ji}, \quad \bff_{\bft}^{ij} = -\bff_{\bft}^{ji}, \quad i, j \in [n].
 \end{equation}
Now, we use \eqref{eq:dem}, \eqref{eq:envft}, \eqref{eq:Ms}, and  \eqref{eq:cfnt} to find 
\begin{align*}
\begin{aligned}
    \frac{\d \bfM_{1,\bfv}}{\d t} & = \frac{\d}{\d t} \sum_{i=1}^n m^i \bfv^i  = \sum_{i=1}^n m^i \frac{\d \bfv^i}{\d t}  =  \sum_{i=1}^n \Big( \frac{1}{n} \sum_{j=1}^{n} (\bff_{\bfn}^{ij} + \bff_{\bft}^{ij}) + \bfF^i_e \Big) \\
    &= \frac{1}{n} \sum_{i, j=1}^{n} (\bff_{\bfn}^{ij} + \bff_{\bft}^{ij})  + \sum_{i=1}^{n} \bfF^i_e =\sum_{i=1}^{n} \bfF^i_e. 
\end{aligned}
\end{align*}

\noindent
For the second assertion, using \eqref{eq:dem} and $\bfv^i \times \bfv^i=0$, we calculate
\begin{align*}
\begin{aligned}
    \frac{\d M_{1,\omega}}{\d t} & = \frac{\d}{\d t} \sum_{i=1}^n (m^i \bfx^i \times \bfv^i + I^i \omega^i) \\
    & = \sum_{i=1}^n \Big(m^i \frac{\d \bfx^i}{\d t}  \times \bfv^i + m^i \bfx^i \times \frac{\d \bfv^i}{\d t} + I^i \frac{\d \omega^i}{\d t} \Big) \\
    & = \sum_{i=1}^n \Big( \bfx^i \times m^i \frac{\d \bfv^i}{\d t} + I^i \frac{\d \omega^i}{\d t} \Big) \\
    & = \sum_{i=1}^n \Big( \bfx^i \times \big(\frac{1}{n} \sum_{j=1}^{n}(\bff_{\bfn}^{ij} + \bff_{\bft}^{ij}) + \bfF^i_e \big) +  \big( \frac{1}{n} \sum_{j=1}^{n} ( r^i\bfn^{ij}\times \bff_{\bft}^{ij})  + T^i_e \big) \Big) \\
    & =: \T+ \sum_{i=1}^{n} \big( \bfx^i \times \bfF^i_e
      + T^i_e  \big),
\end{aligned}
\end{align*}
where 
\begin{equation*}
\T  = \frac{1}{n}  \sum_{i,j=1}^n \big( \bfx^i \times (\bff_{\bfn}^{ij} + \bff_{\bft}^{ij}) + ( r^i\bfn^{ij}\times \bff_{\bft}^{ij}) \big).
\end{equation*}
We denote the contact point of floes $i$ and $j$ as $\bfc^{ij}$. Then we have 
\begin{equation}\label{eq:cij}
    \bfc^{ij} = \bfx^i + r^i \bfn^{ij}, \qquad \bfc^{ji} = \bfx^j + r^j \bfn^{ji}, \qquad \bfc^{ij} = \bfc^{ji}.
\end{equation}

Using \eqref{eq:cij}, anti-symmetry, and $(\bfx^i - \bfx^j) \times \bfn^{ij}= 0$, we calculate 
\begin{align}\label{eq:Ta}
\begin{aligned}
\T & = \frac{1}{n}  \sum_{i,j=1}^n \bfx^i \times \bff_{\bfn}^{ij} + \frac{1}{n}  \sum_{i,j=1}^n \big( \bfx^i \times \bff_{\bft}^{ij} + ( r^i\bfn^{ij}\times \bff_{\bft}^{ij}) \big) \\
& = - \frac{1}{n}  \sum_{i,j=1}^n \bfx^j \times \bff_{\bfn}^{ij} + \frac{1}{n}  \sum_{i,j=1}^n \bfc^{ij}\times \bff_{\bft}^{ij} \\
& = \frac{1}{2n}  \sum_{i,j=1}^n (\bfx^i - \bfx^j) \times \bff_{\bfn}^{ij} + \frac{1}{2n}  \sum_{i,j=1}^n \big( \bfc^{ij}\times \bff_{\bft}^{ij} + \bfc^{ji}\times \bff_{\bft}^{ji} \big) \\
& = 0.
\end{aligned}
\end{align}
This leads to the desired result for the balance law of total angular momentum.
\end{proof}
\begin{remark}
In the derivation of \eqref{eq:Ta}, we used the contact point assumption/equality \eqref{eq:cij}, which implies $\delta^{ij} =0$. This distinguishes it from the contact deformation assumption for computing the normal contact force. 
That is, we assume that the floe particles deform upon contact, and the overlap (deformation) \( \delta^{ij} \) is used to compute the contact forces \( \bm{f}_{\bm{n}}^{ij} \) and \( \bm{f}_{\bm{t}}^{ij} \). However, when computing torques, we assume that the contact occurs at a point located \emph{on the undeformed geometry}, i.e., at a distance \( r^i \bm{n}^{ij} \) from the center.
This introduces a modeling approximation, which is justified by \cite{poschel2005computational,luding2008introduction}:
\begin{enumerate}
  \item \textbf{Deformations are small:} We assume that overlaps are very small compared to particle size, so undeformed geometry introduces only second-order errors. 
  \item \textbf{Approximate consistency:} The total torque from each force pair cancels under Newton's third law, ensuring conservation of angular momentum.
\end{enumerate}
\end{remark}

\begin{lemma}[Total energy balance] \label{lem:sumparticleE}
Let $(\bfx^i, \bfv^i, \theta^i, \omega^i), i\in [n],$ be a global solution to system \eqref{eq:dem}. The following  relation holds:
\begin{align}
\begin{aligned} \label{eq:dtM2}
\frac{\d M_2}{\d t}  & =  \frac{1}{2n} \sum_{i, j=1}^n  \kappa^{ij}_2 |(\bfv^i - \bfv^j )\cdot\bfn^{ij} |^2 
- \frac{1}{2n}  \sum_{i, j=1}^n  \zeta^{ij} \kappa^{ij}_3 t_c^{ij} (\sigma^{ij}_t)^2 \\
& \quad + \sum_{i=1}^{n} \bfF^i_e \cdot \bfv^i + \sum_{i=1}^{n} T^i_e \omega^i.
\end{aligned}
\end{align}
\end{lemma}

\begin{proof}
We take the product of \eqref{dem_v} with $\bfv^i$, and sum up the resulting equations over $i \in [n]$ to obtain
\begin{align}
    \frac{\d M_{2,\bfv}}{\d t} & = \sum_{i=1}^n m^i\frac{\d\bfv^i}{\d t} \cdot \bfv^i =  \sum_{i=1}^n \Big( \frac{1}{n}\sum_{j=1}^{n} (\bff_{\bfn}^{ij} + \bff_{\bft}^{ij})  + \bfF^i_e \Big) \cdot \bfv^i \nonumber  \\
    & = \frac{1}{n}  \sum_{i, j=1}^n    \big( \kappa^{ij}_1 \delta^{ij} + \kappa_2^{ij} (\bfv^i  - \bfv^j )\cdot\bfn^{ij} \big) \bfn^{ij} \cdot \bfv^i  +  \frac{1}{n}  \sum_{i, j=1}^n \bff_{\bft}^{ij} \cdot \bfv^i + \sum_{i=1}^{n} \bfF^i_e \cdot \bfv^i  \nonumber \\
    & =:  \T_{11} + \T_{12} +\T_{13} +   \sum_{i=1}^{n}\bfF^i_e \cdot \bfv^i.  \label{eq:dtM2v}
\end{align}
Below, we estimate the terms $\T_{1i},~i=1,2,$ one by one. For $\T_{13}$, we estimate it combined with the angular velocity component afterwards. 

\noindent $\bullet$~Case A.1:~We use \eqref{dem_x} and the relations
\[ \delta^{ij} = \delta^{ji}, \quad  \bfn^{ij} = -\bfn^{ji}  \]
to rewrite 
\begin{align}
\begin{aligned} \label{eq:T11}
    \T_{11} & =  \frac{1}{n} \sum_{i, j=1}^n \kappa^{ij}_1 \delta^{ij}\bfn^{ij} \cdot \frac{\d\bfx^i}{\d t}  =  -\frac{1}{n} \sum_{i, j=1}^n \kappa^{ij}_1 \delta^{ij}\bfn^{ij} \cdot \frac{\d\bfx^j}{\d t}\\
    & = - \frac{1}{2n} \sum_{i, j=1}^n \kappa^{ij}_1 \delta^{ij}  \bfn^{ij} \cdot \frac{\d (\bfx^j - \bfx^i)}{\d t}  \\
    & =  - \frac{1}{2n} \sum_{i, j=1}^n \kappa^{ij}_1 \delta^{ij} \cdot \frac{\d \delta^{ij}}{\d t} = - \frac{\d M_{2,\bfx}}{\d t},   
\end{aligned}
\end{align}
where we used \eqref{eq:nsedt} and the identity
\begin{align*}
\begin{aligned}
    \frac{\d \delta^{ij}}{\d t} = \frac{\d}{\d t} \big(|\bfx^i-\bfx^j|- (r^i+r^j)\big)  = \frac{\d}{\d t}  \sqrt{(x^j - x^i)^2 + (y^j - y^i)^2} = \bfn^{ij} \cdot \frac{\d (\bfx^j - \bfx^i)}{\d t}.   
\end{aligned}
\end{align*}

\noindent $\bullet$~Case A.2:  Similarly, we have
\begin{align}
\begin{aligned} \label{eq:T12}
    \T_{12} &  =  \frac{1}{n} \sum_{i, j=1}^n \kappa^{ij}_2  \big( (\bfv^i - \bfv^j )\cdot\bfn^{ij}\big) \bfn^{ij}   \cdot \bfv^i = - \frac{1}{n} \sum_{i, j=1}^n \kappa^{ij}_2  \big( (\bfv^i - \bfv^j )\cdot\bfn^{ij}\big) \bfn^{ij}   \cdot \bfv^j   \\
    &  =  \frac{1}{2n} \sum_{i, j=1}^n \kappa^{ij}_2   \big( (\bfv^i - \bfv^j )\cdot\bfn^{ij} \big)^2. 
\end{aligned}   
\end{align}

%
%

For rotational kinetic energy, we take an inner product of \eqref{dem_w} with $\omega^i$, and sum up the resulting equations over $i \in [n]$ to get 
\begin{align}
\begin{aligned} \label{eq:dtM2w}
    \frac{\d M_{2,\omega}}{\d t} & = \sum_{i=1}^n I^i\frac{\d\omega^i}{\d t} \cdot \omega^i \\
    & =  \sum_{i=1}^n \Big( \frac{1}{n} \sum_{j=1}^{n} ( r^i\bfn^{ij}\times \bff_{\bft}^{ij}) + T^i_e \Big) \cdot \omega^i   \\
    & = \frac{1}{n}  \sum_{i, j=1}^n  r^i\bfn^{ij}\times \bff_{\bft}^{ij} \omega^i + \sum_{i=1}^{n} T^i_e \omega^i    \\
    & =:  \T_{14} + \sum_{i=1}^{n} T^i_e \omega^i.
\end{aligned}
\end{align}

Using the definition \eqref{eq:demparm} and that $\bfn^{ij}\times \bft^{ij}=1$ and $\bft^{ij}\cdot \bft^{ij}=1$, we now calculate
\begin{align}
\begin{aligned} \label{eq:i3i4}
    \T_{1,3} + \T_{1,4} & = \frac{1}{n}  \sum_{i, j=1}^n \bff_{\bft}^{ij} \cdot \bfv^i + \frac{1}{n}  \sum_{i, j=1}^n  \bfn^{ij}\times \bff_{\bft}^{ij} r^i \omega^i \\
   & =   -\frac{1}{n}  \sum_{i, j=1}^n \bff_{\bft}^{ij} \cdot \bfv^j + \frac{1}{n}  \sum_{i, j=1}^n  \bfn^{ij}\times \bff_{\bft}^{ij} r^j \omega^j \\
   & =  \frac{1}{2n}  \sum_{i, j=1}^n \bff_{\bft}^{ij} \cdot (\bfv^i - \bfv^j) + \frac{1}{2n}  \sum_{i, j=1}^n  \bfn^{ij}\times \bff_{\bft}^{ij} (r^i \omega^i +r^j \omega^j) \\
   & =  \frac{1}{2n}  \sum_{i, j=1}^n \bff_{\bft}^{ij} \cdot \big[ (\bfv^i - \bfv^j) \cdot \bft^{ij} \big] \bft^{ij} + \frac{1}{2n}  \sum_{i, j=1}^n \bff_{\bft}^{ij} \cdot (r^i \omega^i +r^j \omega^j)  \bft^{ij}\\
    & = - \frac{1}{2n}  \sum_{i, j=1}^n \bff_{\bft}^{ij} \cdot \sigma^{ij}_t \bft^{ij} \\
    & = - \frac{1}{2n}  \sum_{i, j=1}^n  \zeta^{ij} \kappa^{ij}_3 t_c^{ij} (\sigma^{ij}_t)^2.
\end{aligned}
\end{align}

Summing up the above estimates in \eqref{eq:dtM2v} with terms \eqref{eq:T11}, \eqref{eq:T12}, and \eqref{eq:dtM2w} with term \eqref{eq:i3i4} gives the desired estimate for the total energy. 
\end{proof}

\begin{lemma}[Total energy lower bound] \label{lem:sumparticleEb}
Let $(\bfx^i, \bfv^i, \theta^i, \omega^i)$ be a global solution to the system \eqref{eq:dem}. 
Assume a uniform bound on the duration of contact time, i.e., $t_c^{ij} < t_c^{\max}$.
Under zero environmental force and torque, i.e. $\bfF^i_e = \bfzr,\ T^i_e=0, \forall~i \in [n],$ there exist positive constants $A_0$ and $A_1$ such that 
\[ M_2(t)  \ge M_{2}(0) e^{-A_0 t} + \frac{A_1}{A_0} | \bfM_{1,\bfv}(0)|^2  \Big( 1 - e^{-A_0 t} \Big). \]
\end{lemma}

\begin{proof}
With zero environmental force and torque, the equalities \eqref{eq:dtM1v} and \eqref{eq:dtM1w} reduce to 
\begin{equation} \label{eq:dtM1vw0} 
\frac{\d \bfM_{1,\bfv}}{\d t} = \bfzr, 
\quad \frac{\d M_{1,\omega}}{\d t} = 0 
\end{equation}
while the relation \eqref{eq:dtM2} reduces to 
\begin{align*}
\begin{aligned} 
\frac{\d M_2}{\d t}  & =  \frac{1}{2n} \sum_{i, j=1}^n \kappa^{ij}_2 |(\bfv^i - \bfv^j )\cdot\bfn^{ij} |^2 - \frac{1}{2n}  \sum_{i, j=1}^n  \zeta^{ij} \kappa^{ij}_3 t_c^{ij} (\sigma^{ij}_t)^2.
\end{aligned}
\end{align*}

Using definition \eqref{eq:demparm}, $\kappa^{ij}_2 <0$ (since $\eta_0<0$ as in \eqref{eq:beta}), and positivity of the other parameters, we observe
\[
\frac{\d M_2}{\d t} \leq 0, \quad \mbox{i.e., total energy is dissipative.}
\]
Now, we use the Cauchy–Schwarz inequality and boundedness of the mass of floe particles to see that there exist positive constants, as the floe particle radii, thickness, and masses are fixed for a fixed system with $n$ floes, such that 
\begin{align}
\begin{aligned} \label{eq:dtM2v-inq}
&\frac{1}{2n} \sum_{i, j=1}^n |\kappa^{ij}_2|  \big( (\bfv^j  - \bfv^i )\cdot\bfn^{ij} \big)^2  \leq \frac{1}{2n}\kappa^{\max}_2  \sum_{i, j=1}^n |\bfv^j - \bfv^i |^2 \\
&\hspace{2cm} \leq \frac{\kappa^{\max}_2}{2n m^2_{\min}}  \sum_{i, j=1}^n m^im^j |\bfv^j - \bfv^i |^2 \\
&\hspace{2cm} =  \frac{\kappa^{\max}_2}{2n m^2_{\min}}  \sum_{i, j=1}^n m^im^j \Big( |\bfv^j |^2 + |\bfv^i |^2 - 2  \bfv^i \cdot \bfv^j \Big)  \\
&\hspace{2cm}\le  \frac{\kappa^{\max}_2}{2n  m^2_{\min}} \Big( 2n m_{\max} \sum_{i=1}^n m^i |\bfv^i |^2 - 2 \Big| \sum_{i=1}^n   m^i\bfv^i \Big|^2 \Big )  \\
&\hspace{2cm}= \frac{\kappa^{\max}_2 m_{\max}}{ m^2_{\min}} \sum_{i=1}^n m^i |\bfv^i |^2 -   \frac{\kappa^{\max}_2}{n m^2_{\min}}  \Big| \sum_{i=1}^n m^i  \bfv^i \Big|^2  \\
&\hspace{2cm} := \T_{21} - \T_{22},
\end{aligned}
\end{align}
where $\kappa^{\max}_2 = \max_{i,j}  |\kappa^{ij}_2|$ and $m_{\min} = \min_{i} m^i,\ m_{\max} = \max_{i} m^i.$

By definition, we bound $\T_{21}$ as
\begin{align}
\begin{aligned} \label{eq:T21-inq}
\T_{21} = \frac{2\kappa^{\max}_2 m_{\max}}{ m^2_{\min}} M_{2,\bfv} \leq \frac{2\kappa^{\max}_2 m_{\max}}{ m^2_{\min}} M_{2}.
\end{aligned}
\end{align}
We use \eqref{eq:dtM1vw0} to see
\begin{equation} \label{eq:T22-inq}
 \T_{22} = \frac{\kappa^{\max}_2}{n m^2_{\min}}  | \bfM_{1,\bfv}(0)|^2.
\end{equation}


\noindent
To bound the term $\frac{1}{2n}  \sum_{i, j=1}^n  \zeta^{ij} \kappa^{ij}_3 t_c^{ij} (\sigma^{ij}_t)^2$,  
we recall
\[
\sigma^{ij}_t
=
(\bfv^j-\bfv^i)\cdot\bft^{ij}-r^j\omega^j-r^i\omega^i,
\qquad |\bft^{ij}|=1.
\]
Using $(a+b+c)^2\le 3(a^2+b^2+c^2)$ and $|(\bfv^j-\bfv^i)\cdot\bft^{ij}|\le|\bfv^j-\bfv^i|$,

\begin{equation*}
(\sigma_t^{ij})^2
\le
3\Big(|\bfv^j-\bfv^i|^2 + (r^j\omega^j)^2 + (r^i\omega^i)^2\Big).
\end{equation*}
With $\kappa_3^{ij}\le \kappa_3^{\max} := \max_{i,j}\kappa_3^{ij}$ and $0\le \zeta^{ij}\le 1$ by \eqref{eq:zeta}, we then have
\begin{align}
\sum_{i,j=1}^n  \zeta^{ij} \kappa^{ij}_3  t_c^{ij} (\sigma^{ij}_t)^2
&\le
\kappa_3^{\max} t_c^{\max}\sum_{i,j=1}^n (\sigma^{ij}_t)^2 \nonumber\\
&\le
3\kappa_3^{\max} t_c^{\max}
\sum_{i,j=1}^n\Big(|\bfv^j-\bfv^i|^2 + (r^j\omega^j)^2 + (r^i\omega^i)^2\Big).
\label{eq:vrw}
\end{align}
We use
\[
\sum_{i,j=1}^n |\bfv^j-\bfv^i|^2
=
2n\sum_{i=1}^n|\bfv^i|^2
-2\Big|\sum_{i=1}^n \bfv^i\Big|^2
\le
2n\sum_{i=1}^n|\bfv^i|^2,
\]
and
\[
\sum_{i,j=1}^n\Big((r^j\omega^j)^2+(r^i\omega^i)^2\Big)
=
2n\sum_{i=1}^n (r^i\omega^i)^2.
\]
We substitute into \eqref{eq:vrw} to find
\begin{equation}\label{eq:dsum}
\sum_{i,j=1}^n  \zeta^{ij} \kappa^{ij}_3  t_c^{ij} (\sigma^{ij}_t)^2
\le
6n \kappa_3^{\max} t_c^{\max}
\Big(
\sum_{i=1}^n |\bfv^i|^2
+
\sum_{i=1}^n (r^i\omega^i)^2
\Big).
\end{equation}

Similarly, we have
\[
M_{2,\bfv}=\frac12\sum_{i=1}^n m^i|\bfv^i|^2 \ge \frac{m_{\min}}{2}\sum_{i=1}^n|\bfv^i|^2
\quad\Rightarrow\quad
\sum_{i=1}^n|\bfv^i|^2 \le \frac{2}{m_{\min}} M_{2,\bfv}.
\]
Moreover $I^i=\frac{1}{2}m^i(r^i)^2\ge \frac{1}{2}m_{\min}(r^i)^2$, hence we have
\[
I^i(\omega^i)^2 \ge \frac{1}{2}m_{\min}(r^i\omega^i)^2
\quad\Longrightarrow\quad
\sum_{i=1}^n(r^i\omega^i)^2
\le
\frac{2}{m_{\min}}\sum_{i=1}^n I^i(\omega^i)^2
=
\frac{4}{m_{\min}} M_{2,\omega}.
\]
We put these into \eqref{eq:dsum} to get
\[
\frac{1}{2n}\sum_{i,j=1}^n  \zeta^{ij} \kappa^{ij}_3  t_c^{ij} (\sigma^{ij}_t)^2
\le
\frac{6\kappa_3^{\max} t_c^{\max}}{m_{\min}} \Big(M_{2,\bfv}+2M_{2,\omega}\Big) \le \frac{12\kappa_3^{\max} t_c^{\max}}{m_{\min}} M_2.
\]

Finally, we combine \eqref{eq:vrw} with \eqref{eq:dtM2v-inq}, \eqref{eq:T21-inq} and \eqref{eq:T22-inq} to obtain

\begin{align*}
\begin{aligned}
\frac{\d M_2}{\d t} &\geq  -
\Big( \frac{2\kappa^{\max}_2 m_{\max}}{ m^2_{\min}} + \frac{12\kappa_3^{\max} t_c^{\max}}{m_{\min}} \Big) M_{2}
+ \frac{\kappa^{\max}_2}{n m^2_{\min}} | \bfM_{1,\bfv}(0)|^2  \\
&=: - A_0 M_{2} + A_1 | \bfM_{1,\bfv}(0)|^2.
\end{aligned}
\end{align*}

\noindent
We use Grönwall's lemma \cite{gronwall1919note,evans2022partial} to find the desired estimate. 
\end{proof}

\begin{remark}
    Similar to the result \cite{ha2008Kinetic} for the Cucker-Smale model, if the environmental force and torque are zero, Lemma \ref{lem:sumparticleE} and Lemma \ref{lem:sumparticleEb} imply that the total energy $M_2$ is monotonically decreasing with a lower bound. 
    Lemma \ref{lem:sumparticleEb} uses an assumption on the duration of contact $t_c^{ij}$. 
    This is due to the unregularized formula $t_c^{ij}\sim |\bfv^i-\bfv^j|^{-1/5}$ in \eqref{eq:demparm} not being uniformly bounded as $|\bfv^i-\bfv^j|\to 0$. To obtain an energy-type estimate bound, one needs to remove this singularity. The common floe modeling practice \cite{herman2016discrete,damsgaard2018application}: either (i) caps $t_c^{ij}$ by a prescribed maximal contact time (often tied to time-step size $dt$ in numerical simulation); or (ii) regularizes with a small $v_\ast>0$ as 
\[
t_c^{ij}
:= 2.94\Big(\frac{m_e^{ij}}{\kappa_1^{ij}}\Big)^{2/5}
\Big(|\bfv^i-\bfv^j|+v_\ast\Big)^{-1/5},
\qquad v_\ast>0.
\]
On the other hand, Coulomb friction law imposed through $\zeta^{ij}$ also keeps the term $\sum_{i,j=1}^n  \zeta^{ij} \kappa^{ij}_3  t_c^{ij} (\sigma^{ij}_t)^2$ from being unbounded.  
\end{remark}

We now consider a special case where the ocean velocity is constant and the other environmental force and torque are zero. In such a case, we show that the particle translational and rotational velocities converge to constants. 
To establish this result, we first recall Barbalat's lemma to be used in later sections.
It claims that if a uniformly continuous function is (Riemann)-integrable on the positive real line, then it converges to 0, as $t$ goes to infinity.

\begin{lemma}[Barbalat's Lemma] \label{lem:barb}
Suppose that $f:\bbr_+\to \bbr_+$ is a uniformly continuous function, then
\[
\int_0^\infty f(t)\d t<\infty \quad\implies\quad \lim_{t\to\infty} f(t)=0.
\] 
\end{lemma}

\begin{theorem} \label{thm:v0} Suppose that the given ocean surface velocity is constant:
\[  \bfu^i = \bfu^{\infty}:~\mbox{constant}, \quad \forall~i \in [n], \]
and let $(\bfx^i, \bfv^i, \theta^i, \omega^i)$ be a global solution to system \eqref{eq:dem}. 
If $\alpha^i_o,\ \beta^i_o>0, \bfF^i_w=\bfF^i_E=\bfF^i_g=\bfzr$ and $T^i_w=0,$ we have
\[ \lim_{t \to \infty} \|  \bfv^i  -  \bfu^{\infty} \| =0, \quad  \quad \lim_{t \to \infty} \omega^i(t) = 0, \quad \forall~i \in [n].  \]
\end{theorem}
\begin{proof} 
Under such a setting, after applying the Galilean transformation
\[
\bfx_i^{\rm rel}:=\bfx_i-\bfu^\infty t,
\qquad
\bfv_i^{\rm rel}:=\bfv_i-\bfu^\infty,
\]
the constant-ocean-velocity case is reduced to the zero-ocean-velocity case
\(\bfu^\infty=\bfzr\).
In the following, we suppress the
superscript ``{\rm rel}'' for simplicity.
We claim that
\begin{equation} \label{eq:vw0}
\lim_{t \to \infty} \bfv^i(t) = \bfzr, \quad \lim_{t \to \infty} \omega^i(t) = 0.
\end{equation}
It follows from Lemma \ref{lem:sumparticleE} that 
\begin{equation}\label{eq:barbalet1}
\frac{\d M_2}{\d t}  \le - \sum_{i=1}^{n} ( \alpha^i_o \left| \bfv^i\right|^3 + \beta^i_o \left| \omega^i\right|^3),
\end{equation}
where the first two terms in the right-hand side of \eqref{eq:dtM2} are negative 
(since $\kappa^{ij}_2 < 0, \zeta\kappa^{ij}_3 > 0$). 
Thus, total energy is non-increasing, which leads to boundedness of $\kappa_1^{ij}\delta^{ij},\ \bfv^i,$ and $\omega^i$ since each energy $M_{2,\bfx},\ M_{2,\bfv},\ M_{2,\omega}$ becomes bounded.
Next, we integrate inequality~\eqref{eq:barbalet1},
\[
\int_0^\infty \sum_{i=1}^n(\alpha^i_o \left| \bfv^i\right|^3 + \beta^i_o \left| \omega^i\right|^3)\d t\le M_2(0)-\inf_{t>0} M_2(t)\le M_2(0),
\]
to see that total energy is non-negative.
Therefore, if we show that $\bfv^i$ and $\omega^i$ are uniformly continuous, we can show that it tends to 0, by Barbalat's Lemma \ref{lem:barb}.
We will prove it using the boundedness of $\dot{\bfv}^i$ and $\dot{\omega}^i$ (since the time derivative of the integrand of the above integral is these terms multiplied by bounded terms).
Recall
\begin{align*}
m^i\frac{\d\bfv^i}{\d t} &=   \frac{1}{n} \sum_{j=1}^{n}(\bff_{\bfn}^{ij} + \bff_{\bft}^{ij}) - \alpha^i_o \bfv^i\left| \bfv^i\right|,\\
I^i\frac{\d\omega^i}{\d t} &= \frac{1}{n} \sum_{j=1}^{n} ( r^i\bfn^{ij}\times \bff_{\bft}^{ij})  - \beta^i_o \omega^i \left|\omega^i \right|.
\end{align*}
Under the boundedness of $\delta^{ij}$, one can show the boundedness of $\kappa_i$ and $\sigma^{ij}$, which leads to the boundedness of right-hand-side of the above equations.
Therefore, $\bfv^i$ and $\omega^i$ are uniformly continuous, which shows~\eqref{eq:vw0}.
\end{proof}

\section{From particle to kinetic description} \label{sec:kmodel}
In this section, we first recall the formal derivation from the particle model to the kinetic model as the mean-field approximation of the particle model \eqref{eq:dem} with $n \gg 1$. 
We assume that the number of particles involved in the particle system \eqref{eq:dem} is sufficiently large so that it becomes meaningful to use the mean-field approximation via the one-particle distribution function to describe the overall effective dynamics of the original system.

\subsection{Kinetic model for ice floe dynamics}
We first rewrite the floe particle model \eqref{eq:dem} as 
\begin{equation}
\begin{cases} \label{eq:dem2}
 \displaystyle \frac{\d\bfx^i}{\d t} = \bfv^i, \quad \frac{\d\theta^i}{\d t}  = \omega^i, \qquad i \in [n], \\[0.2cm]
 \displaystyle \frac{\d\bfv^i}{\d t} = \frac{1}{m^i} \Big[  \frac{1}{n} \sum_{j=1}^{n} (\bff_{\bfn}^{ij} + \bff_{\bft}^{ij})  + \bfF^i_e \Big] =: \frac{\bfF^i}{m^i}, \\[0.2cm]
   \frac{\d\omega^i}{\d t} = \frac{1}{I^i} \Big[ \frac{1}{n} \sum_{j=1}^{n} ( r^i\bfn^{ij}\times \bff_{\bft}^{ij})  + T^i_e \Big]
  := \frac{T^i}{I^i} \\[0.2cm]
 \displaystyle \frac{dr^i}{dt} = 0, \quad \frac{dh^i}{dt} = 0,
\end{cases} 
\end{equation}
where we assume that the floe sizes and thicknesses do not change in time. In what follows, we adopt the BBGKY hierarchy (Bogoliubov–Born–Green–Kirkwood–Yvon, \cite{Bogoliubov1946,bogoliubov1946kinetic,born1946general,kirkwood1946statistical,yvon1935theorie}) to derive a kinetic equation for the one-particle distribution function over the generalized phase space $\bbr_{\bfx}^2 \times \bbr_{\bfv}^2 \times \mathbb{T} \times \bbr_{\omega} \times \bbr_+ \times \bbr_+$ ($\mathbb{T}$ means 1D torus). 
In general, one assumes that the distribution function belongs to a function class such that it is periodic in the toroidal variables and rapidly decaying in the unbounded variables, ensuring that all boundary terms arising from integration by parts vanish.

\noindent $\bullet$~Step A (Derivation of the Liouville equation for the $n$-particle distribution function):  We define the $n$-particle distribution function on the $n$-particle phase space $\bbr^{4n}\times \mathbb{T}^{n}\times \bbr^{n} \times \bbr_+^{2n}$:
\begin{equation} \label{eq:pdfn}
 F^n = F^n(t,\bfx^1, \bfv^1, \theta^1, \omega^1, r^1, h^1, \cdots, \bfx^n, \bfv^n, \theta^n, \omega^n, r^n, h^n), 
\end{equation}
for $(\bfx^i, \bfv^i, \theta^i, \omega^i, r^i, h^i) \in \mathbb{R}^2\times \mathbb{R}^2 \times \mathbb{T} \times \bbr_{\omega} \times \mathbb{R}_+ \times \mathbb{R}_+, i \in [n]$. 
Note that the $n$-particle probability density function $F^n$ is symmetric in its phase variable in the sense that 
\begin{align}
\begin{aligned} \label{eq:pdfsym}
&  F^n(t,\cdots, \bfx^i, \bfv^i, \theta^i, \omega^i, r^i, h^i, \cdots, \bfx^j, \bfv^j, \theta^j, \omega^j, r^j, h^j, \cdots) \\
 & \hspace{1.5cm}   = F^n(t, \cdots, \bfx^j, \bfv^j, \theta^j, \omega^j, r^j, h^j, \cdots, \bfx^i, \bfv^i, \theta^i, \omega^i, r^i, h^i, \cdots).
\end{aligned}
\end{align}
Then, $F^n$ satisfies the Liouville equation on the generalized $n$-particle phase space:
\begin{equation} \label{eq:le}
\begin{aligned}
\partial_t F^n & + \sum_{i=1}^{n} \nabla_{\bfx^i} \cdot ( \dot{\bfx}^i F^n) + \sum_{i=1}^{n} \nabla_{\bfv^i} \cdot ( \dot{\bfv}^i F^n)  \\
& + \sum_{i=1}^{n} \partial_{\theta^i} ( \dot{\theta}^i F^n)  + \sum_{i=1}^{n} \partial_{\omega^i} ( \dot{\omega}^i F^n)  + \sum_{i=1}^{n} \partial_{r^i} ( \dot{r}^i F^n)  + \sum_{i=1}^{n} \partial_{h^i} ( \dot{h}^i F^n)  = 0,
\end{aligned}
\end{equation}
where $ \nabla_{\bfx^i}  \cdot (~  ~)$ and $ \nabla_{\bfv^i}  \cdot (~  ~)$ denote the divergences in the $\bfx^i$ and $\bfv^i$-variables, respectively. 
By using \eqref{eq:dem2}, the above system can also be written as 
\begin{equation} \label{eq:les}
\partial_t F^n +  \sum_{i=1}^{n} \bfv^i \cdot \nabla_{\bfx^i} F^n + \sum_{i=1}^{n} \nabla_{\bfv^i} \cdot \Big ( \frac{\bfF^i}{m^i} F^n \Big ) + \sum_{i=1}^{n} \omega^i \partial_{\theta^i} F^n + \sum_{i=1}^{n} \partial_{\omega^i} \Big( \frac{T^i}{I^i} F^n \Big)  = 0.
\end{equation}

\noindent $\bullet$~Step B (Derivation of equation  for the $j$-particle distribution function):  For notational simplicity, we set 
\[ \bfz := (\bfx,\bfv,\theta, \omega, r,h)    , \quad   \d\bfz^j = \d\bfx^j \d\bfv^j \d \theta^j \d \omega^j \d r^j \d h^j, \quad j \in [n], \quad  \d E^{n:j} = \prod_{i=j+1}^n\d\bfz^i. \]
Then, we introduce the $j$-marginal distribution function $F^{n:j}$ by the integration of \eqref{eq:pdfn} to give:
\[ F^{n:j}  := \int_{\bbr^{4(n-j)}\times \mathbb{T}^{n-j}\times \bbr^{n-j} \times \bbr_+^{2(n-j)}} F^{n} \d E^{n:j}. \]
We derive an equation for $ F^{n:j}$. For this,  by writing \eqref{eq:les} as follows
\begin{align}
\begin{aligned} \label{eq:les1}
 \partial_t F^n &= -\sum_{i=1}^j   \nabla_{\bfx^i} \cdot  (\bfv^i F^n) - \sum_{i=1}^{j} \nabla_{\bfv^i} \cdot \Big ( \frac{\bfF^i}{m^i} F^n \Big ) - \sum_{i=1}^{j}  \partial_{\theta^i} (\omega^iF^n) - \sum_{i=1}^{j} \partial_{\omega^i} ( \frac{T^i}{I^i} F^n)  \\
 & \quad - \sum_{i=j+1}^{n}  \nabla_{\bfx^i} \cdot ( \bfv^i F^n) - \sum_{i=j+1}^{n} \nabla_{\bfv^i} \cdot \Big ( \frac{\bfF^i}{m^i} F^n \Big ) - \sum_{i=j+1}^{n}  \partial_{\theta^i} (\omega^iF^n) - \sum_{i=j+1}^{n} \partial_{\omega^i} \Big( \frac{T^i}{I^i} F^n \Big).
\end{aligned}
\end{align}

Now, set $\hat{\bbr} = \bbr^{4(n-j)}\times \mathbb{T}^{n-j}\times \bbr^{n-j} \times \bbr_+^{2(n-j)},$ and integrate the Liouville equation \eqref{eq:les1} over $ E^{n:j}$
to obtain
\begin{align}
\begin{aligned} \label{eq:les2}
\partial_t F^{n:j} & =  -\sum_{i=1}^j    \int_{\hat{\bbr}}  \nabla_{\bfx^i} \cdot  (\bfv^i F^n) \d E^{n:j} - \sum_{i=1}^{j} \int_{\hat{\bbr}} \nabla_{\bfv^i} \cdot \Big( \frac{\bfF^i}{m^i} F^n \Big) \d E^{n:j} \\
& \quad -\sum_{i=1}^j  \int_{\hat{\bbr}}  \partial_{\theta^i} (\omega^iF^n) \d E^{n:j} - \sum_{i=1}^{j} \int_{\hat{\bbr}}  \partial_{\omega^i} \Big( \frac{T^i}{I^i} F^n \Big)   \d E^{n:j} \\
& -  \sum_{i=j+1}^{n}  \int_{\hat{\bbr}}  \nabla_{\bfx^i} \cdot ( \bfv^i F^n) \d E^{n:j} - \sum_{i=j+1}^{n} \int_{\hat{\bbr}} \nabla_{\bfv^i} \cdot \Big( \frac{\bfF^i}{m^i} F^n \Big) \d E^{n:j} \\
& -  \sum_{i=j+1}^{n}  \int_{\hat{\bbr}}  \partial_{\theta^i} (\omega^iF^n) \d E^{n:j} - \sum_{i=j+1}^{n} \int_{\hat{\bbr}} \partial_{\omega^i} \Big( \frac{T^i}{I^i} F^n \Big)  \d E^{n:j} \\
&=: \sum_{k=1}^8 \T_{3k}.
\end{aligned}
\end{align}
Using the divergence theorem and the decay condition of $F^n$ at infinity, it is easy to verify the identities
\[
\T_{3k} = 0, \quad k=5,6,7,8.
\]

In the next lemma, we estimate the terms $\T_{3i},~i=1,\cdots, 4$, one by one. 
For simplicity, we denote $\hat{\bbr}_1 = \bbr^{4}\times \mathbb{T}\times \bbr \times \bbr_+^{2}$. 
\begin{lemma} \label{lem:T3}
Let $F^n$ be a global solution to \eqref{eq:le} which decays to zero sufficiently fast at infinity for all phase variables.  Then, we have  the following identities:
\begin{align*}
& (i)~\T_{31} = -  \sum_{i=1}^j  \bfv^i  \cdot \nabla_{\bfx^i} F^{n:j}, \quad \T_{33}  =  -\sum_{i=1}^j  \omega^i  \partial_{\theta^i}  F^{n:j}; \\
& (ii)~\T_{32} =  -\frac{1}{n} \sum_{i=1}^{j}  \nabla_{\bfv^i} \cdot  \Big[ \Big( \frac{1}{m^i}  \sum_{k=1}^{j} \bff_c^{ik} \Big)  F^{n:j} \Big] 
 -\frac{n-j}{n} \sum_{i=1}^{j}  \int_{\hat{\bbr}_1} \nabla_{\bfv^i} \cdot \Big( \frac{1}{m^i} \bff_c^{i(j+1)} F^{n:j+1} \Big)\d\bfz^{j+1} \\
& \hspace{1.5cm}  - \sum_{i=1}^{j} \frac{1}{m^i}  \nabla_{\bfv^i} \cdot \Big(\bfF^i_e F^{n:j} \Big); \\
& (iii)~\T_{34} =-\frac{1}{n} \sum_{i=1}^{j}  \partial_{\omega^i} \Big( \frac{F^{n:j} }{I^i}  \sum_{k=1}^{j} ( r^i\bfn^{ik}\times \bff_{\bft}^{ik}) \Big)  \\
& \qquad \qquad -\frac{n-j}{n} \sum_{i=1}^{j}  \int_{\hat{\bbr}_1} \partial_{\omega^i}\Big( \frac{F^{n:j+1}}{I^i} ( r^i\bfn^{i(j+1)}\times \bff_{\bft}^{i(j+1)}) \Big)\d\bfz^{j+1}
- \sum_{i=1}^{j}  \frac{1}{I^i} \partial_{\omega^i} \Big( T^i_e F^{n:j} \Big).
\end{align*}
\end{lemma}

\begin{proof}
\noindent (i) With variable independence, exchanging derivative and integration gives 
\begin{align*}
\begin{aligned}
\T_{31} & =  -\sum_{i=1}^j  \int_{\hat{\bbr}}  \nabla_{\bfx^i} \cdot  (\bfv^i F^n) \d E^{n:j} =  -\sum_{i=1}^j  \nabla_{\bfx^i} \cdot (\bfv^i F^{n:j}) = -  \sum_{i=1}^j  \bfv^i  \cdot \nabla_{\bfx^i} F^{n:j}, \\
\T_{33} & =  -\sum_{i=1}^j  \int_{\hat{\bbr}}  \partial_{\theta^i} (\omega^iF^n)  \d E^{n:j} =  -\sum_{i=1}^j  \omega^i \partial_{\theta^i} F^{n:j}. 
\end{aligned}
\end{align*}

\noindent
(ii)~Now, we use \eqref{eq:pdfsym} to see that 
\begin{align*}
\begin{aligned} 
\T_{32} &= - \sum_{i=1}^{j} \int_{\hat{\bbr}} \nabla_{\bfv^i} \cdot \Big ( \frac{\bfF^i}{m^i} F^n \Big ) \d E^{n:j} \\
     & = -\frac{1}{n} \sum_{i=1}^{j}  \int_{\hat{\bbr}} \nabla_{\bfv^i} \cdot \Big( \frac{F^n }{m^i}  \sum_{k=1}^{n} \bff_c^{ik} \Big) \d E^{n:j} - \sum_{i=1}^{j}  \frac{1}{m^i} \int_{\hat{\bbr}} \nabla_{\bfv^i} \cdot \Big( \bfF^i_e F^n \Big) \d E^{n:j} \\
   & = - \frac{1}{n} \sum_{i=1}^{j}  \int_{\hat{\bbr}} \nabla_{\bfv^i} \cdot \Big( \frac{F^n }{m^i}  \sum_{k=1}^{j} \bff_c^{ik} \Big) \d E^{n:j} -\frac{1}{n} \sum_{i=1}^{j}  \int_{\hat{\bbr}} \nabla_{\bfv^i} \cdot \Big( \frac{F^n }{m^i}  \sum_{k=j+1}^{n} \bff_c^{ik} \Big) \d E^{n:j}  \\
   & \quad - \sum_{i=1}^{j}  \frac{1}{m^i} \int_{\hat{\bbr}} \nabla_{\bfv^i} \cdot \Big( \bfF^i_e F^n \Big) \d E^{n:j} \\
& =  -\frac{1}{n} \sum_{i=1}^{j}  \nabla_{\bfv^i} \cdot  \Big[ \Big( \frac{1}{m^i}  \sum_{k=1}^{j} \bff_c^{ik} \Big)  F^{n:j} \Big] 
 -\frac{n-j}{n} \sum_{i=1}^{j}  \int_{\hat{\bbr}_1} \nabla_{\bfv^i} \cdot \Big( \frac{1}{m^i} \bff_c^{i(j+1)} F^{n:j+1} \Big)\d\bfz^{j+1} \\
& \quad  - \sum_{i=1}^{j} \frac{1}{m^i}  \nabla_{\bfv^i} \cdot \Big( \bfF^i_e F^{n:j} \Big).
\end{aligned}
\end{align*}
(iii) Similarly, we use \eqref{eq:dem2} and \eqref{eq:pdfsym} to calculate
\begin{align*}
\T_{34} &= - \sum_{i=1}^{j} \int_{\hat{\bbr}}  \partial_{\omega^i} \Big( \frac{T^i}{I^i} F^n \Big)   \d E^{n:j} \\
     & = -\frac{1}{n} \sum_{i=1}^{j}  \int_{\hat{\bbr}} \partial_{\omega^i} \Big( \frac{F^n }{I^i}  \sum_{k=1}^{n} ( r^i\bfn^{ik}\times \bff_{\bft}^{ik}) \Big) \d E^{n:j}
     - \sum_{i=1}^{j}  \frac{1}{I^i} \int_{\hat{\bbr}} \partial_{\omega^i} \Big( T^i_e F^n \Big) \d E^{n:j} \\
    & =-\frac{1}{n} \sum_{i=1}^{j}  \partial_{\omega^i} \Big( \frac{F^{n:j} }{I^i}  \sum_{k=1}^{j} ( r^i\bfn^{ik}\times \bff_{\bft}^{ik}) \Big)  \\
   & \quad -\frac{n-j}{n} \sum_{i=1}^{j}  \int_{\hat{\bbr}_1} \partial_{\omega^i}\Big( \frac{F^{n:j+1}}{I^i} ( r^i\bfn^{i(j+1)}\times \bff_{\bft}^{i(j+1)}) \Big)\d\bfz^{j+1} 
 - \sum_{i=1}^{j}  \frac{1}{I^i} \partial_{\omega^i} \Big( T^i_e F^{n:j} \Big).
\end{align*}
These calculations complete the proof.
\end{proof}

For a fixed $j$, let  $n \to \infty$ and we assume that there exists a limit $F^j$ such that 
\[ \lim_{n \to \infty} F^{n:j} = F^j \quad \mbox{in a suitable sense}.  \]
Then, in \eqref{eq:les2}, we use Lemma \ref{lem:T3} and formally as $n \to \infty$, the limit $F^j$ satisfies 
\begin{align}
\begin{aligned} \label{eq:lej}
  \partial_t F^{j} & +  \sum_{i=1}^j  \bfv^i  \cdot \nabla_{\bfx^i} F^{j}  +\sum_{i=1}^{j}  \int_{\hat{\bbr}_1} \nabla_{\bfv^i} \cdot \Big( \frac{1}{m^i} \bff_c^{i(j+1)} F^{j+1} \Big)\d\bfz^{j+1} \\
  & + \sum_{i=1}^j \omega^i  \partial_{\theta^i}  F^j
  + \sum_{i=1}^{j}  \int_{\hat{\bbr}_1} \partial_{\omega^i}\Big( \frac{F^{j+1}}{I^i} ( r^i\bfn^{i(j+1)}\times \bff_{\bft}^{i(j+1)}) \Big)\d\bfz^{j+1} \\
 & + \sum_{i=1}^{j} \frac{1}{m^i}  \nabla_{\bfv^i} \cdot \Big( \bfF^i_eF^{j} \Big) 
  +  \sum_{i=1}^{j}  \frac{1}{I^i} \partial_{\omega^i} \Big( T^i_e F^j \Big)
= 0.
\end{aligned}
\end{align}
Note that the dynamics of $F^j$ in \eqref{eq:lej} depends on $F^{j+1}$.  In particular, for $j=1$, we remove the superscript for simplicity to arrive at
\begin{align}
\begin{aligned} \label{eq:lej1}
  \partial_t F & +   \bfv  \cdot \nabla_{\bfx} F  + \omega \partial_\theta F + \nabla_{\bfv} \cdot \big( \frac{\bfF_e}{m} F \big)  
  +  \partial_{\omega} \big( \frac{T_e}{I} F \big) \\
& +  \int_{\hat{\bbr}_1} \nabla_{\bfv} \cdot \Big( \frac{1}{m} \bff_c^{12} F^{2} \Big)\d\bfz^{2} 
+ \int_{\hat{\bbr}_1} \partial_{\omega}\Big( \frac{F^2}{I} ( r\bfn^{12}\times \bff_{\bft}^{12}) \Big)\d\bfz^{2}
= 0.
\end{aligned}
\end{align}

\noindent $\bullet$~Step C (Formal derivation of kinetic equation for one-particle distribution function):  We assume the ``molecular chaos assumption" by setting
\begin{equation} \label{eq:mca}
F^2(t, \bfz, \bfz^*) = F(t, \bfz) \otimes F(t, \bfz^*). 
\end{equation}
We then substitute ansatz \eqref{eq:mca} into \eqref{eq:lej1} to get the kinetic equation for $F$ as
\begin{align}
\begin{aligned} \label{eq:ke4f}
  \partial_t F & +   \bfv  \cdot \nabla_{\bfx} F  + \omega \partial_\theta F + \nabla_{\bfv} \cdot \big( \frac{\bfF_e}{m} F \big)  
  +  \partial_{\omega} \big( \frac{T_e}{I} F \big)
  +  \nabla_{\bfv} \cdot \Big[ \Big( \int_{\hat{\bbr}_1} \frac{1}{m}   \bff_c^{12} F(t, \bfz^*)\d\bfz^*  \Big)  F(t, \bfz) \Big ] \\
  & +  \partial_{\omega} \Big[ \Big( \int_{\hat{\bbr}_1} \frac{1}{I}  ( r\bfn^{12}\times \bff_{\bft}^{12}) F(t, \bfz^*)\d\bfz^*  \Big)  F(t, \bfz) \Big ] 
 = 0.
\end{aligned}
\end{align}

In what follows, we use the notation:
\begin{align} 
&   \gamma_1 :=  \frac{\kappa_1}{m},  \quad \gamma_2 :=  \frac{\kappa_2}{m}, \quad \gamma_3 :=  \frac{\zeta\kappa_3}{m}, \quad \gamma_4 :=  \frac{\zeta\kappa_3}{I}, \nonumber \\
& \bff[F] = \bff_e[F] +  \bff_{c}[F],  \quad \bff_e =  \frac{\bfF_e}{m},  \quad \bff_{c}[F]  = \bff_{c,\bfn}[F] +  \bff_{c,\bfv}[F] +  \bff_{c,\bft}[F], \nonumber \\
& \bff_{c,\bfn}[F](t,\bfz) :=  \int_{\hat{\bbr}_1}  \gamma_1  (|\bfx^* - \bfx| - (r + r^*)) \bfn(\bfx, \bfx^*)  F(t, \bfz^*)\d\bfz^*, \nonumber \\
& \bff_{c,\bfv}[F](t,\bfz) :=  \int_{\hat{\bbr}_1}  \gamma_2 [ (\bfv  - \bfv^*) \cdot \bfn(\bfx, \bfx^*)] \bfn(\bfx, \bfx^*) F(t, \bfz^*)\d\bfz^*, \label{eq:hnote} \\
& \bff_{c,\bft}[F](t,\bfz) :=  \int_{\hat{\bbr}_1}  \gamma_3  \sigma  \bft(\bfx, \bfx^*) F(t, \bfz^*)\d\bfz^*, \nonumber \\
& f_{\omega}[F] = f_{\omega,e}[F] + f_{\omega,t}[F], \quad  f_{\omega,e}[F] = \frac{T_e}{I},
\quad f_{\omega,t}[F](t,\bfz) : = \int_{\hat{\bbr}_1}  \gamma_4 r \sigma F(t, \bfz^*)\d\bfz^*, \nonumber
\end{align}
where the last integration uses $\bfn(\bfx, \bfx^*) \times \bft(\bfx, \bfx^*) = 1.$
Note that $\gamma_i$ are all functions of the phase variables $(r,h,r^*,h^*)$. 
Finally, we substitute \eqref{eq:hnote} into \eqref{eq:ke4f} to arrive at the Vlasov-McKean equation
\begin{equation} \label{eq:vme}
 \partial_t F +   \bfv  \cdot \nabla_{\bfx} F + \omega \partial_\theta F +  \nabla_{\bfv} \cdot (\bff[F] F) +  \partial_{\omega} (f_{\omega}[F] F)= 0.
\end{equation}

\noindent
We remark that the final two terms on the left-hand side of \eqref{eq:vme} account for the internal stress within the floe field due to collisions and the external environmental forcing.

\subsection{Macroscopic behavior of the kinetic model}
From now on, as long as there is no confusion, we suppress $t$-dependence in $F$, i.e., 
\[ F(\bfz) \equiv  F(t, \bfz), \quad \bfz \in \hat{\bbr}_1.  \]
For simplicity, we denote $\hat{\bbr}_2 = \bbr^{8}\times \mathbb{T}^2\times \bbr^2 \times \bbr_+^{4}$. 
and define the energy functional as
\begin{equation*}  
\begin{cases} 
\displaystyle \E  := \E_K + \E_P + \E_{\omega}, \\
\displaystyle \E_K := \frac{1}{2} \int_{\hat{\bbr}_1} m|\bfv|^2 F(\bfz) \d \bfz, \quad \E_{\omega} := \frac{1}{2} \int_{\hat{\bbr}_1} I\omega^2 F(\bfz) \d \bfz, \\
\displaystyle \E_P  := \frac{1}{2}\int_{\hat{\bbr}_2} \left(\int_0^{\delta(\bfz,\bfz_*)} \tilde{\kappa}_1(\eta)\eta \d \eta \right)F(\bfz)F(\bfz_*) \d \bfz\d\bfz_*,\qquad \tilde{\kappa}_1(\delta(\bfz,\bfz_*))=\kappa_1(\bfz,\bfz_*).
\end{cases}
\end{equation*}
%
We observe that each functional component in $\E$ is nonnegative, hence the energy functional $\E$ is nonnegative. In the following two lemmas, we establish the macroscopic behavior of the kinetic description of the floe dynamics. 
\begin{lemma}[Macroscopic behavior] \label{lem:kmb}
Let $F$ be a global smooth probability density solution to system \eqref{eq:vme} which decays sufficiently fast at infinity in phase space. Then, the following estimates hold:
\begin{align}
& (i)~  \frac{\d}{\d t} \int_{\hat{\bbr}_1} mF(t,\bfz) \d \bfz = 0. \nonumber \\
& (ii)~\frac{\d}{\d t} \int_{\hat{\bbr}_1}  m\bfv F \d \bfz 
=  \int_{\hat{\bbr}_1}  \bfF_e F(t, \bfz)  \d \bfz. \nonumber \\
& (iii)~\frac{\d}{\d t} \int_{\hat{\bbr}_1} (m\bfx\times \bfv +I \omega)F \d \bfz = \int_{\hat{\bbr}_1} \bfx \times \bfF_e F(t,\bfz)\d \bfz 
 +\int_{\hat{\bbr}_1} T_e F(t,\bfz)\d\bfz. \label{eq:kedt} \\ 
& (iv)~\frac{\d \E}{\d t} = \int_{\hat{\bbr}_1} \bfv \cdot  \bfF_e F(\bfz)  d \bfz  
 +  \int_{\hat{\bbr}_1}  \omega T_e  F(t, \bfz)  d \bfz \nonumber \\
& \hspace{1cm} + \frac{1}{2}  \int_{\hat{\bbr}_2}   \kappa_2 [ (\bfv  - \bfv^*) \cdot \bfn(\bfx, \bfx^*)]^2  F(\bfz) F(\bfz^*)\d\bfz^*\d\bfz \nonumber\\
&\hspace{1cm} -\frac12 \int_{\hat{\bbr}_2}  \zeta\kappa_3t_c \sigma_t^2  F(\bfz)  F(\bfz^*)\d\bfz^*\d\bfz. \nonumber
\end{align}
\end{lemma}
\begin{proof} 
\noindent (i) Note that $\frac{\d r}{\d t} = \frac{\d h}{\d t} = 0$ which implies $\frac{\d m}{\d t}=0.$ We multiply \eqref{eq:vme} by $m$, integrate it over the phase space $\hat{\bbr}_1$ and use the decay of $F$ at infinity to find the conservation of total mass as
\[ \frac{\d}{\d t} \int_{\hat{\bbr}_1} m F(t,\bfz) \d \bfz = 0. \]

(ii)~Next, we derive a balance law for momentum. For this, we multiply \eqref{eq:vme} by $\bfv$ to see that 
\begin{equation} \label{eq:kexv}
\partial_t (\bfv F) +    \nabla_{\bfx} \cdot \big (\bfv \otimes \bfv F \big)  +  \nabla_{\bfv} \cdot  \Big (  \bfv  \otimes \bff[F] F  \Big ) + \partial_\theta (\omega \bfv F) +  \partial_{\omega} (\bfv f_{\omega}[F] F)=  \bff_e[F] F +  \bff_{c}[F] F.
\end{equation}
Similarly, using anti-symmetry, we multiply \eqref{eq:kexv} by $m$ and integrate it over the phase space $\hat{\bbr}_1$ to find 
\begin{align*}
\begin{aligned} 
& \frac{\d}{\d t} \int_{\hat{\bbr}_1} m\bfv F(t,\bfz) \d \bfz 
=  \int_{\hat{\bbr}_1}  \bfF_e F \d\bfz + \int_{\hat{\bbr}_1} m\bff_{c}[F] F \d\bfz \\
&\hspace{1cm}  = \int_{\hat{\bbr}_1}  \bfF_e F \d\bfz  + \int_{\hat{\bbr}_2}  m\bff_c(\bfz, \bfz^*)  F(\bfz^*) F(\bfz)\d\bfz^*\d\bfz \\
&\hspace{1cm}  = \int_{\hat{\bbr}_1}  \bfF_e F \d\bfz  - \int_{\hat{\bbr}_2}  m\bff_c(\bfz^*, \bfz)  F(\bfz^*) F(\bfz)\d\bfz^*\d\bfz \\    
&\hspace{1cm}  =\int_{\hat{\bbr}_1}  \bfF_e F \d\bfz.
\end{aligned}
\end{align*}

\vspace{0.1cm}
\noindent (iii) 
Using the property of cross product $\bfv\times \bfv=0$, we multiply \eqref{eq:vme} by $m\bfx\times \bfv$ and integrate the resulting relation to find
\begin{align*}
&\frac{\d}{\d t}\int_{\hat{\bbr}_1} (m\bfx\times \bfv)F(t,\bfz) \d\bfz\\
&\hspace{1.2cm}=-\int_{\hat{\bbr}_1}(m\bfx\times \bfv)\nabla_{\bfx} \cdot \big (\bfv F \big)  \d\bfz
-  \int_{\hat{\bbr}_1}(m\bfx\times \bfv)\nabla_{\bfv} \cdot  \Big ( \bff[F] F  \Big ) \d\bfz\\
&\hspace{1.2cm}-\int_{\hat{\bbr}_1}\partial_\theta (m\omega (\bfx\times \bfv) F) \d\bfz
- \int_{\hat{\bbr}_1}\partial_{\omega} (m(\bfx\times \bfv) f_{\omega}[F] F)\d\bfz\\
&\hspace{1.2cm}= \int_{\hat{\bbr}_1}m\bfv\times \bfv F   \d\bfz
+ \int_{\hat{\bbr}_1} m\bfx\times \bff[F] F  \d\bfz\\
&\hspace{1.2cm}=\int_{\hat{\bbr}_1} \bfx\times \bfF_e F \d\bfz+\int_{\hat{\bbr}_1} m\bfx\times (\bff_{c, \bfn}[F]+\bff_{c, \bfv}[F]+\bff_{c, \bft}[F]) F  \d\bfz\\
&\hspace{1.2cm}=S_1+S_{21}+S_{22}+S_{23}.
\end{align*}
Next estimate $S_{2i}$, $i=1,2,3$ one by one.
We use the anti-symmetry of the integrand under the transformation $\bfz\leftrightarrow \bfz^*$ to get
\begin{align*}
S_{21} 
&=  \int_{\hat{\bbr}_2}  m\gamma_1  (|\bfx^* - \bfx| - (r + r^*)) \ (\bfx\times  \bfn(\bfx, \bfx^*))  F(t, \bfz^*)\d\bfz^*\d \bfz\\
&=  -\int_{\hat{\bbr}_2}  m\gamma_1  (|\bfx^* - \bfx| - (r + r^*)) \ (\bfx^*\times  \bfn(\bfx, \bfx^*))  F(t, \bfz^*)\d\bfz^*\d \bfz\\
&=  \frac12\int_{\hat{\bbr}_2}  m\gamma_1  (|\bfx^* - \bfx| - (r + r^*)) \ ((\bfx-\bfx^*)\times  \bfn(\bfx, \bfx^*))  F(t, \bfz^*)\d\bfz^*\d \bfz=0. \\
S_{22} 
&=  \int_{\hat{\bbr}_2}  m\gamma_2 [ (\bfv  - \bfv^*) \cdot \bfn(\bfx, \bfx^*)]  \ (\bfx\times \bfn(\bfx, \bfx^*)) F(t, \bfz^*)\d\bfz^*\d \bfz\\
&=  -\int_{\hat{\bbr}_2}  m\gamma_2 [ (\bfv  - \bfv^*) \cdot \bfn(\bfx, \bfx^*)]  \ (\bfx^*\times \bfn(\bfx, \bfx^*))  F(t, \bfz^*)\d\bfz^*\d \bfz\\
&=  \frac12\int_{\hat{\bbr}_2}m\gamma_2 [ (\bfv  - \bfv^*) \cdot \bfn(\bfx, \bfx^*)]  \ ((\bfx-\bfx^*)\times \bfn(\bfx, \bfx^*))     F(t, \bfz^*)\d\bfz^*\d \bfz=0. \\
S_{23} 
&= \int_{\hat{\bbr}_2}  m(\bfx\times  \bff_{c, \bft}[F])F(t, \bfz^*) \d\bfz^*\d \bfz. 
\end{align*}
Next, we calculate the spin angular part to get
\begin{align*}
\frac{\d}{\d t} \int_{\hat{\bbr}_1}I \omega F(t,\bfz) \d \bfz 
&=- \int_{\hat{\bbr}_1}\nabla_{\bfx} \cdot \big (I \omega\ \bfv F \big)  \d\bfz
-  \int_{\hat{\bbr}_1}\nabla_{\bfv} \cdot  \Big ( I \omega\bff[F] F  \Big ) \d\bfz\\
&-\int_{\hat{\bbr}_1}\partial_\theta (I \omega^2 F) \d\bfz
- \int_{\hat{\bbr}_1} I \omega \partial_{\omega} ( f_{\omega}[F] F)\d\bfz\\
&=\int_{\hat{\bbr}_1} T_e F  + I f_{\omega,t}[F] F\d\bfz\\
&= S_3+S_4.
\end{align*}

Following the same argument as in \eqref{eq:Ta}, we combine $S_{23}$ and $S_4$ to calculate 
\begin{align*}
S_{23} + S_4 
&= \int_{\hat{\bbr}_2}  m(\bfc\times  \bff_{c, \bft}[F])F(t, \bfz^*) \d\bfz^*\d \bfz=0.
\end{align*}
This leads to the desired result by summing up the above terms.

\vspace{0.1cm}
\noindent (iv)~We multiply $\frac{1}{2} |\bfv |^2$ by \eqref{eq:vme} and $\frac{1}{2} \omega^2$ to \eqref{eq:vme} to get
\begin{align}
\begin{aligned} \label{eq:kevv}
& \partial_t \Big(\frac{1}{2} |\bfv |^2 F \Big) +   \nabla_{\bfx} \cdot \Big (\frac{1}{2} |\bfv |^2 \bfv F \Big)  +  \nabla_{\bfv} \cdot  \Big (\frac{1}{2} |\bfv |^2 \bff[F] F  \Big ) + \partial_\theta (\omega |\bfv |^2 F)/2 +  \partial_{\omega} (|\bfv |^2 f_{\omega}[F] F)/2 \\
& \hspace{2.5cm} = (\bfv \cdot \bfF_e/m) F +  (\bfv \cdot \bff_{c}[F]) F, \\
& \partial_t \Big(\frac{1}{2} \omega^2 F \Big) +   \nabla_{\bfx} \cdot \Big (\frac{\omega^2}{2} \bfv F \Big)  +  \nabla_{\bfv} \cdot  \Big (\frac{\omega^2}{2}  \bff[F] F  \Big ) + \partial_\theta (\omega^3  F)/2 +  \partial_{\omega} (\omega^2 f_{\omega}[F] F)/2 \\
& \hspace{2.5cm} = (\omega T_e/I) F +  (\omega f_{\omega,t}[F]) F.
\end{aligned}
\end{align}
Now, we multiply both equations in \eqref{eq:kevv} by $m$ and $I$, respectively, and integrate them over $\hat{\bbr}_1$ to find 
\begin{align}
\begin{aligned} \label{eq:kevvww}
 \frac{\d}{\d t} \int_{\hat{\bbr}_1} \frac{1}{2} m|\bfv |^2 F(\bfz) \d \bfz &  = \int_{\hat{\bbr}_1}  (\bfv \cdot \bfF_e) F(\bfz)  d \bfz + \int_{\hat{\bbr}_1}   m(\bfv \cdot \bff_{c}[F])(\bfz) F(t,\bfz)   d \bfz 
 \\
&  =:  \T_{41} + \T_{42}+ \T_{43}+ \T_{44}, \\
 \frac{\d}{\d t} \int_{\hat{\bbr}_1} \frac{1}{2} I\omega^2 F(\bfz) \d \bfz &  = \int_{\hat{\bbr}_1}  (\omega T_e) F(\bfz)  d \bfz + \int_{\hat{\bbr}_1}   I(\omega f_{\omega,t}[F])(\bfz) F(t,\bfz)   d \bfz 
 \\
&  =:  \T_{45} + \T_{46}.
\end{aligned}
\end{align}
$\T_{41}$ and $\T_{45}$ are terms for the environmental component. 
Below, we estimate the term $\T_{4i}, i=2,3,4,6,$ one by one. \newline

\noindent $\bullet$~Case D.1 (Estimate of $\T_{42}$):~By direct calculation, we obtain
\begin{align}
\begin{aligned} \label{eq:T42}
\T_{42} & = \int_{\hat{\bbr}_2}   \kappa_1 (|\bfx^* - \bfx| - (r + r^*)) \bfv \cdot \bfn(\bfx, \bfx^*)  F( \bfz^*) F(\bfz)\d\bfz^* d \bfz   \\
& = \int_{\hat{\bbr}_2}   \kappa_1  (|\bfx - \bfx^*| - (r^* + r)) \bfv^* \cdot \bfn(\bfx^*, \bfx)  F( \bfz) F(\bfz^*)\d\bfz d \bfz^*   \\
&=  -\int_{\hat{\bbr}_2}   \kappa_1  (|\bfx^* - \bfx| - (r + r^*)) \bfv^* \cdot \bfn(\bfx, \bfx^*)  F( \bfz^*) F(\bfz)\d\bfz^* d \bfz   \\
&= \frac{1}{2} \int_{\hat{\bbr}_2}   \kappa_1  (|\bfx^* - \bfx| - (r + r^*))( \bfv - \bfv^*)  \cdot \bfn(\bfx, \bfx^*)  F( \bfz^*) F(\bfz)\d\bfz^* d \bfz \\
&= -\frac{\d \E_P}{\d t},
\end{aligned}
\end{align}
where we used an exchange map $\bfz~\longleftrightarrow~\bfz^*$ and $\bfn(\bfx^*, \bfx) = -\bfn(\bfx, \bfx^*)$. 
For simplicity, we set 
\[
p(\bfz,\bfz_*):=\int_0^{\delta(\bfz,\bfz_*)} \tilde{\kappa}_1(\eta)\eta \d \eta,
\]
and note that it is only dependent on $\bfx,\ r,\ m,$ and $h$.
Since
\begin{align*}
\partial_t(F(\bfz)F(\bfz_*))
&+   (\bfv,\bfv_*)  \cdot (\nabla_{\bfx},\nabla_{\bfx_*}) (F(\bfz)F(\bfz_*))\\
&+ (\omega,\omega_*)\cdot (\partial_\theta, \partial_{\theta_*}) (F(\bfz)F(\bfz_*))\\
&+  (\nabla_{\bfv},\nabla_{\bfv_*}) \cdot \big((\bff[F](\bfz),\bff[F](\bfz_*)) F(\bfz)F(\bfz_*)\big)\\
&+  (\partial_{\omega}, \partial_{\omega_*}) \cdot\big((f_{\omega}[F](\bfz),f_{\omega}[F](\bfz_*)) F(\bfz)F(\bfz_*)\big)= 0,
\end{align*}
by integration by parts, we have
\begin{align*}
\frac{\d \E_P}{\d t}
&= \frac{1}{2}\int_{\hat{\bbr}_2} p(\bfz,\bfz_*)\partial_t(F(\bfz)F(\bfz_*)) \d \bfz\d\bfz_*\\
&= -\frac{1}{2}\int_{\hat{\bbr}_2} p(\bfz,\bfz_*) (\bfv,\bfv_*)  \cdot (\nabla_{\bfx},\nabla_{\bfx_*}) (F(\bfz)F(\bfz_*)) \d \bfz\d\bfz_*\\
&= \frac{1}{2}\int_{\hat{\bbr}_2} \left[\nabla_{\bfx}p(\bfz,\bfz_*)\cdot\bfv+\nabla_{\bfx_*}p(\bfz,\bfz_*)\cdot\bfv_*\right] F(\bfz)F(\bfz_*) \d \bfz\d\bfz_*.
\end{align*}
Note that divergence other than in $\bfx,\ \bfx_*$ cancels $p(\bfz,\bfz_*)$ so that they are zero.
Since $\nabla_{\bfx}p(\bfz,\bfz_*) = -\kappa_1(\bfz,\bfz_*)\delta(\bfz,\bfz_*)\bfn(\bfx,\bfx_*)$, 
\[
\frac{\d \E_P}{\d t}=-\frac{1}{2}\int \kappa_1\delta ( \bfv - \bfv^*)  \cdot \bfn(\bfx, \bfx^*)  F( \bfz^*) F(\bfz) \d\bfz^* \d \bfz=-\mathcal{T}_{42}.
\]

\noindent $\bullet$~Case D.2 (Estimate of $\T_{43}$): Similar to Case D.1, we have
\begin{align}
\begin{aligned} \label{eq:T43}
\T_{43} & = \int_{\hat{\bbr}_2} m  \gamma_2 [ (\bfv  - \bfv^*) \cdot \bfn(\bfx, \bfx^*)] \bfn(\bfx, \bfx^*)  \cdot \bfv F(\bfz) F(\bfz^*)\d\bfz^*\d\bfz \\
& = -\int_{\hat{\bbr}_2}  m \gamma_2 [ (\bfv  - \bfv^*) \cdot \bfn(\bfx, \bfx^*)] \bfn(\bfx, \bfx^*)  \cdot \bfv^* F(\bfz) F(\bfz^*)\d\bfz^*\d\bfz \\
& = \frac{1}{2}  \int_{\hat{\bbr}_2}   \kappa_2 [ (\bfv  - \bfv^*) \cdot \bfn(\bfx, \bfx^*)]^2  F(\bfz) F(\bfz^*)\d\bfz^*\d\bfz\le 0. \\
\end{aligned}
\end{align}

\noindent $\bullet$~Case D.3 (Estimate of $\T_{44}+\T_{46}$): Similarly, using the definition \eqref{eq:demparm}, we have
\begin{align}
\begin{aligned} \label{eq:T44T46}
\T_{44} + \T_{46}
& = \int_{\hat{\bbr}_2}   [m\gamma_3\bff_{\bft}(\bfz,\bfz^*) \cdot \bfv + I\gamma_4  r \omega\bfn (\bfz,\bfz^*)\times \bff_{\bft}(\bfz,\bfz^*)  ] F(\bfz)  F(\bfz^*)\d\bfz^*\d\bfz \\
& = \int_{\hat{\bbr}_2}  [\bff_{\bft}(\bfz,\bfz^*) \cdot \bfv +  r \omega\bfn (\bfz,\bfz^*)\times \bff_{\bft}(\bfz,\bfz^*)  ] F(\bfz)  F(\bfz^*)\d\bfz^*\d\bfz \\
& = \int_{\hat{\bbr}_2}   [-\bff_{\bft}(\bfz,\bfz^*) \cdot \bfv^* +   r^* \omega^*\bfn (\bfz,\bfz^*)\times \bff_{\bft}(\bfz,\bfz^*)  ] F(\bfz)  F(\bfz^*)\d\bfz^*\d\bfz \\
& = \frac12\int_{\hat{\bbr}_2}   [(\bff_{\bft}(\bfz,\bfz^*) \cdot\bft(\bfz,\bfz^*)) \bft(\bfz,\bfz^*)\cdot (\bfv-\bfv^*)\\
&\hspace{1cm}+   (r \omega+r^*\omega^*)(\bff_{\bft}(\bfz,\bfz^*) \cdot\bft(\bfz,\bfz^*))\bfn (\bfz,\bfz^*)\times \bft(\bfz,\bfz^*) ] F(\bfz)  F(\bfz^*)\d\bfz^*\d\bfz \\
& =-\frac12 \int_{\hat{\bbr}_2}  (\bff_{\bft} \cdot \bft)[(\bfv^* - \bfv)\cdot \bft - r\omega -r^*\omega^*]  F(\bfz)  F(\bfz^*)\d\bfz^*\d\bfz\\
& =-\frac12 \int_{\hat{\bbr}_2}  \zeta\kappa_3t_c \sigma_t^2 F(\bfz)  F(\bfz^*)\d\bfz^*\d\bfz\le 0,
\end{aligned}
\end{align}
where we used anti-symmetry and
\[
\bff_{\bft}(\bfz,\bfz^*) = (\bff_{\bft}(\bfz,\bfz^*) \cdot\bft(\bfz,\bfz^*)) \bft(\bfz,\bfz^*),\qquad \bfn (\bfz,\bfz^*)\times \bft(\bfz,\bfz^*)=1.
\]
In \eqref{eq:kevvww}, we combine all the estimates \eqref{eq:T42}, \eqref{eq:T43} and \eqref{eq:T44T46}  to get the desired estimate in \eqref{eq:kedt}. 
\end{proof}

Similar to the particle model case, the last equation in \eqref{eq:kedt} implies energy dissipation when the environmental force and torque are zero. Now, we consider a case with constant ocean velocity and zero impact from atmospheric drag, Coriolis effect, and ocean tilt.

\begin{theorem} \label{thm:vw=0}
Suppose the ocean velocity is a constant $\bfu^\infty$ with 
$\alpha_o,\ \beta_o>0, \bfF_w=\bfF_E=\bfF_g=\bfzr$ and $T_w=0.$
Let $F(\bfx, \bfv, \theta, \omega, r,h)$ be a global solution to system \eqref{eq:vme}.
Then, we have
\[
\frac{\d \E}{\d t}\le \int_{\hat{\bbr}_1}\alpha_o \bfv\cdot(\bfu^\infty-\bfv)|\bfu^{\infty}-\bfv| F(t,\bfz)d\bfz-\int_{\hat{\bbr}_1}\beta_o |\omega|^3 F(t,\bfz)d\bfz.
\]
Furthermore, if the solution is compactly supported in the $\bfuu_o^\infty$-moving frame, and
the floe radius $r$ and thickness $h$ are bounded from below and above once the
system is initialized, then
\[
\lim_{t\to \infty}\E_K^{\rm rel}(t)=0,
\qquad
\lim_{t\to \infty}\E_\omega(t)=0,
\]
where
\[
\E_K^{\rm rel}(t)
:=
\frac12
\int_{\hat{\bbr}_1}
m|\bfv-\bfu^\infty|^2F(t,\bfz)\,d\bfz .
\]
\end{theorem}
\begin{proof} 
We use the proof of Theorem 3.5 in \cite{deng2025particle} and 
Lemma \ref{lem:kmb} to find 
\begin{equation*}
\frac{\d \E}{\d t}\le\int_{\hat{\bbr}_1}\alpha_o \bfv\cdot(\bfu^\infty-\bfv)|\bfu^{\infty}-\bfv| F(t,\bfz)d\bfz-\int_{\hat{\bbr}_1}\beta_o |\omega|^3 F(t,\bfz)d\bfz,
\end{equation*}
since the other two terms are non-positive (noting that $\kappa_2 \le 0$).
We introduce auxiliary variables
\[
\bfx^{\rm rel}:=\bfx-\bfu^\infty t,
\qquad
\bfv^{\rm rel}:=\bfv-\bfu^\infty,
\]
and define
\[
F^{\rm rel}
(t,\bfx^{\rm rel},\bfv^{\rm rel},\theta,\omega,r,h)
:=
F(t,\bfx^{\rm rel}+\bfu^\infty t,
      \bfv^{\rm rel}+\bfu^\infty,\theta,\omega,r,h).
\]
With these variables, the system has the same form as the
case $\bfu^\infty=\bfzr$. In the rest of the proof, we suppress the superscript ``{\rm rel}'' for simplicity.
The above inequality reduces to 
\begin{equation}\label{eq:kEnInq}
\frac{\d \E}{\d t}\le-\int_{\hat{\bbr}_1}\alpha_o |\bfv|^3 F(t,\bfz)d\bfz-\int_{\hat{\bbr}_1}\beta_o |\omega|^3 F(t,\bfz)d\bfz,
\end{equation}
so that the total energy is non-increasing.
Next, we integrate inequality~\eqref{eq:kEnInq}, to get
\[
\int_0^\infty \left(\int_{\hat{\bbr}_1}  \alpha_o \left| \bfv \right|^3 F(t, \bfz)  d \bfz   +  \int_{\hat{\bbr}_1} \beta_o  |\omega|^3  F(t, \bfz)  d \bfz\right)\d t\le \E(0)-\inf_{t>0} \E(t)\le \E(0),
\]
as the total energy is non-negative.
Therefore, it is enough to show the uniform continuity of these two terms to prove the first claim using Barbalat's lemma.
We evaluate the time-derivative of the above terms as
\begin{align*}
\int_{\hat{\bbr}_1}\alpha_o & \left| \bfv \right|^3 \partial_t F(\bfz)  \d \bfz \\
& = -\int_{\hat{\bbr}_1}\alpha_o\left| \bfv \right|^3 \big( \bfv  \cdot \nabla_{\bfx} F + \omega \partial_\theta F +  \nabla_{\bfv} \cdot (\bff[F] F) +  \partial_{\omega} (f_{\omega}[F] F) \big) \d \bfz  \\
& = -\int_{\hat{\bbr}_1}\alpha_o\left| \bfv \right|^3  \nabla_{\bfv} \cdot (\bff[F] F)  \d \bfz  \\
& =  
3\int_{\hat{\bbr}_1}\alpha_o\left|\bfv \right| \bfv \cdot (\bfF_e/m +\bff_c [F])F  \d \bfz := \T_\alpha,\\
\int_{\hat{\bbr}_1}\beta_o & \left| \omega \right|^3 \partial_t F(\bfz)  \d \bfz \\
& = -\int_{\hat{\bbr}_1}\beta_o\left| \omega \right|^3 \big( \bfv  \cdot \nabla_{\bfx} F + \omega \partial_\theta F +  \nabla_{\bfv} \cdot (\bff[F] F) +  \partial_{\omega} (f_{\omega}[F] F) \big) \d \bfz  \\
& = -\int_{\hat{\bbr}_1}\beta_o\left| \omega \right|^3  \partial_{\omega} (f_{\omega}[F] F) \d \bfz  \\
& = 3\int_{\hat{\bbr}_1}\beta_o\left| \omega \right| \omega (f_{\omega}[F] F) \d \bfz  := \T_\beta,
\end{align*}
where we used \eqref{eq:vme} and the divergence theorem with fast decay in the density function $F$ for vanishing boundaries. 

The integrals $\T_\alpha$ and $\T_\beta$ are bounded by some time-independent constant $C$. 
This can be shown by assuming that there exists a compact set $K\subset \hat{\mathbb R}_1$ and $r_0,h_0>0$ such that
\[
\text{supp } F(t,\cdot)\subset K\cap \{r\ge r_0\}\cap \{h\ge h_0\} \text{ for all } t\ge0.
\]
Since $K$ is compact, we can set
\[
\bar\alpha_o:=\|\alpha_o\|_{L^\infty(K)},\quad 
\bar\beta_o:=\|\beta_o\|_{L^\infty(K)},\quad 
\bar U:=\sup_{K}|\bfv|,\quad 
\bar\Omega:=\sup_{K}|\omega|.
\]
Moreover, all coefficients in the force fields are bounded on $K\times K$.
Using the force decomposition and $\bfF_w=\bfF_E=\bfF_g=\bfzr$ and $T_w=0$, we have
\[
\bff[F]= \bfF_e/m + \bff_c[F] = \bfF_o/m + \bff_c[F] \text{ and }\ f_\omega[F] = T_e/I +f_{\omega,c}[F] = T_o/I +f_{\omega,c}[F],
\]
which leads to 
\[
\|\bfF_o\|_{L^\infty(K)}\le  \bar\alpha_o\bar U^2\le C,\qquad
\|T_e\|_{L^\infty(K)}\le \bar\beta_o \bar\Omega^2\le C.
\]
For the contact terms, by compactness and mass conservation, we have
\begin{align*}
\|\bff_{c,\bfn}[F] \|_{L^\infty(K)} & \le \|\gamma_1\delta\|_{L^\infty(K^2)} \le C,\\
\|\bff_{c,\bfv}[F] \|_{L^\infty(K)} & \le \|\gamma_2|(\bfv-\bfv_*)\cdot\bfn|\|_{L^\infty(K^2)} \le C.
\end{align*}
The tangential part is controlled by the Coulomb cut-off
$|\zeta \kappa_3 \sigma|\le \mu |\kappa_1\delta|$:
\[
\|\bff_{c,\bft}[F] \|_{L^\infty(K)}\le \left\|\frac{\mu|\kappa_1\delta|}{m}\right\|_{L^\infty(K^2)} \le C.
\]
Hence 
\[
\|\bff[F] \|_{L^\infty(K)}+\|f_\omega[F] \|_{L^\infty(K)}\le C\ \text{ uniformly in }\ t.
\]
Therefore we have
\[
|\T_\alpha(t)|
\le 3\bar\alpha_o \bar U^2 \|\bff[F] \|_{L^\infty(K)}\int_K F(t,\bfz) d\bfz
\le C.
\]
Similarly, one has
\[
|\T_\beta(t)|
\le 3\bar\beta_o \bar\Omega^2 \|f_\omega[F] \|_{L^\infty(K)}\int_K F(t,\bfz) d\bfz
\le C,
\]
with a constant $C>0$ independent of $t$.
Thus $t\mapsto \int \alpha_o |\bfv|^3F$ and $t\mapsto \int \beta_o |\omega|^3F$
are uniformly continuous on $[0,\infty)$.
We use Barbalat's Lemma \ref{lem:barb} to get:
\[
\lim_{t\to\infty}\int_{\hat{\bbr}_1}\alpha_o |\bfv|^3 F(t,\bfz)d\bfz=0,\quad\text{and}\quad\lim_{t\to\infty}\int_{\hat{\bbr}_1}\beta_o |\omega|^3 F(t,\bfz)d\bfz=0.
\]
Lastly, using weighted Hölder's inequality, one can bound the translational kinetic energy as
\begin{align*}
    \E_K & = \frac{1}{2} \int_{\hat{\bbr}_1} m|\bfv|^2 F \d \bfz 
    \le \frac12\left(\int_{\hat{\bbr}_1} \alpha_o|\bfv|^3 F \d\bfz\right)^{2/3} \left(\int_{\hat{\bbr}_1}\frac{m^3F }{\alpha_o^2}\d\bfz\right)^{1/3} 
     \le C\left(\int_{\hat{\bbr}_1} \alpha_o|\bfv|^3 F \d\bfz\right)^{2/3}, 
\end{align*}
where the positive constant $C$ depends on the constants that bound the floe radius and thickness. 
Taking the limit $t\to \infty$ gives
\[
\lim_{t\to \infty}\E_K=0
\]
which is equivalently rewritten as,
\[
\lim_{t\to\infty}\E_K^{\rm rel}(t)=0.
\]
Similar arguments apply to estimate $\E_\omega$ to get the desired estimate.
\end{proof}
\color{black}

\section{From kinetic to hydrodynamic description} \label{sec:hmodel}

In this section, we derive the macroscopic (hydrodynamic) balance laws associated with the kinetic model
\eqref{eq:vme}, and then close them via a mono-kinetic ansatz.

\subsection{Hydrodynamic balance laws}\label{sec:hblaw}
Throughout this section, we denote
\[
\bfy := (\bfv,\theta,\omega,r,h)\in \bbr^2\times\mathbb{T} \times \bbr \times \bbr_+^2,
\qquad
\d\bfy := \d\bfv \d\theta \d\omega \d r \d h,
\]
$\hat{\bbr}_3 = \bbr^2\times\mathbb{T} \times \bbr \times \bbr_+^2$ and 
$\hat{\bbr}_4 = \bbr^4\times\mathbb{T} \times \bbr \times \bbr_+^2$.
Let $F=F(t,\bfx,\bfy)\ge 0$ be a sufficiently smooth solution of \eqref{eq:vme} that decays rapidly as $|\bfv|+|\omega| \to\infty$ and as $r+h\to\infty$, so that all integrations by parts below are justified, and boundary terms vanish.

Recall that the particle mass, $m$, and moment of inertia, $I$, depend on the size variables $(r,h)$, i.e.,
\[
m=m(r,h)>0,\qquad I=I(r,h)>0.
\]
From \eqref{eq:dem2}, 
\[
\frac{\d r}{\d t} = \frac{\d h}{\d t} = 0\quad\implies\quad \frac{\d m}{\d t}=\frac{\d I}{\d t}=0.
\]
We define the \emph{mass-weighted} and \emph{inertia-weighted} local moments (hydrodynamic fields):
\begin{equation}\label{eq:hlocmom}
\begin{cases}
\displaystyle \rho(t,\bfx) := \int_{\hat{\bbr}_3} m F(t,\bfx,\bfy) \d\bfy,\\[0.18cm]
\displaystyle (\rho \bfuu)(t,\bfx) := \int_{\hat{\bbr}_3} m \bfv F(t,\bfx,\bfy) \d\bfy,\\[0.18cm]
\displaystyle \rho_I(t,\bfx) := \int_{\hat{\bbr}_3} I F(t,\bfx,\bfy) \d\bfy,\\[0.18cm]
\displaystyle (\rho_I\bar\omega)(t,\bfx) := \int_{\hat{\bbr}_3} I \omega F(t,\bfx,\bfy) \d\bfy,\\[0.18cm]
\displaystyle \hat\omega(t,\bfx) := \int_{\hat{\bbr}_3} (m \bfx \times \bfv +I \omega) F(t,\bfx,\bfy) \d\bfy = \bfx \times (\rho \bfuu) + \rho_I \bar\omega,\\[0.18cm]
\displaystyle E(t,\bfx) := \int_{\hat{\bbr}_3} 
\frac12\Big(m|\bfv|^2+I\omega^2\Big) F(t,\bfx,\bfy) \d\bfy.
\end{cases}
\end{equation}
Herein, 
$\hat\omega$ denotes total angular momentum density, while $\bar\omega$ is the mean spin.
\noindent
We introduce $\tilde{\bfv}:=\bfv-\bfuu$ and the weighted fluxes:
\begin{equation}\label{eq:hflux}
\begin{aligned}
P(t,\bfx) & :=\int_{\hat{\bbr}_3} m \tilde{\bfv}\otimes\tilde{\bfv} F \d\bfy;\qquad
\bfs{q}(t,\bfx):=\int_{\hat{\bbr}_3}\Big(m|\tilde{\bfv}|^2+I(\omega-\bar\omega)^2\Big) \tilde{\bfv} F \d\bfy; \\
\bfs{J}(t,\bfx) &:=\int_{\hat{\bbr}_3} I (\omega-\bar\omega) \tilde{\bfv} F \d\bfy;
\qquad \bfs{J}_{\rm orb}(t,\bfx)
:=\int_{\hat{\bbr}_3} m (\bfx\times\tilde{\bfv})\otimes \tilde{\bfv}  F(t,\bfx,\bfy) \d\bfy,
\end{aligned}
\end{equation}
where one may omit $\otimes$ in the definition of the orbital flux tensor $\bfs{J}_{\rm orb}$ as it is 2D and $\bfx\times\tilde{\bfv}$ gives a scalar (we keep it here for consistency). 
Observe that
\begin{equation}\label{eq:hdenergy}
\begin{aligned}
E & =E_K+E_R+E_{\rm int},
\qquad
E_K:=\frac12\rho|\bfuu|^2,\quad
E_R:=\frac12\rho_I|\bar\omega|^2,\\
E_{\rm int} & = \rho e := \frac12\int_{\hat{\bbr}_3}\Big(m|\tilde{\bfv}|^2+I(\omega-\bar\omega)^2\Big)F \d\bfy.
\end{aligned}
\end{equation}

With this setting in mind, we have the following balance laws. 

\begin{lemma}[Hydrodynamic balance laws]\label{lem:hbalance}
Let $F$ be a smooth solution of \eqref{eq:vme} with sufficient decay at infinity in $\bfy$.
Then the moments \eqref{eq:hlocmom} satisfy the local balance laws:
\begin{subequations}\label{eq:hblaw}
\begin{align}
&\partial_t\rho+\nabla_{\bfx}\cdot(\rho\bfuu)=0; \label{eq:hmass}\\
&\partial_t(\rho\bfuu)+\nabla_{\bfx}\cdot(\rho\bfuu\otimes\bfuu+P)=\int_{\hat{\bbr}_3} m \bff[F] F \d\bfy; \label{eq:hmom}\\
&\partial_t\hat\omega+\nabla_{\bfx}\cdot(\hat\omega \bfuu+\bfs{J} + \bfs{J}_{\rm orb} +\bar\omega\int_{\hat{\bbr}_3} I \tilde{\bfv} F \d\bfy)=\int_{\hat{\bbr}_3} \big( m \bfx\times\bff[F] + I f_\omega[F] \big) F \d\bfy; \label{eq:hangmom}\\
&\partial_tE+\nabla_{\bfx}\cdot\Big(E \bfuu+P \bfuu+\bar\omega \bfs{J}+\frac12\bfs{q}+\frac{1}{2}\bar\omega^2\int_{\hat{\bbr}_3} I \tilde{\bfv} F \d\bfy\Big) \nonumber\\
&\hspace{1cm}=\int_{\hat{\bbr}_3} m \bfv\cdot\bff[F] F \d\bfy + \int_{\hat{\bbr}_3} I \omega f_\omega[F] F \d\bfy.\label{eq:henergy}
\end{align}
\end{subequations}
\end{lemma}

\begin{proof}
We establish the balance laws in several steps as follows.

\noindent Step A: (Conservation of mass):
We multiply \eqref{eq:vme} by $m$ and integrate the resulting relation over $\bfy$. 
Using the integration by parts in $\theta$, $\bfv$ and $\omega$, one has
\[
\partial_t \rho + \nabla_{\bfx}\cdot\int_{\hat{\bbr}_3} m \bfv F \d\bfy = 0.
\]
This is exactly \eqref{eq:hmass} by the definition of $\rho\bfuu$.

\vspace{0.1cm}
\noindent Step B: (Balance of Momentum):
We multiply \eqref{eq:vme} by $m\bfv$ and integrate the resulting relation over $\bfy$ to find
\[
\partial_t(\rho\bfuu)+\nabla_{\bfx}\cdot\int_{\hat{\bbr}_3} m \bfv\otimes\bfv F \d\bfy
+\int_{\hat{\bbr}_3} m \bfv \nabla_{\bfv}\cdot(\bff[F]F) \d\bfy =0.
\]
By integration by parts in $\bfv$ (boundary term vanishes), one has
\[
\int_{\hat{\bbr}_3} m \bfv \nabla_{\bfv}\cdot(\bff[F]F) \d\bfy = -\int_{\hat{\bbr}_3} m \bff[F]F \d\bfy.
\]
Next, we decompose $\bfv=\bfuu+\tilde{\bfv}$ to get
\[
\int_{\hat{\bbr}_3} m \bfv\otimes\bfv F \d\bfy
=\int_{\hat{\bbr}_3} m(\bfuu+\tilde{\bfv})\otimes(\bfuu+\tilde{\bfv}) F \d\bfy
=\rho\bfuu\otimes\bfuu + P,
\]
since $\int_{\hat{\bbr}_3} m \tilde{\bfv} F \d\bfy=0$ by the definition of $\bfuu$.
This gives \eqref{eq:hmom}.

\noindent Step C: (Spin angular momentum):
We multiply \eqref{eq:vme} by $I\omega$ and integrate it over $\bfy$ to get
\[
\partial_t(\rho_I\bar\omega)+\nabla_{\bfx}\cdot\int_{\hat{\bbr}_3} I\omega \bfv F \d\bfy
+\int_{\hat{\bbr}_3} I\omega \partial_\omega(f_\omega[F]F) \d\bfy =0.
\]
Integration by parts in $\omega$ yields
\[
\int_{\hat{\bbr}_3} I\omega \partial_\omega(f_\omega[F]F) \d\bfy=-\int_{\hat{\bbr}_3} I f_\omega[F]F \d\bfy.
\]
Moreover, we use $\bfv=\bfuu+\tilde{\bfv}$ and $\omega=\bar\omega+(\omega-\bar\omega)$ to obtain
\[
\int_{\hat{\bbr}_3} I\omega \bfv F \d\bfy
=\int_{\hat{\bbr}_3} I(\bar\omega+(\omega-\bar\omega))(\bfuu+\tilde{\bfv}) F \d\bfy
=\rho_I\bar\omega \bfuu+\bfs{J} +\bar{\omega}\int_{\hat{\bbr}_3} I \tilde{\bfv} F \d\bfy,
\]
where we used 
\[
\int_{\hat{\bbr}_3} I(\omega-\bar\omega) F \d\bfy=0.
\]
This yields 
\begin{equation}
    \partial_t(\rho_I\bar\omega)+\nabla_{\bfx}\cdot(\rho_I\bar\omega \bfuu+\bfs{J}+\bar\omega\int_{\hat{\bbr}_3} I \tilde{\bfv} F \d\bfy)=\int_{\hat{\bbr}_3} I f_\omega[F] F \d\bfy. \label{eq:spinmom}
\end{equation}

\noindent Step D: (Orbital angular momentum):
We multiply \eqref{eq:vme} by $m\bfx\times\bfv$ and integrate the resulting relation over $\bfy$ term by term: \\
\noindent (i) Time derivative,
\[
\int_{\hat{\bbr}_3} m(\bfx\times\bfv) \partial_t F \d\bfy
= \partial_t \left(\bfx\times(\rho\bfuu)\right).
\]

\noindent (ii) For $\bfx$-transport, we use
\[
\bfv\cdot\nabla_{\bfx}F=\nabla_{\bfx}\cdot(\bfv F),\qquad \bfv\cdot\nabla_{\bfx}(\bfx\times\bfv)=0,
\]
(since $\bfv\times\bfv=0$), to get
\[
\int_{\hat{\bbr}_3} m(\bfx\times\bfv) \bfv\cdot\nabla_{\bfx}F \d\bfy
=
\nabla_{\bfx}\cdot\left(\int_{\hat{\bbr}_3} m(\bfx\times\bfv)\otimes\bfv F \d\bfy\right) 
= \nabla_{\bfx}\cdot \big( \bfx\times (\rho\bfuu) \otimes\bfuu + \bfs{J}_{\rm orb}\big),
\]
where the mixed terms vanish as before. 

\vspace{0.1cm}
\noindent (iii) For $\bfv$-divergence term,
assuming sufficient decay in $\bfv$, so boundary terms vanish, we have
\[
\int_{\hat{\bbr}_3} m(\bfx\times\bfv) \nabla_{\bfv}\cdot(\bff[F]F) \d\bfy
=
-\int_{\hat{\bbr}_3} m \big(\nabla_{\bfv}(\bfx\times\bfv)\big) \cdot \bff[F]F \d\bfy
=
-\int_{\hat{\bbr}_3} m \bfx\times\bff[F] F \d\bfy.
\]

The $\theta$-transport and $\omega$-flux terms vanish by using integration by parts. Collecting (i)--(iii), we obtain the local orbital angular momentum balance as
\begin{equation}\label{eq:orbmom}
\partial_t \left(\bfx\times(\rho\bfuu)\right)
+
\nabla_{\bfx}\cdot \big( \bfx\times (\rho\bfuu) \otimes \bfuu + \bfs{J}_{\rm orb}\big)
=
\int_{\hat{\bbr}_3} m \bfx\times\bff[F] F \d\bfy.
\end{equation}
Finally, summing up the spin and orbital angular momentum balances, \eqref{eq:spinmom} and \eqref{eq:orbmom}, we arrive at the desired balance law for total angular momentum \eqref{eq:hangmom}.

\noindent Step E: (Total translational and rotational energy):
We multiply \eqref{eq:vme} by $\frac12\big(m|\bfv|^2+I\omega^2\big)$ and integrate the resulting relation over $\bfy$.
The transport terms give
\[
\partial_tE + \nabla_{\bfx}\cdot\int_{\hat{\bbr}_3} \tfrac12\big(m|\bfv|^2+I\omega^2\big) \bfv F \d\bfy,
\]
and the $\theta$-term vanishes by integration by parts.
For the force terms, we use integration by parts in $\bfv$ and $\omega$ (the mixed terms vanish using integration by parts) to give
\begin{align*}
\int_{\hat{\bbr}_3} \tfrac12 m|\bfv|^2  \nabla_{\bfv}\cdot(\bff[F]F) \d\bfy
&=  -\int_{\hat{\bbr}_3} m \bfv\cdot\bff[F]F \d\bfy,\\
\int_{\hat{\bbr}_3} \tfrac12 I\omega^2  \partial_\omega\big(f_\omega[F]F\big) \d\bfy
&=  -\int_{\hat{\bbr}_3} I \omega f_\omega[F]F \d\bfy.
\end{align*}
Finally, we decompose the energy flux, using $\bfv=\bfuu+\tilde{\bfv}$ and $\omega=\bar\omega+(\omega-\bar\omega)$, to obtain
\begin{align*}
\int_{\hat{\bbr}_3} \tfrac12\big(m|\bfv|^2+I\omega^2\big) \bfv F \d\bfy
&=E \bfuu + P \bfuu + \bar\omega \bfs{J} + \tfrac12\bfs{q}+\frac{1}{2}\bar\omega^2\int_{\hat{\bbr}_3} I \tilde{\bfv} F \d\bfy,
\end{align*}
where $P,\bfs{J},\bfs{q}$ are defined in \eqref{eq:hflux}--\eqref{eq:hdenergy} and where we have used
\[
\int_{\hat{\bbr}_3} m \tilde{\bfv} F \d\bfy=0,\qquad\int_{\hat{\bbr}_3} I(\omega-\bar\omega) F \d\bfy=0.
\]
Finally, we collect all the estimates to get the desired result.
\end{proof}

To close \eqref{eq:hblaw}, we follow \cite{deng2025particle,figalli2018rigorous} and adopt the mono-kinetic ansatz
\begin{equation}\label{eq:monoans}
F(t,\bfx,\bfv,\theta,\omega,r,h)
=
\frac{\rho(t,\bfx)}{m(r,h)}\Phi(t,\bfx,\theta, r,h)
\delta(\bfv-\bfuu(t,\bfx))\delta(\omega-\bar\omega(t,\bfx)),
\end{equation}
where $\Phi\ge 0$ and
\begin{equation}\label{eq:mu_norm}
\int_{\mathbb{T}\times \bbr_+^2}\Phi(t,\bfx,\theta,r,h)\d r\d h \d \theta =1, \qquad\text{for all }(t,\bfx).
\end{equation}
Under \eqref{eq:monoans} and \eqref{eq:mu_norm}, the weighted definitions are consistent, that is
\[
\int mF\d\bfy=\rho,\qquad \int m\bfv F\d\bfy=\rho\bfuu,
\qquad 
\rho_I(t,\bfx)=\rho(t,\bfx)\int_{\mathbb{T}\times\bbr_+^2}\frac{I(r,h)}{m(r,h)}\Phi(t,\bfx,\theta,r,h)\d r\d h \d \theta,
\]
\[
\int_{\hat{\bbr}_3} I \omega F(t,\bfx,\bfy) \d\bfy = \rho(t,\bfx)\ \bar\omega\int_{\mathbb{T}\times\bbr_+^2}\frac{I(r,h)}{m(r,h)} \Phi(t,\bfx,\theta,r,h)\d r\d h \d \theta=\rho_I\bar\omega.
\]
With this in mind, using \eqref{eq:monoans} and following \cite{deng2025particle}, we arrive at
\[
P\equiv 0,\qquad \bfs{q}\equiv 0,\qquad \bfs{J} \equiv 0,\qquad \bfs{J}_{\rm orb} \equiv 0,\qquad \rho e \equiv 0,
\]
and the energy reduces to the bulk form
\[
E(t,\bfx)=\frac12\rho(t,\bfx)|\bfuu(t,\bfx)|^2+\frac12\rho_I(t,\bfx) \bar\omega(t,\bfx)^2.
\]
The balance laws \eqref{eq:hblaw} reduce to
\begin{subequations}
    \label{eq:MK2H}
\begin{align}
&\partial_t\rho+\nabla_{\bfx}\cdot(\rho\bfuu)=0,\label{eq:MK_mass}\\
&\partial_t(\rho\bfuu)+\nabla_{\bfx}\cdot(\rho\bfuu\otimes\bfuu)=\int_{\hat{\bbr}_3} m \bff[F]F\d\bfy,\label{eq:MK_mom}\\
&\partial_t\hat\omega+\nabla_{\bfx}\cdot(\hat\omega \bfuu )= \bfx\times \int_{\hat{\bbr}_3}  m \bff[F] F\d\bfy + \int_{\hat{\bbr}_3} I f_\omega[F]  F \d\bfy,\label{eq:MK_angmom} \\
&\partial_tE
+\nabla_{\bfx}\cdot(E\bfuu)
=\bfuu\cdot\int_{\hat{\bbr}_3} m \bff[F]F\d \bfy+\bar\omega\int_{\hat{\bbr}_3} I  f_\omega[F]F\d\bfy.
\label{eq:MK_energy}
\end{align}
\end{subequations}
Using the definition \eqref{eq:hlocmom} and the mono-kinetic ansatz \eqref{eq:monoans}, we evaluate the right-hand side integrals of \eqref{eq:MK_energy} as
\begin{align}
\int_{\hat{\bbr}_3} m \bff[F] F \d\bfy
&=\int_{\hat{\bbr}_3} \bfF_e F + m ( \bff_{c,\bfn}[F] +  \bff_{c,\bfv}[F] +  \bff_{c,\bft}[F] ) F \d\bfy := \sum_{j=1}^4 \T_{5j}, \label{eq:T5j1}\\
\int_{\hat{\bbr}_3} I f_\omega[F] F \d\bfy
&=\int_{\hat{\bbr}_3} T_e F + I  \bff_{\omega,\bft}[F] F \d\bfy:= \T_{55} + \T_{56}. \label{eq:T5j2}
\end{align}
(a) For environmental-force induced terms $\T_{51}$ and $\T_{55}$, we have
\begin{align}
    \int_{\hat{\bbr}_3} \bfF_e F\d \bfy &= \bar\alpha_o (\bfu-\bfuu)|\bfu-\bfuu| + \bar\alpha_w (\bfw-\bfuu)|\bfw-\bfuu| - \rho f_E \bfz \times \bfuu + \rho \bfuu_g,  \label{eq:T51} \\
    \int_{\hat{\bbr}_3} T_e F\d\bfy & = \bar\beta_o \left(\frac{\nabla\times\bfu}{2}-\bar\omega\right)\left|\frac{\nabla\times\bfu}{2}-\bar\omega\right| + \bar\beta_w \left(\frac{\nabla\times\bfw}{2}-\bar\omega\right)\left|\frac{\nabla\times\bfw}{2}-\bar\omega\right|, 
    \label{eq:T55}
\end{align}
where 
\[
\bar\alpha_s = \int_{\hat{\bbr}_3} \alpha_s F \d \bfy,\quad\text{and}\quad \bar\beta_s = \int_{\hat{\bbr}_3} \beta_s F \d \bfy, \qquad s = o, w.
\]

\noindent
(b) For $\T_{52}$, by definition \eqref{eq:hnote}, we have
\begin{align} \label{eq:T52}
    \int_{\hat{\bbr}_3} m\bff_{c,\bfn}[F] F\d \bfy &= \int_{\hat{\bbr}_3} \int_{\hat{\bbr}_4}  \kappa_1 (|\bfx^* - \bfx| - (r + r^*)) \bfn(\bfx, \bfx^*)  F(t, \bfz^*)\d\bfz^* F\d \bfy.
\end{align}
The term $\T_{52}$ is generally non-zero pointwise over $ \ bfx$ due to the lack of integration over $\bfx$ that induces anti-symmetry of pair forces (see Lemma \ref{lem:kmb} on how this term vanishes when integration over $\bfx$ in included). If imposing the local homogeneity assumption (a strong assumption inspired by the kinetic theory of gas particles \cite{bardos1991fluid, golse2003mean, golse2005boltzmann}; cf., Appendix A), this term reduces to zero. Other terms vanish similarly, i.e., $\T_{53} = \bfzr, \T_{54} = \bfzr$ and $\bfx \times \T_{54} + \T_{56} = 0$. 
In this case, using \eqref{eq:T51}, \eqref{eq:T55} and the above setting, the first three balance laws in \eqref{eq:hblaw} or \eqref{eq:MK2H} further reduce to

\begin{subequations}
    \label{eq:LI2H}
\begin{align}
&\partial_t\rho+\nabla_{\bfx}\cdot(\rho\bfuu)=0,\label{eq:LI_mass}\\
&\partial_t(\rho\bfuu)+\nabla_{\bfx}\cdot(\rho\bfuu\otimes\bfuu)= \bar\alpha_o (\bfu-\bfuu)|\bfu-\bfuu| +\bar\alpha_w (\bfw-\bfuu)|\bfw-\bfuu| \nonumber\\
&\hspace{4cm}- \rho f_E \bfz \times \bfuu + \rho \bfuu_g := \bfF_m,\label{eq:LI_mom}\\
&\partial_t\hat\omega+\nabla_{\bfx}\cdot(\hat\omega \bfuu )=  \bfx\times \bfF_m + \bar\beta_o \left(\frac{\nabla\times\bfu}{2}-\bar\omega\right)\left|\frac{\nabla\times\bfu}{2}-\bar\omega\right|\label{eq:LI_angmom} \\
& \hspace{4.5cm} + \bar\beta_w \left(\frac{\nabla\times\bfw}{2}-\bar\omega\right)\left|\frac{\nabla\times\bfw}{2}-\bar\omega\right|. \nonumber
\end{align}
\end{subequations}

\begin{remark}[Contact operator]
    The main difference between \eqref{eq:MK2H} and \eqref{eq:LI2H} is that the right-hand side term in \eqref{eq:MK2H} includes the contact operator acting on the phase variables that are to be integrated in the distribution sense. 
    The right-hand side terms in \eqref{eq:LI_mom} and \eqref{eq:LI_angmom} 
    are environmental-force induced forcing terms. 
    The local homogeneity assumption resembles the assumption on the contact operator so that the physical laws of mass, momentum and energy conservation
    during collisions are satisfied (cf., \cite[Eq. (4)]{bardos1991fluid} and \cite{dufty2001kinetic,cercignani2013mathematical,saint2009hydrodynamic,golse2003mean,golse2005boltzmann,golse2016dynamics} among many for gas particles). 
    This assumption on the contact operator allows the simplification of the right-hand side terms in \eqref{eq:LI2H}.
    This represents a simplification of sea ice floe dynamics. In more realistic sea ice rheology, the contact-forcing terms in the density functionals generate the contact stress tensor within the ice cover. We refer to \cite{shen1987role,feltham2008sea,herman2022granular} for discussions of sea ice rheology and to Appendix A for further details.
\end{remark}

\subsection{Total energy for the closed system}\label{sec:hmodel_en}
We now analyze the energy dissipation properties of the closed system \eqref{eq:MK2H}.
We consider the more general case of \eqref{eq:MK2H} instead of \eqref{eq:LI2H} since the total energies are defined as integrals over $\bfx$, 
which implies vanishing right-hand side terms in \eqref{eq:MK_mom} and \eqref{eq:MK_angmom}. 
We first define the energies as
\begin{equation}  \label{eq:defhe}
\begin{cases} 
\displaystyle \E  := \E_K + \E_{\omega} + \E_P , \\
\displaystyle \E_K := \frac12\int_{\bbr^2}\rho |\bfuu|^2\d \bfx, \qquad \E_{\omega} := \frac12\int_{\bbr^2}\rho_I \bar\omega^2\d \bfx, \\
\displaystyle \E_P  := \frac12\int_{\bbr^4}
\Bigg[\int_{\hat\bbr_3^2}
F F^* 
\Big(\int_0^{\delta(\bfx,\bfx^*,r,r^*)}\kappa_1(\eta)\eta\d \eta\Big)
\d\bfy \d \bfy^*\Bigg]\d \bfx\d \bfx^*,
\end{cases}
\end{equation}
where $F^*$ denotes the density function in dual variables.
Following \cite{deng2025particle}, we have the following results.
\begin{lemma}[Kinetic and rotational energy]\label{lem:hend}
Let $(\rho,\bfuu,\hat\omega, E)$ be a global smooth solution to system \eqref{eq:MK2H}.
The following holds:
\begin{align}
    \frac{\d \E}{\d t} & = \int_{\bbr^2} \bar\alpha_o \bfuu \cdot (\bfu-\bfuu)|\bfu-\bfuu| \d\bfx + \int_{\bbr^2} \bar\alpha_w\bfuu \cdot (\bfw-\bfuu)|\bfw-\bfuu| \d\bfx 
    + \int_{\bbr^2} \bfuu \cdot \rho \bfuu_g \d\bfx \nonumber\\
    & \quad + \int_{\bbr^2} \bar\beta_o \bar\omega \left(\frac{\nabla\times\bfu}{2}-\bar\omega\right)\left|\frac{\nabla\times\bfu}{2}-\bar\omega\right| \d\bfx
    + \int_{\bbr^2} \bar\beta_w \bar\omega \left(\frac{\nabla\times\bfw}{2}-\bar\omega\right)\left|\frac{\nabla\times\bfw}{2}-\bar\omega\right| \d\bfx \nonumber \\
    & \quad + \frac{1}{2}  \int_{\hat{\bbr}_2}   \kappa_2 [ (\bfuu  - \bfuu^*) \cdot \bfn(\bfx, \bfx^*)]^2  F(\bfz) F(\bfz^*)\d\bfz^*\d\bfz \label{eq:hkre} \\
    & \quad - \frac12 \int_{\hat{\bbr}_2}  \zeta\kappa_3t_c[(\bfuu^* - \bfuu)\cdot \bft - r\bar\omega -r^*\bar\omega^*]^2 F(\bfz)  F(\bfz^*)\d\bfz^*\d\bfz. \nonumber
\end{align}
\end{lemma}

\begin{proof}
We integrate \eqref{eq:MK_energy} over $\bbr^2$ for $\bfx$ term-by-term. 

\vspace{0.1cm}
\noindent (i) Time derivative. We apply the definitions \eqref{eq:hlocmom} and \eqref{eq:defhe} and the mono-kinetic ansatz \eqref{eq:monoans} to compute
\[
\int_{\bbr^2} \partial_t E \d\bfx
= \partial_t \int_{\bbr^2} \int_{\hat{\bbr}_3} 
\frac12\Big(m|\bfv|^2+I\omega^2\Big) F(t,\bfx,\bfy) \d\bfy \d\bfx = \frac{\d}{\d t}(\E_K + \E_\omega).
\]

\noindent (ii) $\bfx$-transport.
Using the divergence theorem with vanishing boundaries gives 
\[
\int_{\bbr^2} \nabla_{\bfx}\cdot (E\bfuu) \d\bfx = 0.
\]

\noindent (iii) Right-hand side term. The calculation follows the calculation of the integrals in \eqref{eq:T5j1} and \eqref{eq:T5j2} with further integration over $\bbr^2$ for $\bfx$.
First, we apply the \eqref{eq:T51} and \eqref{eq:T55} to get the first four terms on the right-hand side of \eqref{eq:hkre}. Therein, we use $\bfuu \cdot \bfz \times \bfuu = 0$ to see that the Coriolis force conserves energy.
Secondly, for the normal contact force term (the term corresponding to $\T_{52}$ in \eqref{eq:T52}), using the mono-kinetic ansatz \eqref{eq:monoans}, we integrate it over $\bbr^2$ for $\bfx$ to arrive at
\begin{align*} 
    \int_{\bbr^2} \bfuu \cdot \int_{\hat{\bbr}_3} & m\bff_{c,\bfn}[F] F\d \bfy \d\bfx  \\
    & = \int_{\bbr^2} \bfuu \cdot \int_{\hat{\bbr}_3} \int_{\hat{\bbr}_4}  \kappa_1 (|\bfx^* - \bfx| - (r + r^*)) \bfn(\bfx, \bfx^*)  F(t, \bfz^*)\d\bfz^* F\d \bfy\d \bfx \\
    & = \int_{\hat{\bbr}_4} \int_{\hat{\bbr}_4}  \kappa_1 \bfuu \cdot  (|\bfx^* - \bfx| - (r + r^*)) \bfn(\bfx, \bfx^*)  F(t, \bfz^*)\d\bfz^* F\d \bfz \\
    & = \int_{\hat\bbr_2 }  \kappa_1 (\bfuu - \bfv + \bfv) \cdot  (|\bfx^* - \bfx| - (r + r^*)) \bfn(\bfx, \bfx^*)  F(t, \bfz^*)\d\bfz^* F\d \bfz \\
    & = \int_{\hat\bbr_2 }  \kappa_1 \bfv \cdot  (|\bfx^* - \bfx| - (r + r^*)) \bfn(\bfx, \bfx^*)  F(t, \bfz^*)\d\bfz^* F\d \bfz \\
    &= -\frac{\d \E_P}{\d t},
\end{align*}
where the last equality follows \eqref{eq:T42}.
Thirdly, for the damping term, we evaluate the integral in a similar fashion to give
\begin{align*} 
    \int_{\bbr^2} \bfuu \cdot \int_{\hat{\bbr}_3} & m\bff_{c,\bfv}[F] F\d \bfy \d\bfx  \\
    & = \int_{\bbr^2} \bfuu \cdot \int_{\hat{\bbr}_3} \int_{\hat{\bbr}_4}  \kappa_2 [ (\bfv  - \bfv^*) \cdot \bfn(\bfx, \bfx^*)] \bfn(\bfx, \bfx^*)  F(t, \bfz^*)\d\bfz^* F\d \bfy\d \bfx \\
    & = \int_{\hat\bbr^2 }  \kappa_2 \bfuu \cdot [ (\bfuu  - \bfuu^*) \cdot \bfn(\bfx, \bfx^*)] \bfn(\bfx, \bfx^*) F(t, \bfz^*)\d\bfz^* F\d \bfz \\
    & = \frac{1}{2}  \int_{\hat{\bbr}_2}   \kappa_2 [ (\bfuu  - \bfuu^*) \cdot \bfn(\bfx, \bfx^*)]^2  F(\bfz) F(\bfz^*)\d\bfz^*\d\bfz, 
\end{align*}
where the last equality follows from \eqref{eq:T43}.
We combine the final two integrals to arrive at a similar estimate as \eqref{eq:T44T46}:
\begin{align*} 
    \int_{\bbr^2} \bfuu \cdot \int_{\hat{\bbr}_3} & m\bff_{c,\bft}[F] F\d \bfy \d\bfx + \int_{\bbr^2} \bar\omega \int_{\hat{\bbr}_3}  I f_{\omega,\bft}[F] F\d \bfy \d\bfx \\
    & = \int_{\hat{\bbr}_2} \zeta \kappa_3 [\bfuu \cdot \sigma \bft(\bfz,\bfz^*)  + r \bar\omega\bfn (\bfz,\bfz^*)\times \sigma \bft(\bfz,\bfz^*)  ] F(\bfz)  F(\bfz^*)\d\bfz^*\d\bfz \\
    & = \int_{\hat{\bbr}_2} \zeta \kappa_3 \sigma [\bfuu \cdot  \bft(\bfz,\bfz^*)  + r \bar\omega] F(\bfz)  F(\bfz^*)\d\bfz^*\d\bfz \\
    & = \int_{\hat{\bbr}_2} \zeta \kappa_3 \sigma [-\bfuu^* \cdot  \bft(\bfz,\bfz^*)  + r^* \bar\omega^*] F(\bfz)  F(\bfz^*)\d\bfz^*\d\bfz \\
    & =-\frac12 \int_{\hat{\bbr}_2} \zeta\kappa_3 \sigma [(\bfuu^* - \bfuu)\cdot \bft - r\bar\omega -r^*\bar\omega^*]  F(\bfz)  F(\bfz^*)\d\bfz^*\d\bfz\\
    & =-\frac12 \int_{\hat{\bbr}_2}  \zeta\kappa_3 t_c[(\bfuu^* - \bfuu)\cdot \bft - r\bar\omega -r^*\bar\omega^*]^2 F(\bfz)  F(\bfz^*)\d\bfz^*\d\bfz.
\end{align*}
Summing up the estimates above and moving $-\frac{\d \E_P}{\d t}$ to the right-hand side gives the desired result. 
\end{proof}

\noindent
Lastly, we consider two special cases as an analogy in the particle and kinetic descriptions in Sections \ref{sec:pmodel} and \ref{sec:kmodel}.

\begin{itemize}

\item Case A: When the environmental forces are removed, the total energy estimate  in Lemma \ref{lem:hend} implies energy dissipation:
\begin{align*}
    \frac{\d \E}{\d t} & = \frac{1}{2}  \int_{\hat{\bbr}_2}   \kappa_2 [ (\bfuu  - \bfuu^*) \cdot \bfn(\bfx, \bfx^*)]^2  F(\bfz) F(\bfz^*)\d\bfz^*\d\bfz \\
    & \quad - \frac12 \int_{\hat{\bbr}_2}  \zeta\kappa_3 t_c[(\bfuu^* - \bfuu)\cdot \bft - r\bar\omega -r^*\bar\omega^*]^2 F(\bfz)  F(\bfz^*)\d\bfz^*\d\bfz \le 0,
\end{align*}
where the last inequality is a result of the fact that $\kappa_2<0$ and $\zeta\kappa_3>0.$ 

\item Case B: When the ocean force is constant $\bfu = \bfu^\infty$, $\bfF_w=\bfF_E=\bfF_g=\bfzr,$ and $T_w=0$, one expects similar decay result as Theorem \ref{thm:v0} for the particle model, and Theorem \ref{thm:vw=0} for the kinetic model. 
In particular, 
one can follow the proof of Theorem 4.3 of \cite{deng2025particle} with Galilean transform to establish  
\[
\lim_{t\to \infty}\E_K^{\rm rel}(t)=0,
\qquad
\lim_{t\to \infty}\E_\omega(t)=0,
\]
where
\[
\E_K^{\rm rel}(t)
:=
\frac12
\int_{\bbr^2}
\rho|\bfuu-\bfu^\infty|^2\,d\bfx .
\]
\end{itemize}

\section{Numerical simulations} \label{sec:num}
The purpose of the numerical experiments is twofold.
First, we verify the energy dissipation of the particle model.
In the simulations, we track the translational kinetic energy, the
rotational kinetic energy, and the contact potential energy, and confirm
that the total energy is dissipated
due to collisions, in agreement with the analytical energy estimates derived in
Section~\ref{sec:pmodel} (see results in Subsection \ref{subsec:EX1}).
Second, we investigate the consistency between the particle and continuum
descriptions by comparing ensemble-averaged particle quantities with the
corresponding solutions of hydrodynamic system, thereby validating the particle--kinetic--hydrodynamic hierarchy developed in this work (see results in Subsection \ref{subsec:Ex2}).
For this goal,
we apply the widely-used forward Euler scheme \cite{butcher2016numerical,hairer1993solving} to discretize \eqref{eq:dem} and \eqref{eq:LI2H} in time and the finite element method \cite{deng2021softfem,hughes2003finite} to discretize \eqref{eq:LI2H} in space. These methods are widely used, and we omit the details of the implementation of the numerical methods for simplicity.

\subsection{Example 1: Floes under constant ocean forcing}\label{subsec:EX1}

We solve the particle model \eqref{eq:dem} in a two-dimensional square dimensionless domain
$
\Omega := [-\pi,\pi]^2
$
with doubly periodic boundary conditions. 
To study and verify the asymptotical behavior of the particle model established in Section \ref{sec:pmodel}, we consider the simplified scenario where the environmental force is induced solely by the ocean. Thus, we assume $\bfF_w=\bfF_E=\bfF_g=\bfzr,$ and $T_w=0$ in \eqref{eq:dem}. 
The prescribed ocean surface velocity is 
\[
\bfu(x,y)= (0.3,0)^{T},
\]
which implies $\nabla\times\bfu = \bfzr$.
We generate $n=100$ floes, for which the radii satisfy
\[
r^i\in[r_{\min},r_{\max}],\qquad r_{\min}=0.08,\quad r_{\max}=0.32,
\]
and are sampled from a power-law distribution with exponent $2$, i.e., the density is proportional to $(r^i)^{-2}$ on $[r_{\min},r_{\max}]$. 
The thicknesses are sampled independently from the distribution
\[
h^i\sim \mathrm{Unif}[0.5,2].
\]
The ice density is fixed to be $\rho_{\mathrm{ice}}=1$ so that
$m^i=\rho_{\mathrm{ice}}\pi (r^i)^2 h^i, I^i=m^i(r^i)^2/2$.
The initial positions $\bfx^i(0)$ are sampled uniformly in $\Omega$ and accepted only if the configuration is non-overlapping in the periodic metric
\[
\big|\bfx^j(0)-\bfx^i(0)\big|\ge r^i+r^j,\qquad \forall  i\neq j.
\]
The floe radii are rescaled so that the initial concentration is approximately $0.5$.
The initial floe velocities and angular velocities are prescribed as (see Figure \ref{fig:Ex1v0w0})
\[
\bfv^i(0)=0.4\varphi(\bfx^i)\begin{pmatrix}-\sin(y^i)\\ \sin(x^i)\end{pmatrix}, 
\qquad 
\omega^i(0)=\tilde{\varphi}(\bfx^i),\qquad i\in[n],
\]
where
\begin{align*}
    \tilde{\varphi}(\bfx^i) & = 0.4 \phi(\bfx^i)\left(\cos(x^i)+\cos(y^i)-\frac{\sin^2(x^i)+\sin^2(y^i)}{4.41}\right), \\
    \varphi(\bfx^i) & =\exp\left(\frac{\cos(x^i)+\cos(y^i)-2}{4.41}\right).
\end{align*}
\begin{figure}[ht]
    \centering \includegraphics[width=13cm]{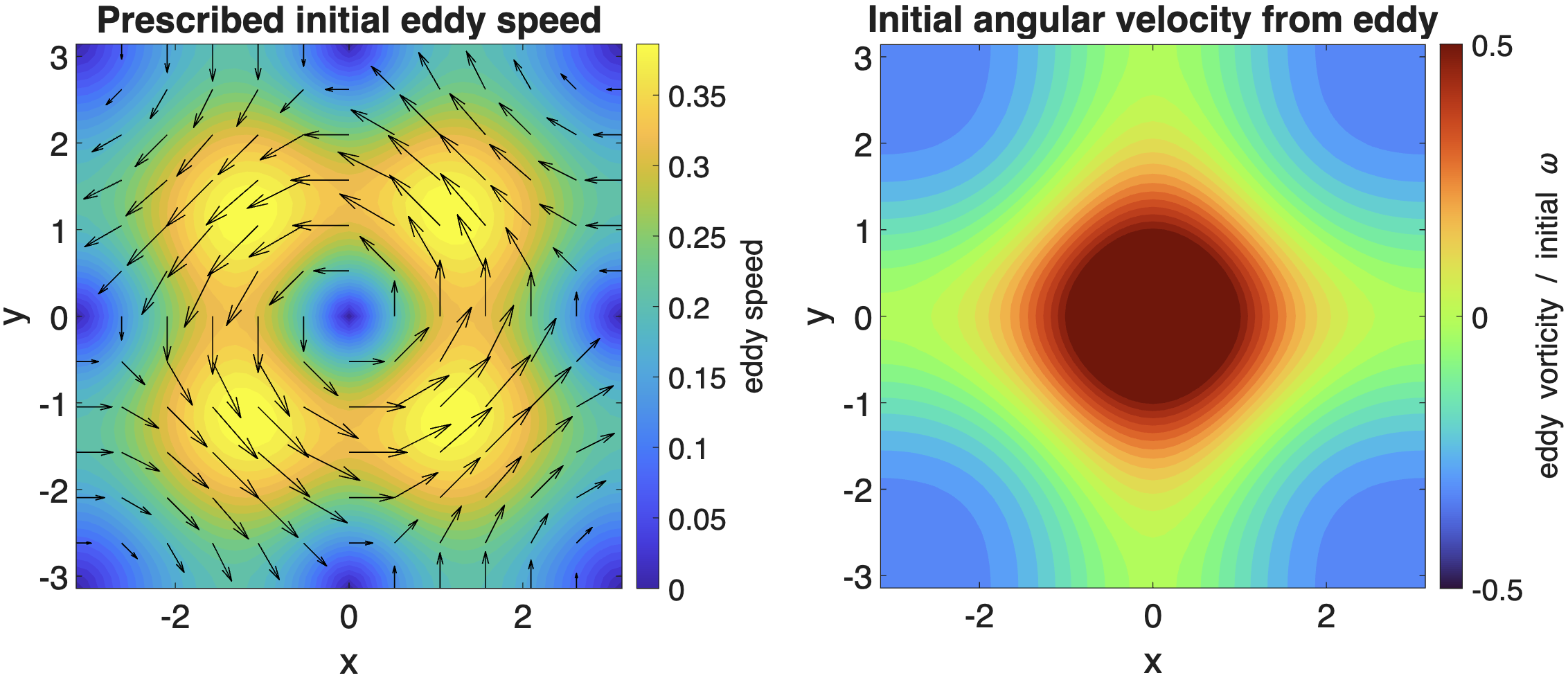} 
    \caption{Example 1 velocity and angular velocity field.}
    \label{fig:Ex1v0w0}
\end{figure}

We use the following values in a dimensionless setting (thus no units specified) for physical parameters for the model \eqref{eq:dem} described in Section \ref{sec:dem}, i.e.,
\[
\rho_{\mathrm{ice}}=1,\qquad e_r=0.15,\qquad \mu^{ij}\equiv 0.2,\qquad E=10^4,\qquad \nu=0.7,
\]
and compute the effective moduli $E_e$ and $G_e$ via \eqref{eq:EG}. 
The damping parameter $\eta$ is computed from \eqref{eq:beta}. 
For the ocean-induced forcing coefficients in \eqref{eq:demparm}, we set drag coefficients $\alpha_o^i=\beta_o^i=1$ to be $\mathcal{O}(1)$ constants (dimensionless scaling) so that the quadratic drag produces relaxation to the ocean velocity on the time interval $[0,10]$.
\begin{figure}[ht]
    \centering \includegraphics[width=13cm]{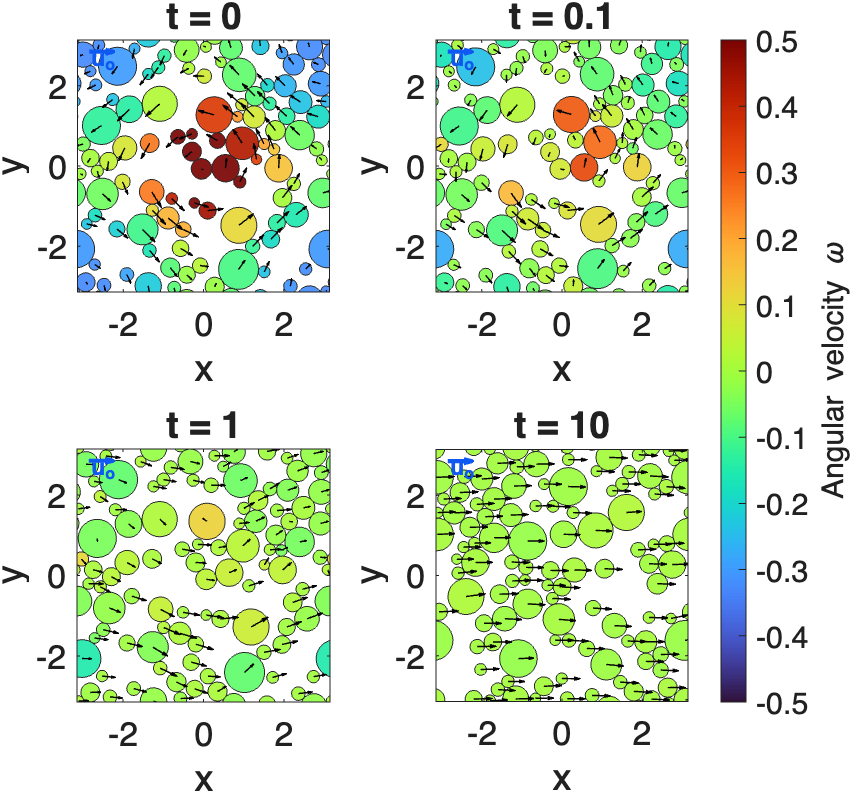} 
    \caption{Example 1 floe trajectories at $t=0, 0.1, 1, 10$. Arrows represent floe velocities and colors represent floe angular velocity. Ocean velocity $\bfu=(0.3, 0)^T$.}
    \label{fig:Ex1T}
\end{figure}
We use the forward Euler method for \eqref{eq:dem} with final time $t=10$ and time step size $\Delta t=10^{-3}$, i.e.\ $N_t=10^4$ steps.
During the simulation, we compute the following global observables:
\begin{itemize}
\item Total momentum
\[
M_{1,\bfv}(t):=\sum_{i=1}^n m^i \bfv^i(t).
\]
\item Total angular momentum (orbital $+$ spin)
\[
L_z(t):=\sum_{i=1}^n \big(\bfx^i(t)\times m^i\bfv^i(t)\big)\cdot \hat{\bfs z}
+
\sum_{i=1}^n I^i\omega^i(t).
\]
\item Translational and rotational kinetic energies (KE)
\[
KE_{\mathrm{t}}(t):=\frac12\sum_{i=1}^n m^i|\bfv^i(t)|^2,
\qquad 
KE_{\mathrm{r}}(t):=\frac12\sum_{i=1}^n I^i|\omega^i(t)|^2.
\]
\item Total normal elastic strain energy (stored in the Hertz-type normal law)
\[
U_n(t):=\sum_{1\le i<j\le n}\frac12 \kappa^{ij}_1(t) \big(\delta^{ij}(t)\big)^2 \chi^{ij}(t),
\]
where $\delta^{ij}$, $\kappa_1^{ij}$ and $\chi^{ij}$ are defined in \eqref{eq:demparm} and \eqref{eq:chi}, and the pairwise distance $d^{ij}$ entering $\delta^{ij}$ is computed using the periodic minimum-image convention.
This is an approximation of the equation~\eqref{eq:nse} by handling $\kappa_1^{ij}$ as a constant during the collision time.
\end{itemize}
We additionally track the mean velocity $\langle \bfv\rangle(t):=\frac1n\sum_{i=1}^n \bfv^i(t)$ and the mean mismatch $\frac1n\sum_{i=1}^n|\bfv^i(t)-\bfu^i|$ to quantify convergence to the ocean drift. Similarly, we use $\langle \omega\rangle(t)$ to denote the mean angular velocity.

\begin{figure}[ht]
    \centering \includegraphics[width=14cm]{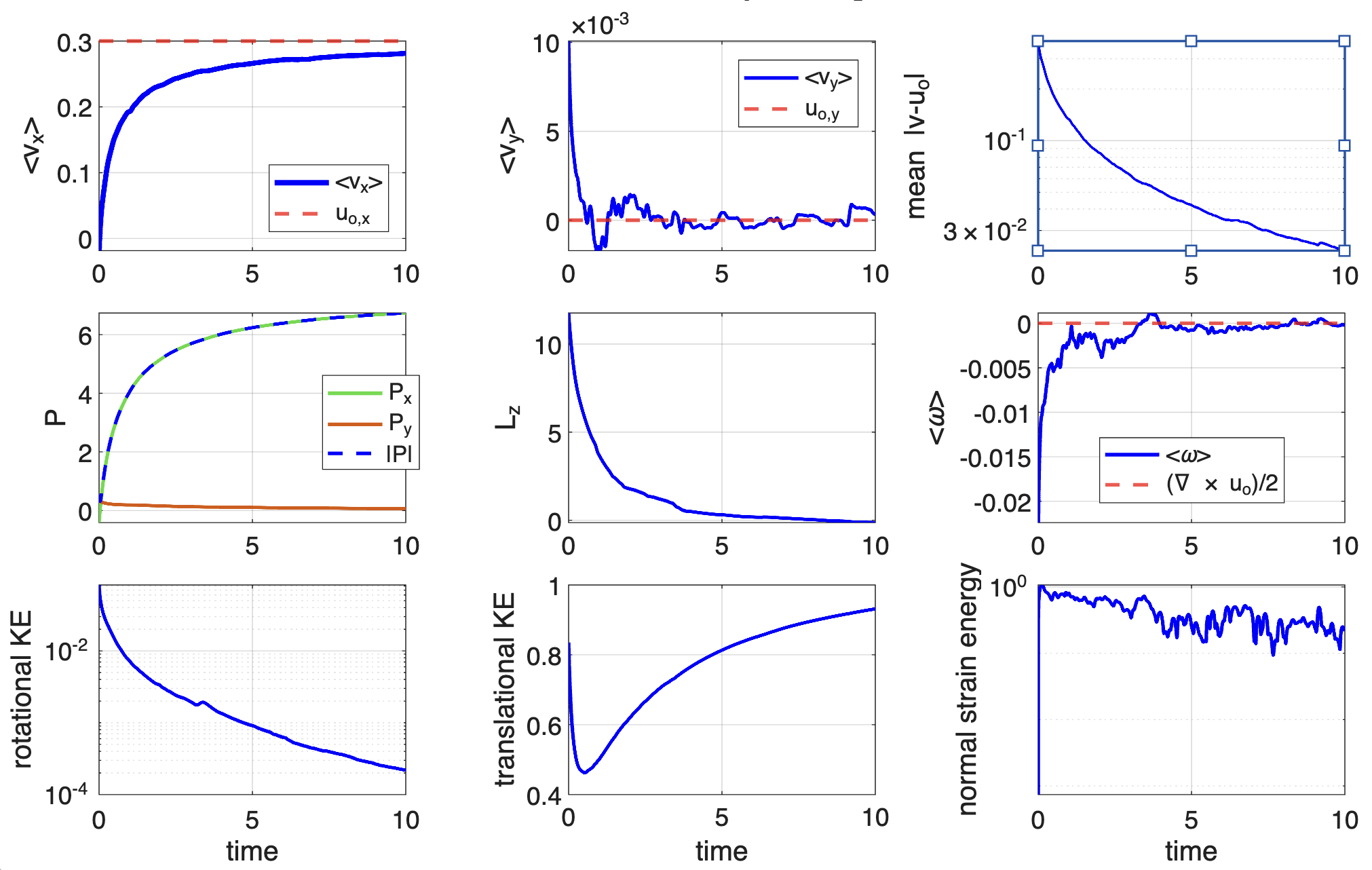} 
    \caption{Example 1 floe moments. Time evolution of mean velocities, velocity differences, and angular velocities, total momentum, total angular momentum, and kinetic and normal contact strain energies. }
    \label{fig:Ex1Ms}
\end{figure}

In Figure \ref{fig:Ex1T}, we can see the floe trajectories at $t=0.0, 0.1, 1, 10.$ In each subplot, the black arrows represent floe velocities $\bfv^i$ and the floe color represents $\omega^i$ (used as a vorticity) with a fixed colorbar range $[-0.5,0.5]$. 
In Figure \ref{fig:Ex1Ms}, we illustrate the trajectories of the floe velocities, angular velocities, and energies. 
The numerical results are consistent with the long-term behavior predicted by the theory for the constant-ocean case in Theorem \ref{thm:v0}. In particular, we observe the following:
\begin{itemize}
\item \textit{Convergence of translational velocities to the ocean velocity.}
We observe that the mean mismatch
$\frac1n\sum_{i=1}^n|\bfv^i(t)-\bfu^i|$ decays in time, and the mean (over all floes) velocity components $\langle v_x\rangle(t)$ and $\langle v_y\rangle(t)$ approach $\bfuu_{o,x}=0.3$ and $\bfuu_{o,y}=0$. 
This is also reflected in the translational kinetic energy $KE_{\mathrm{t}}(t)$, which approaches the kinetic energy of the drift state in which all floes move with $\bfu$, i.e.,
\[
KE_{\mathrm{t}}(t)\ \longrightarrow\ \frac12\sum_{i=1}^n m^i|\bfu|^2
\qquad (t\to\infty),
\]
up to small fluctuations due to intermittent collisions.   In other words, $KE_{\mathrm{t}}(t)$ indicates that $\bfv^i(t)\to\bfu^i$ in accordance with Theorem \ref{thm:v0}.
We also observe that the collisions are less frequent, as all floe velocities align with the ocean velocity.
\item \textit{Decay of angular velocity to zero.}
Since $\nabla\times\bfu^i\equiv 0$, the quadratic rotational drag in \eqref{dem_w},
\[
\beta^i\Big(\nabla\times \bfu^i/2-\omega^i\Big)\Big|\nabla\times \bfu^i/2-\omega^i\Big|
= - \beta^i \omega^i|\omega^i|,
\]
drives $\omega^i(t)\to 0$. Numerically, the mean angular velocity $\langle\omega\rangle(t)$ decays toward zero, and the rotational kinetic energy $KE_{\mathrm{r}}(t)$ decays toward zero as well. This matches the theoretical prediction that, for a constant irrotational ocean velocity, the long-time equilibrium satisfies $\omega^i\equiv 0$.

\item \textit{Momentum balance with ocean forcing.}
The contact forces lead to zero net momentum, i.e.,
\[
\sum_{i=1}^n\frac{1}{n}\sum_{j=1}^n\big(\bff_{\bfn}^{ij}+\bff_{\bft}^{ij}\big)=\bf0,
\]
up to numerical error. Therefore, the evolution of $M_{1,\bfv}(t)$ is determined solely by the ocean drag (see Lemma \ref{lem:sumparticleMV})
\[
\frac{d}{dt}M_{1,\bfv}(t)=\sum_{i=1}^n \alpha^i\big(\bfu^i-\bfv^i\big)\big|\bfu^i-\bfv^i\big|.
\]
Thus $M_{1,\bfv}(t)$ is not conserved in general; instead it drives toward the ocean-drift momentum
\[
P := M_{1,\bfv}^\infty = \sum_{i=1}^n m^i\bfu^i,
\]
which is consistent with the observed convergence $\bfv^i\to\bfu^i$. In particular, after a certain time (roughly when $t>5$), the total momentum becomes approximately constant because $\bfu^i-\bfv^i\approx \bf0$ and the drag contribution becomes negligible.

\item \textit{Normal strain energy is collision-localized.}
The total normal strain energy $U_n(t)$ is nonzero only during contact events (when $\chi^{ij}=1$ and $\delta^{ij}<0$). Numerically, it appears as intermittent bursts (see the last plot in Figure \ref{fig:Ex1Ms}) corresponding to collisions, and decays as the system becomes dilute and relative motion decreases. This decay depends on floe concentration and drag coefficient. This behavior is consistent with the short-range Hertz-type contact interaction.
\end{itemize}

\subsection{Example 2: Consistency test between the particle and hydrodynamic models}\label{subsec:Ex2}

This experiment is designed to test the consistency between the particle model
\eqref{eq:dem} and the hydrodynamic (continuum) model \eqref{eq:LI2H},
thereby validating the particle--kinetic--hydrodynamic hierarchy developed in
this work. Unless otherwise stated, all parameters and numerical conventions
are the same as in Example~1 in Subsection \ref{subsec:EX1}.
We set the ocean velocity to be spatially varying:
\[
\bfu(x,y) = (-y s, x s)^T,\qquad
s=\frac{x^2+y^2-4}{32}
\exp\left(-\frac{(x^2+y^2)(x^2+y^2-4)}{8}\right).
\]
This choice yields a smooth rotational flow with nontrivial spatial structure
and a nonzero vorticity field $\nabla\times\bfu$ (see Figure \ref{fig:ex2f}). In addition, we prescribe a
constant wind velocity
\[
\bfw(x,y)\equiv (10,0)^T.
\]
\begin{figure}[ht]
    \centering \includegraphics[width=14cm]{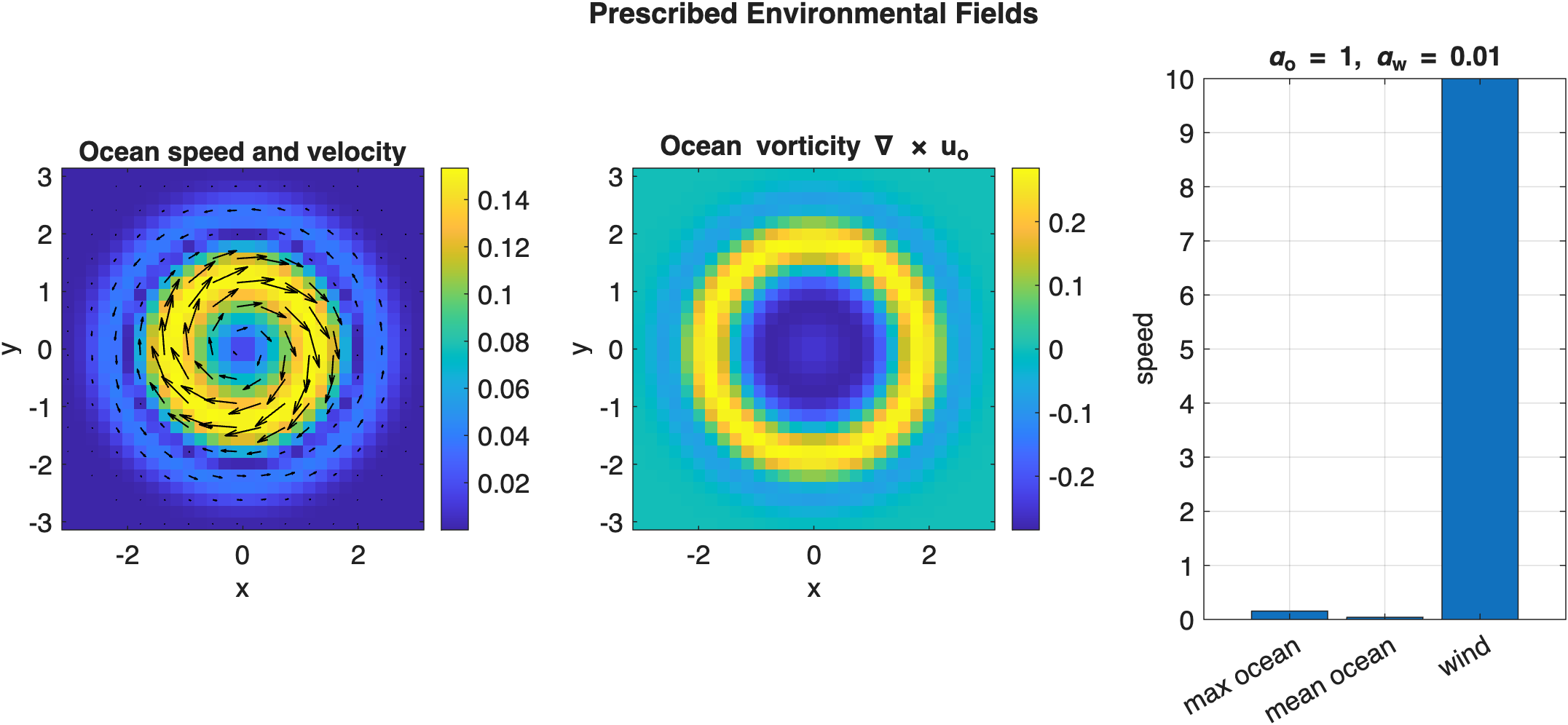} 
    \caption{Oceanic and atmospheric forcing for example 2. }
    \label{fig:ex2f}
\end{figure}
The Coriolis parameter is computed at latitude $\phi=-60^\circ$, i.e.,
$f_E=2\Omega\sin(-60^\circ)$.
The translational and rotational drag coefficients are fixed as
\[
\alpha_o=\beta_o=1,\qquad \alpha_w=\beta_w=0.01.
\]
With the above wind field, the characteristic wind-drag scale is
$\alpha_w|\bfw|^2=1$, while the maximum ocean-drag scale in this experiment is
approximately $\alpha_o\max_{\Omega}|\bfu|^2=2.35\times 10^{-2}$. Thus the
wind forcing provides the dominant large-scale translational drift (cf., \cite{untersteiner2013geophysics,spreen2011trends}), while the
spatially varying ocean velocity and vorticity generate the nonuniform
structure in the velocity and angular-velocity fields.

The kinetic and hydrodynamic models are developed based on the assumption that
the number of particles goes to infinity. For the numerical consistency test,
we initialize $n=14400$ floes. The domain $\Omega=[-\pi,\pi]^2$ is first
partitioned into a uniform $30\times 30$ grid. Each comparison cell
contains $4\times 4=16$ floes, arranged uniformly; equivalently, the particle
centers form a $120\times 120$ uniform lattice. All floes share the same radius and thickness, respectively, 
\[
r^i\equiv 0.02,\qquad h^i\equiv 1,
\]
so that
\[
m^i=\rho_{\mathrm{ice}}\pi (r^i)^2 h^i
=\pi(0.02)^2,\qquad I^i=m^i(r^i)^2.
\]
The initial translational and angular velocities are set to zero, such that
\[
\bfv^i(0)=\bfzr,\qquad \omega^i(0)=0,
\qquad i=1,2,\ldots,14400.
\]
The particle dynamics are evolved by the forward Euler discretization of
\eqref{eq:dem}, where the pairwise contact forces are computed using the same
Hertz-type normal law and Coulomb-capped tangential law as in
\eqref{eq:demparm}--\eqref{eq:zeta}, together with the same doubly periodic
minimum-image convention for contact detection.

\paragraph{Hydrodynamic model configuration and discretization.}
We solve the continuum system \eqref{eq:LI2H} for
$(\rho,\bfuu,\hat\omega)$ in the same domain $[-\pi, \pi]^2$ using uniform $30\times 30$ grid. 
In the numerical comparison, $\rho$ denotes the
cell-integrated floe mass, while $\bfuu$ denotes the corresponding
mass-averaged floe velocity in each cell. 
For the particle model, we collect these quantities as averages over the floes in each cell of the $30\times 30$ grid.
Thus the discrete variables satisfy the cellwise relation
\[
(\rho\bfuu)_K = \rho_K \bfuu_K,
\]
where $\rho_K$ is the total floe mass in cell $K$. The angular velocity
$\bar\omega$ is obtained from
\[
\hat\omega = \bfx\times(\rho\bfuu)+\rho_I\bar\omega,
\]
consistent with the definition in \eqref{eq:hlocmom}. The continuum equations
evolve with the same forward Euler time step size $\Delta t=10^{-3}$ with final time $t=10.$
With the above particle initialization, the continuum initial conditions are
prescribed cellwise as
\[
\rho_K(0)=16 \pi(0.02)^2,\qquad
\bfuu_K(0)=\bfzr,\qquad
\bar\omega_K(0)=0.
\]
The continuum drag coefficients are chosen to match the particle-level
acceleration coefficients, so that
\[
\bar\alpha_o=\alpha_o=1,\qquad
\bar\beta_o=\beta_o=1,\qquad
\bar\alpha_w=\alpha_w=0.01,\qquad
\bar\beta_w=\beta_w=0.01.
\]

\paragraph{Particle-to-continuum observables and comparison.}
For a quantitative comparison, we coarse-grain the particle data onto the same
$30\times 30$ grid used by the continuum solver. For each coarse cell $K$, we
compute the cell-integrated particle mass, the mass-averaged velocity, and the
inertia-averaged angular velocity:
\[
\rho_{\mathrm{p}}(K,t)
:=\sum_{i: \bfx^i(t)\in K} m^i,
\]
\[
\bfuu_{\mathrm{p}}(K,t)
:=
\frac{\displaystyle\sum_{i: \bfx^i(t)\in K}m^i\bfv^i(t)}
{\displaystyle\sum_{i: \bfx^i(t)\in K}m^i},
\qquad
\bar\omega_{\mathrm{p}}(K,t)
:=
\frac{\displaystyle\sum_{i: \bfx^i(t)\in K}I^i\omega^i(t)}
{\displaystyle\sum_{i: \bfx^i(t)\in K}I^i}.
\]

We compare these coarse-grained particle fields with the continuum fields
$(\rho,\bfuu,\bar\omega)$ at times $t=1$ and $t=10$ using side-by-side plots
and cellwise differences. The corresponding discrete $L^2$-errors are reported
in Table~\ref{tab:ex2_errors}.

\begin{table}[ht]
\centering
\begin{tabular}{c|ccc}
\hline
Time & $\|\rho_{\mathrm{p}}-\rho\|_{L^2}$
& $\|\bfuu_{\mathrm{p}}-\bfuu\|_{L^2}$
& $\|\bar\omega_{\mathrm{p}}-\bar\omega\|_{L^2}$ \\
\hline
$t=1$  & $2.80219\times 10^{-3}$ & $9.26343\times 10^{-3}$ & $6.91040\times 10^{-3}$ \\
$t=10$ & $7.89747\times 10^{-3}$ & $2.19277\times 10^{-2}$ & $1.29297\times 10^{-2}$ \\
\hline
\end{tabular}
\caption{Discrete $L^2$ errors between the coarse-grained particle fields and
the continuum fields for the updated $14400$-floe simulation.}
\label{tab:ex2_errors}
\end{table}

\begin{figure}[ht]
    \centering
    \includegraphics[width=14cm]{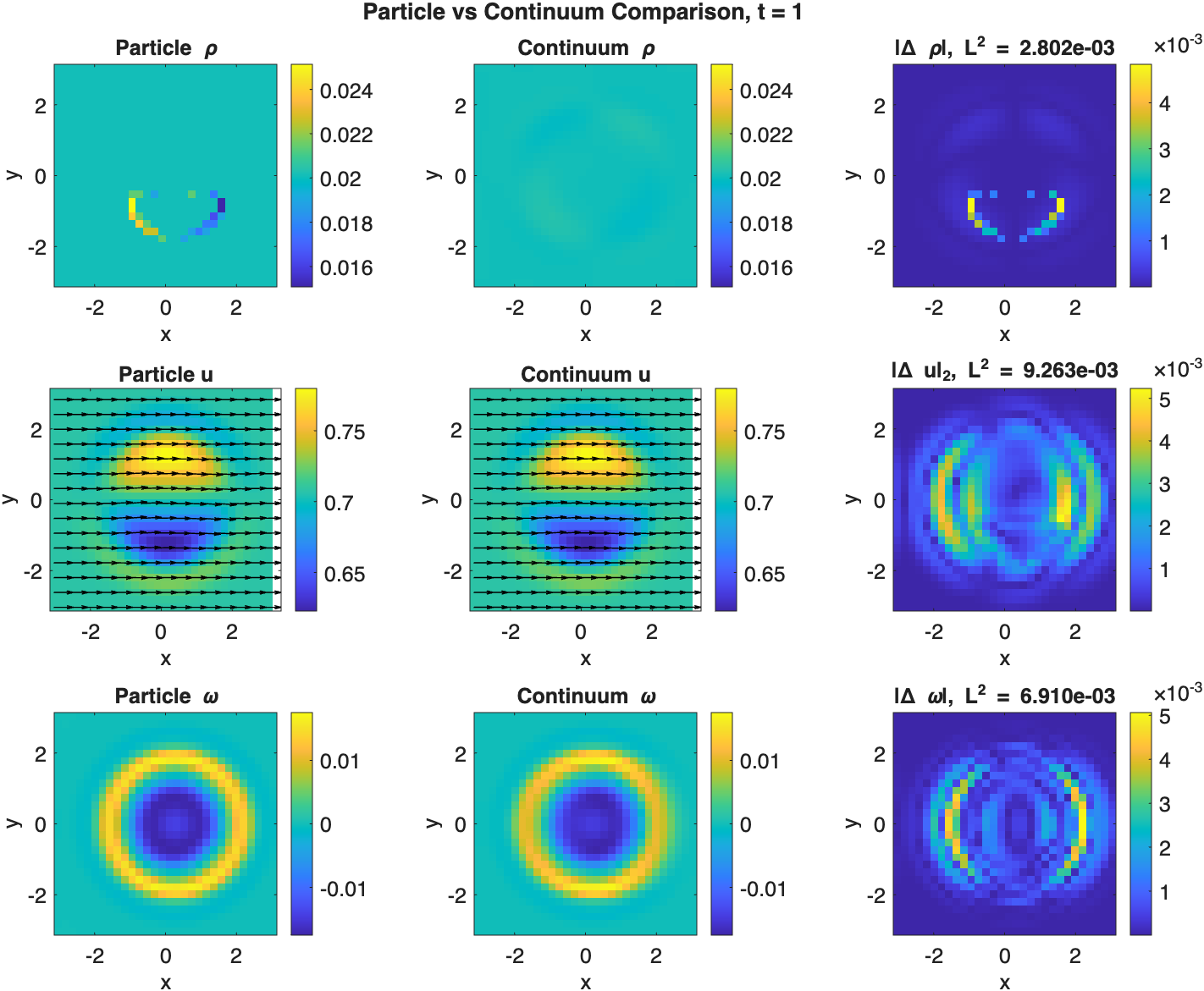}
    \caption{Model comparison at $t=1$ for the updated simulation with
    $14400$ floes, a $30\times 30$ comparison grid, wind velocity
    $\bfw=(10,0)^T$, and drag coefficients
    $\alpha_o=\beta_o=1$, $\alpha_w=\beta_w=0.01$. Here $\rho$ denotes
    cell-integrated floe mass.}
    \label{fig:ex2T1}
\end{figure}

\begin{figure}[ht]
    \centering
    \includegraphics[width=15cm]{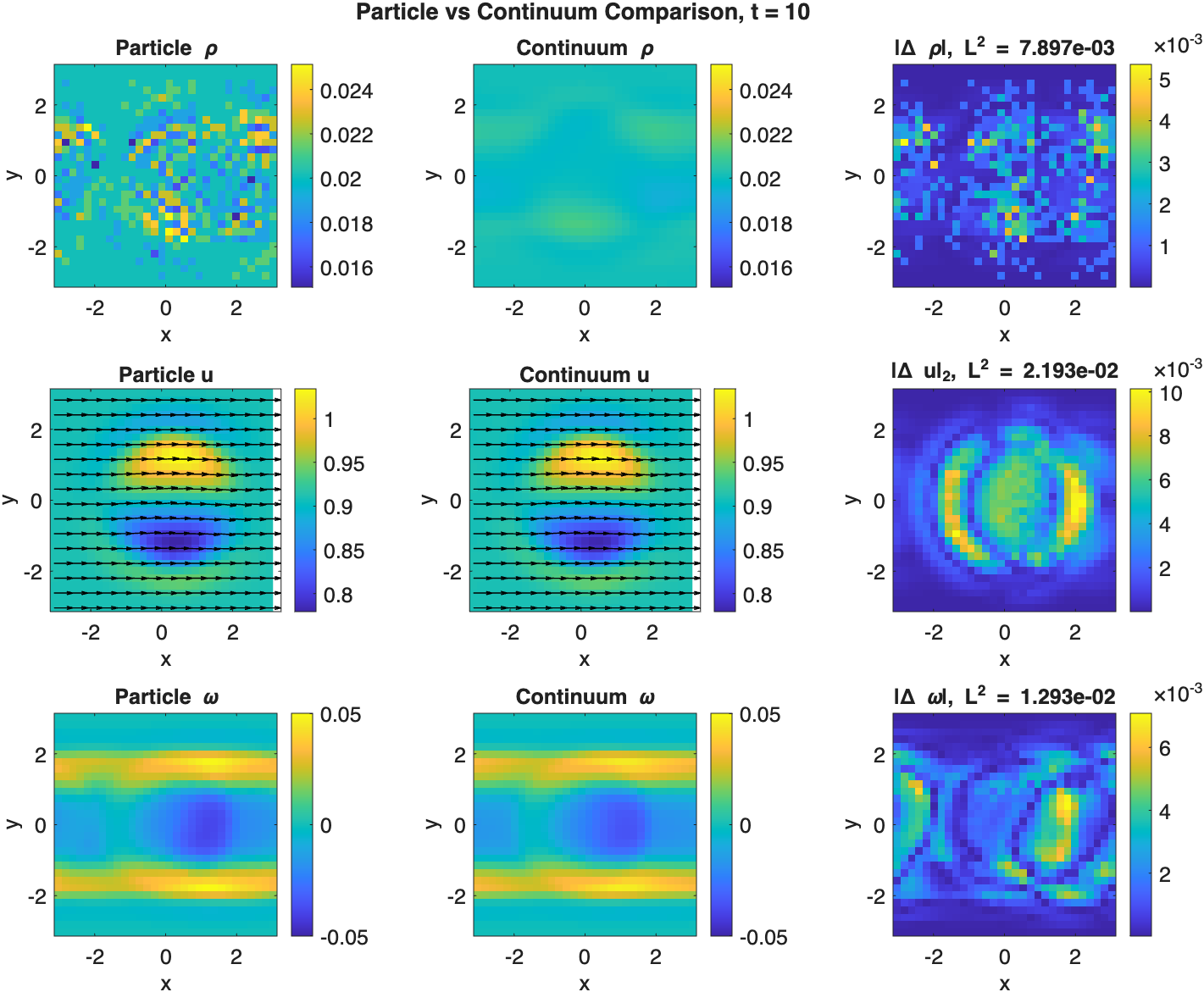}
    \caption{Model comparison at $t=10$ for the updated simulation with
    $14400$ floes, a $30\times 30$ comparison grid, wind velocity
    $\bfw=(10,0)^T$, and drag coefficients
    $\alpha_o=\beta_o=1$, $\alpha_w=\beta_w=0.01$. Here $\rho$ denotes
    cell-integrated floe mass.}
    \label{fig:ex2T10}
\end{figure}

\paragraph{Results and discussion.}
The continuum solution and the coarse-grained particle statistics match well at
the reported times, providing direct numerical evidence for the consistency
between the particle and hydrodynamic descriptions; see
Figure~\ref{fig:ex2T1}, Figure~\ref{fig:ex2T10}, and
Table~\ref{tab:ex2_errors}. In particular, we have the following observations.

\begin{itemize}
\item \textit{Cell-mass consistency.}
For comparison, $\rho$ is interpreted as the total floe mass in each
grid cell. Since each cell initially contains 16 identical floes, the initial
cell mass is
\[
\rho_K(0)=16\pi(0.02)^2 \approx 2.01062\times 10^{-2}.
\]
The particle and continuum total masses remain conserved throughout the
simulation.
The cellwise particle mass $\rho_{\mathrm{p}}(K,t)$ agrees closely with the
continuum cell mass $\rho_K(t)$. The remaining cell-to-cell fluctuations in
$\rho_{\mathrm{p}}$ are caused by finite particle binning. Even though the
initial placement is uniform, advection of discrete floes across cell
boundaries produces small sampling jumps that are absent from the continuum field. The quantified errors are reported in Table \ref{tab:ex2_errors}.

\item \textit{Velocity consistency under combined ocean and wind drag.}
Starting from $\bfv^i(0)=\bf0$ and $\bfuu_K(0)=\bf0$, both models accelerate
under the same environmental forces, including the combined ocean and wind drag mechanisms
\[
\alpha_o(\bfu^i-\bfv^i)|\bfu^i-\bfv^i|
+\alpha_w(\bfw^i-\bfv^i)|\bfw^i-\bfv^i|
\]
in the particle model and
\[
\bar\alpha_o(\bfu-\bfuu)|\bfu-\bfuu|
+\bar\alpha_w(\bfw-\bfuu)|\bfw-\bfuu|
\]
in the continuum model. Because $\bfw=(10,0)^T$ and $\alpha_w=0.01$, the wind
drag provides a strong rightward drift, while the spatially varying ocean flow
adds the rotational perturbation visible in the velocity field. The velocity
patterns from the coarse-grained particles and the continuum model agree well.
The discrete $L^2$ velocity errors are $9.26343\times 10^{-3}$ at $t=1$ and
$2.19277\times 10^{-2}$ at $t=10$.

\item \textit{Angular velocity consistency.}
The angular velocity moments also exhibit strong agreement. Since the wind
field is constant, $\nabla\times\bfw=0$, and the wind rotational drag acts as a
damping mechanism for $\omega$. The nontrivial angular structure is primarily
induced by the ocean vorticity through the relaxation toward
$\frac12\nabla\times\bfu$. The particle angular velocities and the continuum
field $\bar\omega$ capture the same sign structure and magnitude distribution.
The corresponding discrete $L^2$ errors are $6.91040\times 10^{-3}$ at $t=1$
and $1.29297\times 10^{-2}$ at $t=10$.

\item \textit{Overall consistency of the hierarchy.}
The agreement of $(\rho,\bfuu,\bar\omega)$ between the
continuum solver and the coarse-grained particle statistics supports the
consistency of the particle-to-continuum closure in the monokinetic regime.
The comparison validates the environmental forcing and moment closure components of the hierarchy. As the
number of floes per continuum cell increases and the time step and grid spacing
are refined, one expects the coarse-grained particle statistics to become
smoother and to agree more closely with the continuum solution.
\end{itemize}

\begin{figure}[ht]\centering
\includegraphics[width=0.98\linewidth]{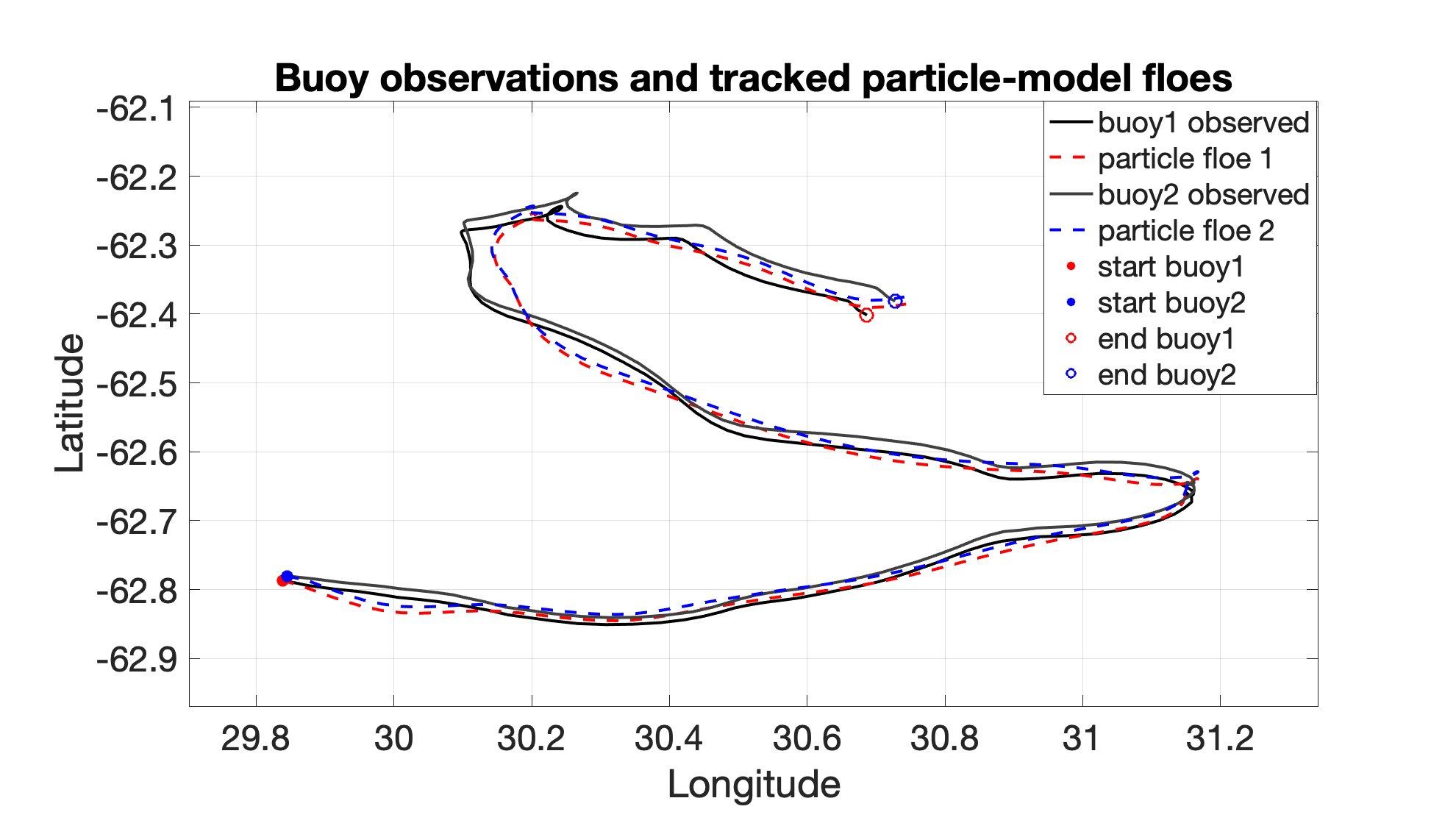}
\caption{Comparison of the simulated floe and true buoy trajectories.}
\label{fig:ex3t}
\end{figure}

\begin{figure}[ht]\centering
\includegraphics[width=0.98\linewidth]{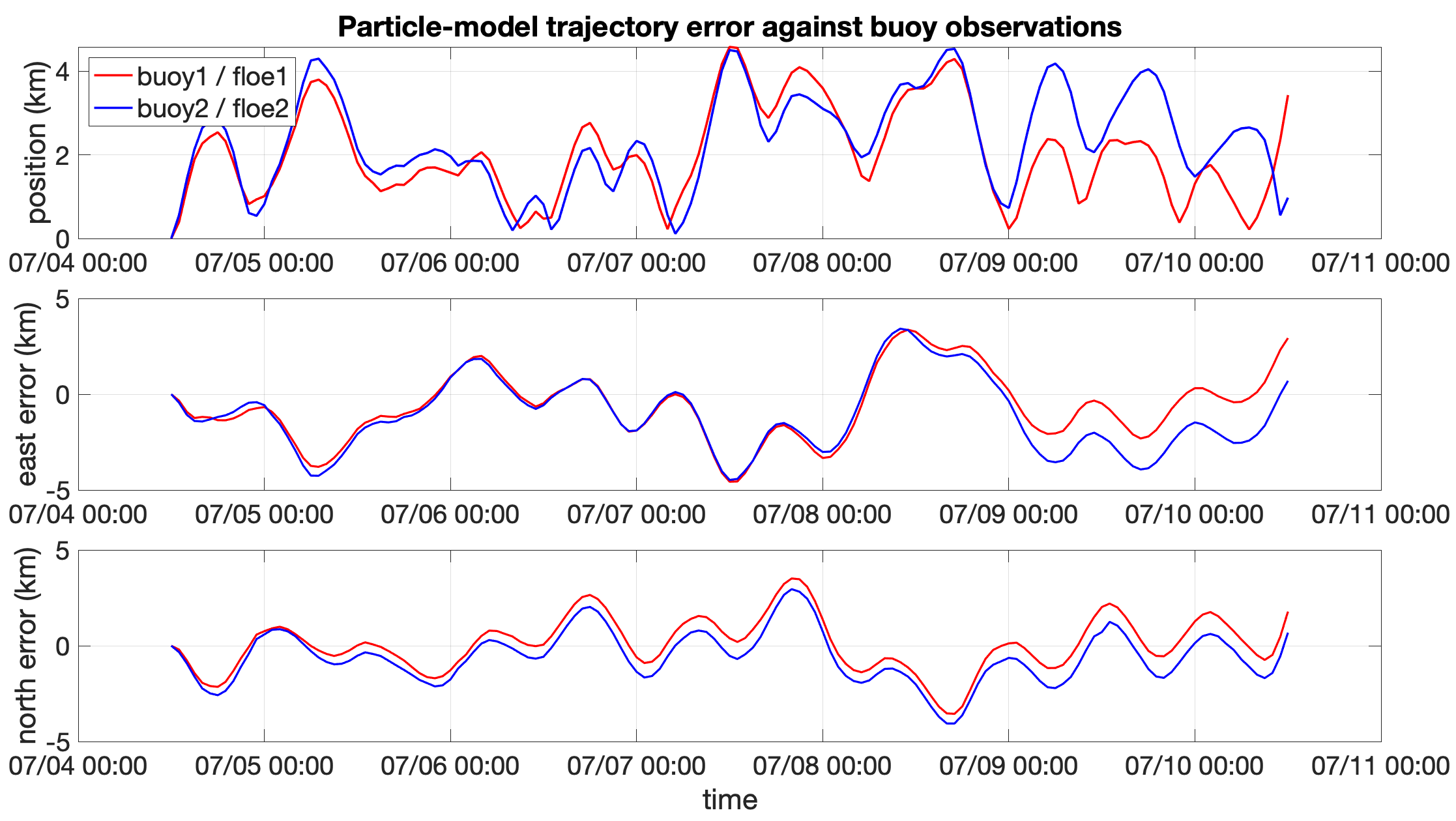}
\caption{Time evolution of errors of the simulated floe and true buoy trajectories.}
\label{fig:ex3e}
\end{figure}

\begin{figure}[ht]\centering
\includegraphics[width=0.98\linewidth]{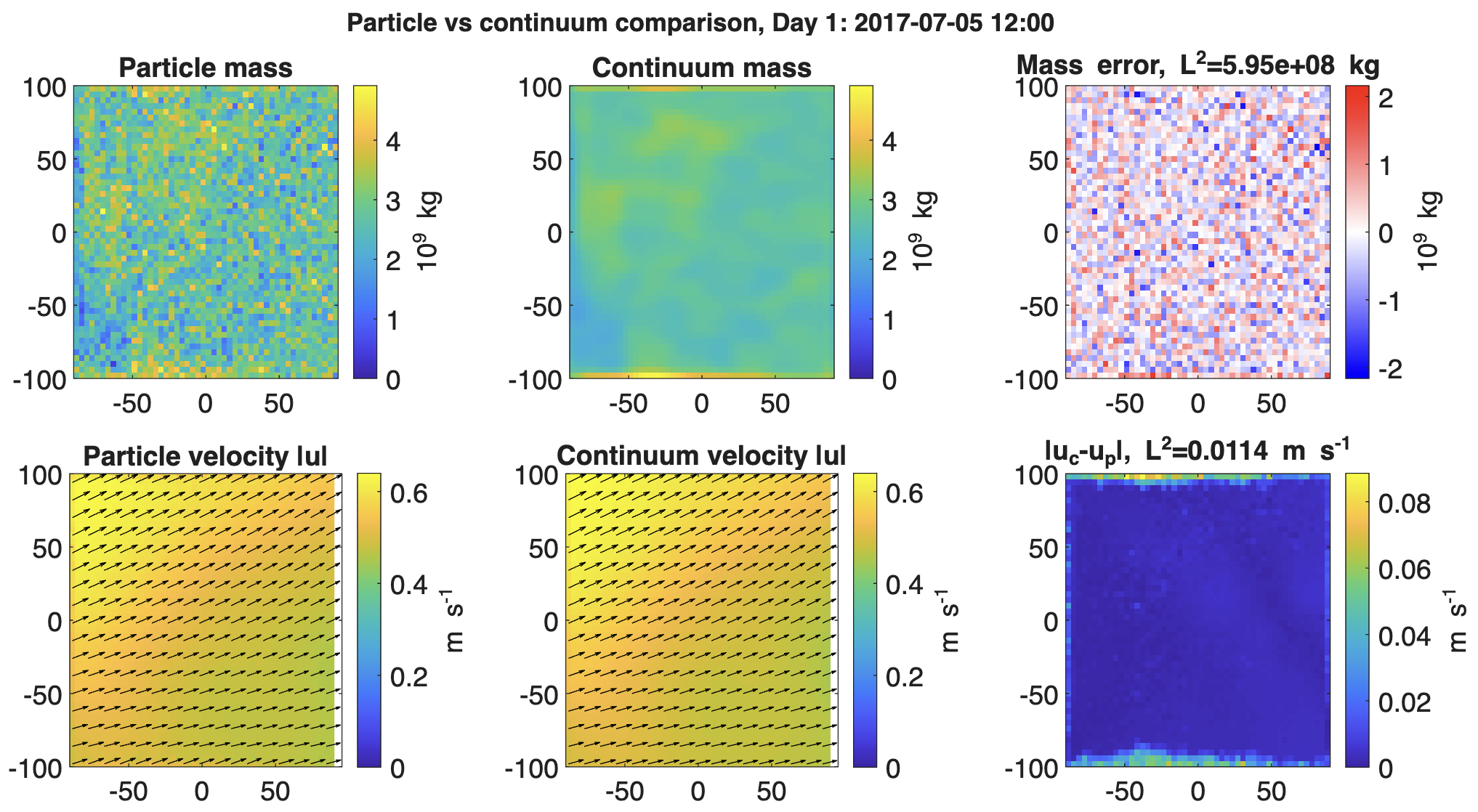}
\caption{Particle--continuum comparison of cell mass and linear velocity at Day 1, 2017-07-05 12:00.}
\label{fig:ex3d1}
\end{figure}

\begin{figure}[ht]\centering
\includegraphics[width=0.98\linewidth]{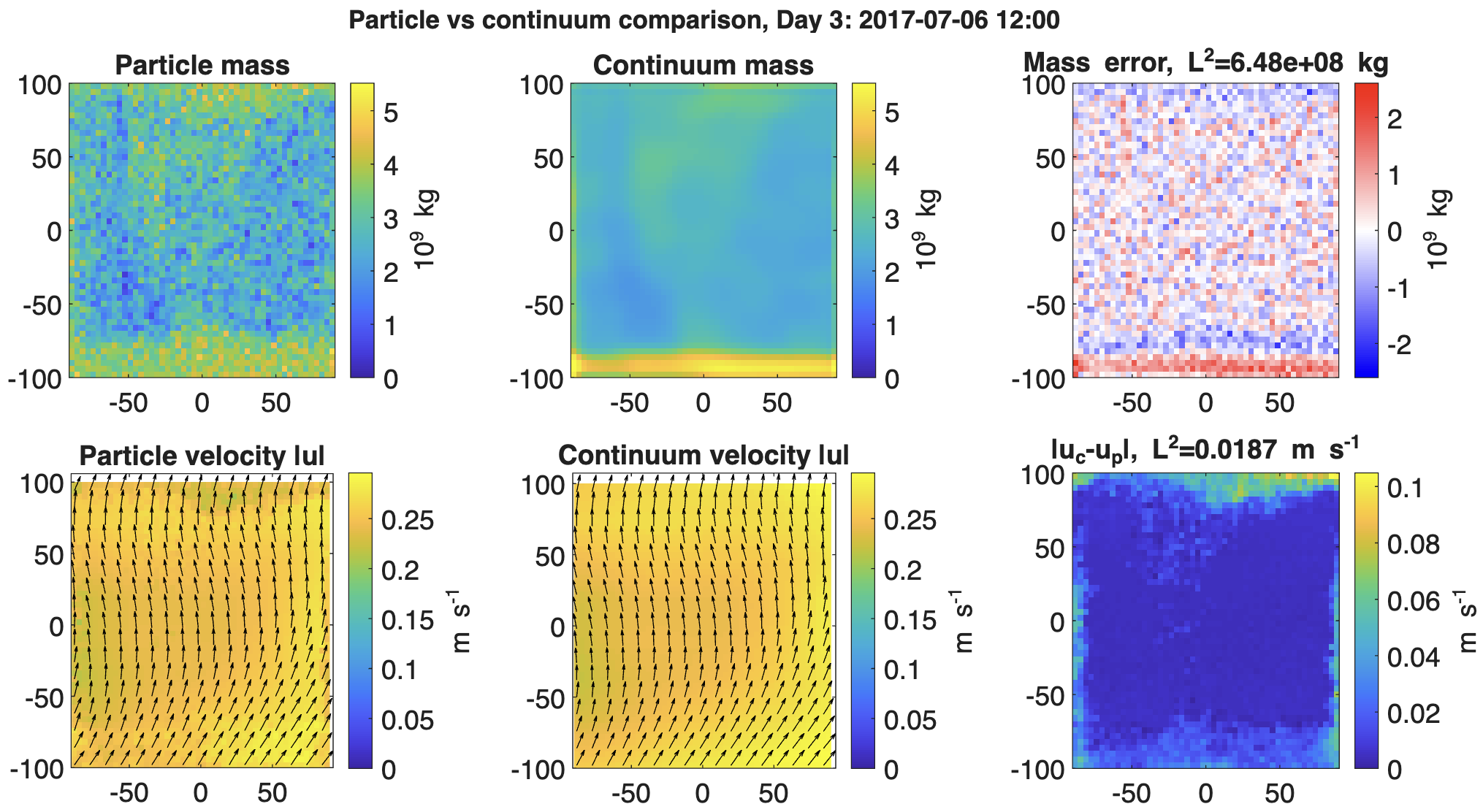}
\caption{Particle--continuum comparison of cell mass and linear velocity at Day 3, 2017-07-06 12:00.}
\label{fig:ex3d3}
\end{figure}

\subsection{Example 3: Comparison with field observations}\label{sec:ex3}

This example tests the particle and hydrodynamic models in a winter Antarctic MIZ setting.  The field data were collected during a voyage of the icebreaker \emph{SA Agulhas II} \cite{de2018sea}, which departed Cape Town on 28 June 2017 and reached the MIZ on 4 July 2017 near $62.5^\circ$S, $30^\circ$E, while a polar cyclone was crossing the ice edge \cite{alberello2020drift}.  Two buoys started recording at 12:00 on 4 July 2017 near $62.8^\circ$S, $29.8^\circ$E, about 100 km from the ice edge and about 1 km apart.  The water depth was about 5000 m.  The buoys recorded position and wave spectral data, but did not carry meteorological instruments.  Buoy B1 sampled every 15 min for 8 days and 18 hr, until 06:00 on 13 July 2017.  Buoy B2 sampled every 15 min for the first 6 days and then every 2 hr, giving a record of about 3 weeks.  Here we use the common early interval needed for the two-buoy drift comparison.  The relative position precision was about 1 m, corresponding to a velocity uncertainty of approximately $10^{-3}$ m s$^{-1}$, which is small compared with the observed drift speeds.


The wind forcing is taken from ERA5 10m wind fields \cite{hersbach2020era5}, provided hourly on a $0.25^\circ$ grid.  The ocean-current forcing is taken from the NASA ECCO ocean velocity daily mean product at $0.5^\circ$ resolution \cite{fenty2020ecco}.  We use the zonal and meridional velocity components from this product as the large-scale ocean drift forcing, interpolated to the floe positions and continuum grid.  
For the simulation, we consider the domain bounded by longitude East $[28.9,32.4]$ and latitude South $[-63.3,-61.5]$, corresponding to a rectangular domain of size 180 km by 200 km.  This domain contains both buoy trajectories and the available atmospheric and ocean forcing data.  The simulation interval is 2017-07-04 12:00--2017-07-10 12:00, covering the common high-frequency portion of the two buoy records used in the comparison.

The field data show that the floes are of diameters around 3.2 m \cite{alberello2019brief,alberello2022three}, which leads to $\ge 10^9$ floes in the given simulation region. For computational cost reduction and simplicity, we adopt the superfloe parameterization \cite{chen2022superfloe} to initialize the floe field with a concentration of 0.6, diameters from 500 m to 2 km, and a power law distribution with exponent 2. This leads to $N=27570$ floes in the region with initial floe concentration 0.597. The two buoy-tracked particles are initialized at the observed buoy locations.  To avoid an initial overlap between the two tracked floes, their numerical diameter is capped at 400 m, based on the observed initial buoy separation 788.76 m. The remaining floes are initialized randomly in the domain for the particle simulation.  The particle equations include wind drag, ocean drag, Coriolis force, surface-slope forcing, and floe--floe collisions through the normal contact and Coulomb-capped tangential contact laws used in the preceding examples.  The drag coefficients are set as $\alpha_w=1.34217\times10^{-6}$ and $\alpha_o=1.12059\times10^{-3}$.

The continuum model is solved on a $50\times50$ grid over the domain.  
The grid is used to coarse-grain the particle data for the particle--continuum comparison.  Both models are initialized with the buoy1 translational velocity.  The initial angular velocity is computed from the ERA5 wind-vorticity field and normalized so that $\max |\omega|\le 10^{-3}$ rad s$^{-1}$.  For each grid cell, the cell mass is the total mass of the floes in that cell.  The particle velocity is the mass-weighted cell average.

Figure~\ref{fig:ex3t} compares the observed buoy trajectories with the simulated floe trajectories.  The model follows the observed large-scale drift well for both buoys over the six-day simulation window.  The remaining errors are modest relative to the approximately 100 km drift scale and to the severe wind, wave, and ocean forcing during the storm-influenced MIZ.  The mean position errors are 2.16 km for buoy1 and 2.50 km for buoy2 (see Figure ~\ref{fig:ex3e}). The final position errors are 3.12 km and 0.70 km. The maximum errors are 4.94 km and 4.80 km, respectively.  
Figure~\ref{fig:ex3e} shows that the east--west, north--south, and total error curves remain bounded and do not exhibit systematic growth over the full interval.  This indicates that the particle model captures the dominant dynamics of the observed floes.

Figures~\ref{fig:ex3d1}--\ref{fig:ex3d3} compare the coarse-grained particle fields with the continuum solution on the $50\times50$ grid. The particle mass field contains visible cell-to-cell fluctuations, as expected from binning a finite number of discrete floes, while the continuum mass field is smoother.  Nevertheless, the two models produce the same large-scale redistribution pattern.  The velocity fields agree well. The particle and continuum velocities have consistent direction and magnitude over most of the domain.
At 2017-07-05 12:00, the relative cell-mass $L^2$-error is 0.2057 and the velocity $L^2$-error is 0.0115 m s$^{-1}$.  At 2017-07-06 12:00, the corresponding errors are 0.2262 and 0.0185 m s$^{-1}$.  At 2017-07-10 12:00, the errors are 0.2388 and 0.0406 m s$^{-1}$. These values are consistent with the visual agreement in the plotted fields.

In summary,
Examples~1--3 provide complementary validation of the particle dynamics and the particle--kinetic--hydrodynamic hierarchy developed in this work. In Example~1, we simulate $n=100$ rotating, colliding floes in the periodic domain $\Omega=[-\pi,\pi]^2$ under a constant ocean velocity $\bfu\equiv(0.3,0)^{T}$. The results confirm the theoretical long-time velocities and angular velocity alignment with the ocean velocity. The corresponding kinetic energies exhibit the expected dissipative trends, while momentum and angular momentum satisfy the correct balance laws in the presence of ocean drag. 
In Example~2, we use the same domain, time stepping, and physical parameters, but impose a spatially varying rotational ocean flow and place $n=14400$ identical floes on a $120\times 120$ lattice with zero initial velocities. Solving the continuum system \eqref{eq:LI2H} on a $30\times 30$ grid using a linear finite element mesh and coarse-graining the particle data onto the same grid, we find that the particle and continuum fields match well.
Example 3 is a realistic simulation that shows agreement with the buoy observations and with the continuum model.
Overall, these examples demonstrate that the hydrodynamic closure reproduces the large-scale statistics of the particle model and thereby support the consistency of the proposed multiscale hierarchy.

\section{Concluding remarks} \label{sec:conclusion}
In this paper, we have extended the particle–kinetic–hydrodynamic hierarchy developed in Part~I to a richer and more realistic setting by incorporating rotational degrees of freedom, nonlinear contact interactions, Coriolis force, and ocean tilt among sea-ice floes. 
Starting from a rigid-body particle description with force–torque coupling, we derived an associated kinetic equation on an extended phase space and obtained macroscopic hydrodynamic balance laws for mass, linear momentum, and angular momentum through moment closures. 
We established the total energy dissipation, which captures the physics of energy loss due to the floe-floe collision.
The resulting framework reveals how nonlinear collisions, frictional effects, and rotational dynamics generate additional stress and dissipation mechanisms at the macroscopic level, thereby providing a systematic and physically consistent multiscale description of fragmented sea ice in the marginal ice zone.

Several important directions remain open for future research. One natural extension is to couple the present mechanical framework with a temperature field, allowing floe sizes, masses, and moments of inertia to evolve through melting and refreezing processes. Such a thermo-mechanical coupling would introduce additional transport and source terms at all scales and raise fundamental questions about energy consistency, scale separation, and closure strategies, leading to more realistic sea ice rheology. 
A further important extension is to incorporate wave–ice interactions, including wave scattering and attenuation, which are fundamental to modeling the MIZ. 
This direction is also consistent with the perspective of \cite{squire2022prognosticative}, which emphasizes the need for sea-ice rheologies that better represent the MIZ under diverse stress regimes through DEM particle models with suitably aggregated continuum descriptions.
Future work may also consider out-of-plane floe motions such as rafting or ridging, as well as edge crumbling during collisions, which would introduce additional deformation, fragmentation, and dissipation mechanisms beyond planar rigid-floe dynamics.
Another challenging direction concerns floe fracture and bonding, corresponding to a dynamically varying number of particles due to breakup, aggregation, or refreezing-induced bonding. From a modeling and analytical perspective, this leads to nontrivial difficulties in tracking collective behavior, conservation laws, and statistical descriptions, when the underlying particle number changes in time. Developing a coherent multiscale theory that accommodates variable particle numbers while retaining tractable kinetic and hydrodynamic limits remains an open and promising problem for future study.

\section*{Acknowledgements}
Q.D. is partially supported by the Tsinghua University Dushi program, the start-up funding from the Yau Mathematical Sciences Center, Tsinghua University, and the Australian National Computing Infrastructure (NCI) national facility under grant zv32. The work of S.-Y. Ha is supported by National Research Foundation(NRF) grant funded by the Korea government(MSIT) (RS-2025-00514472). LGB is funded by the Australian Research Council (FT190100404, DP240100325).

\bibliographystyle{siam}
\bibliography{ref}

\vspace{1cm}
\appendix{\noindent \textbf{\Large{Appendix A: Local Homogeneity Assumption}}}
\label{app:locassum}

In the derivation of macroscopic balance laws \eqref{eq:LI2H} from the kinetic description of sea ice, we adopted the local homogeneity assumption to further simplify the integrals on contact operator in \eqref{eq:T5j1} and \eqref{eq:T5j2}. 
This assumption formalizes the notion that, at sufficiently small scales around a given macroscopic position, the ice floe distribution exhibits negligible spatial gradients. Under this somewhat strong assumption, contact-operator related terms vanish in analogy with \cite{bardos1991fluid,golse2003mean,golse2005boltzmann} for gas particles, enabling the simplified hydrodynamic equations \eqref{eq:LI2H}.


\paragraph{(Local homogeneity assumption)}
We say that \(F(\bfx,\bfv,\theta, \omega, r,h)\) is locally homogeneous at \(\bfx\) if, for \(\bfss\) in a neighborhood of zero with \(|\bfss|\ll L\) for some macroscopic scale \(L>0\),
\begin{equation*} \label{eq:lochomo}
F(t,\bfx+\bfss,\bfv,\theta, \omega, r,h)\approx F(t,\bfx,\bfv,\theta, \omega, r,h),
\end{equation*}
and similarly for its moments.
This condition asserts that the distribution varies slowly relative to the characteristic length scales of the collision/contact interactions.

The local homogeneity assumption reflects the physical idea that, at scales comparable with the interaction range of floes (e.g., contact kernel support), the macroscopic fields vary negligibly. In other words, the spacing between colliding or interacting floes is small relative to macroscopic gradients in density and velocity. This justifies treating \(F\) and its moments as essentially constant over the support of the collision/contact operator.



Recall \eqref{eq:T52} that 
\begin{align*}
\T_{52}
& = \int_{\hat{\bbr}_3} m \bff_{c,\bfn}[F](t,\bfz) F(t,\bfy) d\bfy \\
& = \int_{\hat{\bbr}_3}\int_{\hat{\bbr}_4}
m \kappa_1\big(|\bfx^*-\bfx|-(r+r^*)\big) \bfn(\bfx,\bfx^*) 
F(t,\bfz^*) d\bfz^* F(t,\bfy) d\bfy, \nonumber
\end{align*}
where $\bfy=(\bfv,\theta, \omega, r,h)$, $\bfz^*=(\bfx^*,\bfv^*,\theta^*, \omega^*,r^*, h^*)$.
Using the local homogeneity assumption and exchange of variables, 
the internal integration is approximated as
\begin{align*}
m \bff_{c,\bfn}[F](t,\bfz) & =
\int_{\hat{\bbr}_4}
m \kappa_1\big(|\bfx^*-\bfx|  -(r+r^*)\big) \bfn(\bfx,\bfx^*) 
F(t,\bfx^*, \bfy^*) d\bfx^* d\bfy^* \\
& = \int_{\hat{\bbr}_4}
m \kappa_1\big(|\bfx-\bfx^*|  -(r^*+r)\big) \bfn(\bfx^*,\bfx) 
F(t,\bfx, \bfy^*) d\bfx d\bfy^* \\
& = - \int_{\hat{\bbr}_4}
m \kappa_1\big(|\bfx^*-\bfx|  -(r+r^*)\big) \bfn(\bfx,\bfx^*) 
F(t,\bfx, \bfy^*) d\bfx d\bfy^* \\
& = \frac{1}{2}\int_{\hat{\bbr}_4}
m \kappa_1\big(|\bfx^*-\bfx|  -(r+r^*)\big) \bfn(\bfx,\bfx^*) 
\big( F(t,\bfx^*, \bfy^*) d\bfx^* - F(t,\bfx, \bfy^*) d\bfx \big) d\bfy^* \\
& \approx 0,
\end{align*}
where the last approximation is based on the local homogeneity assumption of moments. Thus, $\T_{52}\approx 0.$ 
This may be interpreted as the collisional force being approximately zero (forces surrounding a floe at $\bfx$ cancel) in the distribution sense, as the number of floes goes to infinity, assumed in Section \ref{sec:kmodel} for kinetic description.
Similarly, one can derive $\T_{53} \approx 0$ and using this local assumption and that $(\bfv  - \bfv^*) \cdot \bfn(\bfx, \bfx^*) = (\bfv^*  - \bfv) \cdot \bfn(\bfx^*, \bfx)$ and $\T_{54} \approx 0$ using $\bft(\bfx, \bfx^*) = - \bft(\bfx^*, \bfx)$.
The term $\bfx \times \T_{54} + \T_{56}$ vanishes following a similar derivation in \eqref{eq:Ta}. This leads to the simplified hydrodynamic model \eqref{eq:LI2H} which was studied in the numerical experiments in Section \ref{sec:num}.

On the other hand, the local homogeneity assumption have limitations. 
It is a heuristic approximation that ensures the contact operator conserves mass and momentum as in  \cite{bardos1991fluid,golse2003mean,golse2005boltzmann} for gas particles. 
However, sea ice floe particles, with different Knudsen number \cite{karniadakis2005microflows}, are different from gas particles (typically for dilute regime with small rare collisions, i.e., Boltzmann–Grad limit).  
They are most justified when macroscopic gradients vary on scales much larger than the contact interaction range. 
In regimes with strong shear or boundary effects, deviations from local homogeneity should be considered and this would lead to more realistic sea ice rheology \cite{shen1987role,feltham2008sea,herman2022granular}.

We take the momentum balance law \eqref{eq:MK_mom} as an example to derive the contact stress tensor as in sea ice rheology discussed in \cite{feltham2008sea}.
We recall the balance law for momentum in \eqref{eq:MK_mom}:
\begin{equation*}\label{eq:MK2H_recall}
\partial_t(\rho \bfuu)+\nabla_{\bfx}\cdot(\rho \bfuu\otimes \bfuu)
=\int_{\hat{\bbr}_3} m\bff[F]Fd\bfy, \qquad \bff[F]=\bff_o+\bff_c[F],
\end{equation*}
where $\bff_o=\gamma_o(\bfu-\bfv)|\bfu-\bfv|$ and
$
\bff_c[F](t,\bfz)=\int_{\hat{\bbr}_4}\bff_c(\bfz,\bfz^*)F(t,\bfz^*)d\bfz^*
$
with $\bff_c(\bfz,\bfz^*)$ containing the normal, damping, and tangential components defined in \eqref{eq:hnote}.
With mono-kinetic ansatz \eqref{eq:monoans}, we define and calculate the contact contribution
\begin{align*}
    \bfb_c(t,\bfx) & :=\int_{\hat{\bbr}_3} m\bff_c[F](t,\bfx,\bfy)F(t,\bfx,\bfy)d\bfy \\
    & = \int_{\hat{\bbr}_3}\int_{\hat{\bbr}_4}
m(\bfy)\bff_c(\bfz,\bfz^*)
F(t,\bfz)F(t,\bfz^*)d\bfz^*d\bfy \\
& = \rho(t,\bfx) \int_{\hat{\bbr}_3}\int_{\hat{\bbr}_4} 
\Phi \frac{\rho(t,\bfx^*)}{m(r^*,h^*)}\Phi^* \bff_c(\bfz,\bfz^*) d\bfz^*d\bfy,
\end{align*}
where $\Phi = \Phi(t,\bfx,\theta,r,h)$ and $\Phi^* = \Phi(t,\bfx^*,\theta^*,r^*,h^*)$.
To obtain the contact stress tensor, we adopt the Irving--Kirkwood form and introduce the bond-localization kernel
\begin{equation*}
B(\bfx;\bfx,\bfx^*):=\int_0^1 \hat\delta\!\big(\bfx-(1-s)\bfx-s\bfx^*\big)ds,
\end{equation*}
and use the standard Irving--Kirkwood identity
\[
\nabla_{\bfx}\cdot\big((\bfx^*-\bfx)B(\bfx;\bfx,\bfx^*)\big)=\hat\delta(\bfx-\bfx^*)-\hat\delta(\bfx-\bfx).
\]
Then the contact force density admits the stress representation
\begin{equation*}
\bfb_c(t,\bfx)=\nabla_{\bfx}\cdot\sigma_c^{\rm MK}(t,\bfx),
\end{equation*}
where the monokinetic contact stress tensor is
\begin{equation*}
\sigma_c^{\rm MK}(t,\bfx)
=
-\frac12 \int_{\hat{\bbr}_4^2}  (\bfx^*-\bfx) \otimes \big(\rho(t,\bfx)\Phi \big)
\frac{\rho(t,\bfx^*)}{m(r^*,h^*)}\Phi^* \bff_c(\bfz,\bfz^*) B(\bfx;\bfx,\bfx^*) d\bfz^*d\bfy
\end{equation*}
Moreover, one may decompose the terms as 
$
\sigma_c^{\rm MK}=\sigma_{\bfn}^{\rm MK}+\sigma_{\bfv}^{\rm MK}+\sigma_{\bft}^{\rm MK},
$
obtained by replacing $\bff_c$ by $\bff_{\bfn}$, $\bff_{\bfv}$, and $\bff_{\bft}$, respectively.
Using the contact stress tensor, the momentum balance law \eqref{eq:MK_mom} reduces to 
\begin{equation*}
\partial_t(\rho\bfuu)+\nabla_{\bfx}\cdot(\rho\bfuu\otimes\bfuu)
=
\nabla_{\bfx}\cdot\sigma_c^{\rm MK}
+\bar\alpha_o (\bfu-\bfuu)|\bfu-\bfuu|+\bar\alpha_w (\bfw-\bfuu)|\bfw-\bfuu| - \rho f_E \bfz \times \bfuu + \rho \bfuu_g.
\end{equation*}
Herein, the right-hand side consists of a contact stress divergence and an effective environmental-force induced stress as in sea ice rheology discussed in \cite{hibler1979dynamic,shen1987role,feltham2008sea,herman2022granular}.
We remark that the balance law for angular momentum can be developed similarly. Incorporating rotational effects enriches the continuum description of sea ice rheology in the literature; a detailed development will be pursued in future work.

\end{document}